\documentclass[11pt]{amsart}
\usepackage[T1]{fontenc}
\usepackage[utf8]{inputenc}
\usepackage{amsmath,amssymb,amsthm,mathtools}
\usepackage{ebproof}
\usepackage{forest}
\usepackage{bm}
\usepackage{esvect}
\usepackage{enumitem}
\setlist[enumerate,1]{label=\textup{(\roman*)}}
\usepackage{fullpage}
\usepackage{microtype}
\usepackage{mleftright}
\usepackage{subcaption}
\usepackage{float}
\usepackage[pdfusetitle,colorlinks=true,linkcolor=blue,citecolor=blue,urlcolor=blue,pagebackref]{hyperref}
\usepackage[alphabetic,backrefs,msc-links]{amsrefs}
\usepackage{tikz}
\usetikzlibrary{positioning,calc}

\numberwithin{equation}{section}

\theoremstyle{plain}
\newtheorem{theorem}{Theorem}[section]
\newtheorem{lemma}[theorem]{Lemma}
\newtheorem{proposition}[theorem]{Proposition}
\newtheorem{corollary}[theorem]{Corollary}
\newtheorem{theoremA}{Theorem}

\newtheorem{theoremB}{Theorem}

\theoremstyle{definition}
\newtheorem{definition}[theorem]{Definition}
\newtheorem{remark}[theorem]{Remark}
\newtheorem{example}[theorem]{Example}
\newtheorem{openproblem}[theorem]{Open Problem}

\newcommand{\setmid}{\mathrel{}\middle|\mathrel{}}
\newcommand{\Z}{\mathbb{Z}}
\newcommand{\N}{\mathbb{N}}
\newcommand{\abs}[1]{\mleft|#1\mright|}
\newcommand{\len}[1]{\mleft|#1\mright|}
\newcommand{\card}[1]{\mleft|#1\mright|}
\newcommand{\tta}{\mathtt{a}}
\newcommand{\ttb}{\mathtt{b}}
\newcommand{\ttc}{\mathtt{c}}
\newcommand{\ttd}{\mathtt{d}}
\newcommand{\Sigman}[1]{\Sigma_{#1}}
\newcommand{\psin}[1]{\psi_{#1}}
\newcommand{\MIX}{\mathrm{MIX}}
\newcommand{\MCFL}{\mathrm{MCFL}}
\newcommand{\MCFLwn}{\mathrm{MCFL}_{\mathrm{wn}}}
\newcommand{\mMCFL}[1]{#1\text{-}\MCFL}
\newcommand{\mMCFLwn}[1]{#1\text{-}\MCFLwn}
\newcommand{\CFL}{\mathrm{CFL}}
\newcommand{\TAL}{\mathrm{TAL}}
\newcommand{\IL}{\mathrm{IL}}
\newcommand{\Wpsi}{W_\psi}
\DeclareMathOperator{\Regroup}{Regroup}
\DeclareMathOperator{\Perm}{Perm}
\DeclareMathOperator{\concat}{concat}
\DeclareMathOperator{\proj}{proj}
\DeclareMathOperator{\join}{join}
\DeclareMathOperator{\Split}{split}
\DeclareMathOperator{\tc}{tc}
\DeclareMathOperator{\inst}{inst}
\DeclareMathOperator{\pad}{pad}
\DeclareMathOperator{\Rep}{Rep}
\newcommand{\Minsubseq}{\operatorname{Min}_{\preccurlyeq}}
\newcommand{\Reppsi}{\Rep_\psi}
\newcommand{\separator}{\prime}
\newcommand{\placeholder}{\mathord{\bullet}}
\newcommand{\Gpsi}[1]{G_\psi[#1]}
\newcommand{\Gnbpsi}[1]{G_\psi^{\mathrm{nb}}[#1]}
\newcommand{\Gn}[2]{G_{#1}[#2]}
\newcommand{\Gnbn}[2]{G_{#1}^{\mathrm{nb}}[#2]}
\newcommand{\proofheading}[1]{%
  \par\addvspace{\medskipamount}%
  \noindent\textbf{#1.}\enspace\ignorespaces
}

\title{On the Kanazawa--Salvati Conjecture}
\author{Takao Yuyama}
\address{Faculty of Social Informatics, ZEN University, 3-12-11 Shinjuku, Zushi, Kanagawa 249-0007, Japan}
\email{takao\_yuyama@zen.ac.jp}
\date{}

\begin{document}

\begin{abstract}
  The language \(\mathrm{MIX}\) consists of all words over a three-letter alphabet that have an equal number of occurrences of each letter.
  It is also the word problem of \(\mathbb{Z}^2\) with respect to a suitable choice of generators.
  The Kanazawa--Salvati conjecture states that \(\mathrm{MIX}\) is not a well-nested multiple context-free language.
  Every well-nested multiple context-free language is an indexed language.
  We reduce the conjecture to an explicit combinatorial problem about tuples of words,
  which is easier to state than the original formulation in terms of arbitrary well-nested multiple context-free grammars.
  More generally, for every surjective monoid homomorphism \(\psi\colon \Sigma^* \to \mathbb{Z}^d\),
  we define a family of well-nested multiple context-free grammars \(G_\psi[r]\) for \(r \geq 1\),
  each of which generates a sublanguage of \(\psi^{-1}(\mathbf{0})\).
  We prove that every well-nested multiple context-free sublanguage of \(\psi^{-1}(\mathbf{0})\)
  is contained in \(L(G_\psi[r])\) for some \(r \geq 1\).
  Using this family, we prove that the four-letter analogue \(\mathrm{MIX}_4\),
  which is a word problem of \(\mathbb{Z}^3\),
  is not a well-nested multiple context-free language.
  The proof reduces this claim to a result of Bishop--Elder--Evetts--Gallot--Levine stating
  that the two-letter analogue \(\mathrm{MIX}_2\) is not generated by any non-branching multiple context-free grammar.
  The Kanazawa--Salvati conjecture itself remains open.
\end{abstract}

\maketitle
\vspace*{-1.5\baselineskip}

\begingroup
  \linespread{0.85}\selectfont
  \tableofcontents
\endgroup

\section{Introduction}\label{sec:introduction}

\subsection{Background}
\label{subsec:intro-background}

The Kanazawa--Salvati conjecture concerns the language
\[
  \MIX
  =
  \MIX_3
  =
  \{w \in \{\tta, \ttb, \ttc\}^* \mid \len{w}_{\tta} = \len{w}_{\ttb} = \len{w}_{\ttc}\},
\]
where \(\len{w}_x\) denotes the number of occurrences of the letter \(x\) in \(w\).
Thus \(\MIX\) consists of all words over \(\{\tta, \ttb, \ttc\}\) that have an equal number of occurrences of each letter.
There is no restriction on the order of the letters.
For example, the words \(\tta\ttb\ttc\), \(\ttc\tta\ttb\), \(\ttb\tta\tta\ttc\ttb\ttc\), and the empty word \(\varepsilon\) all belong to \(\MIX\),
whereas \(\ttb\tta\tta\) and \(\tta\tta\ttb\ttc\ttc\) do not.

The language \(\MIX\) has been studied from two perspectives
(see the introductions in
\cites{kanazawa-salvati2012,salvati2015,gebhardt-meunier-salvati2022}).
In mathematical and computational linguistics, \(\MIX\) is regarded as an extreme case of free word order.
In group theory, \(\MIX\) is a word problem of \(\Z^2\).
We describe these motivations below and then formulate the Kanazawa--Salvati conjecture.

\subsubsection{Free word order and mild context-sensitivity}

One problem in mathematical and computational linguistics is to determine
how much generative and computational power is needed to model the structure of natural language.

Bach used \(\MIX\) in an exercise showing that permutation closure does not preserve context-freeness
and discussed the example in connection with scrambling
\cite{bach1981discontinuous}*{Exercise~2 and the answer to Exercise~2}.
Bach also remarked that human languages do not appear to permit such complete freedom of order
and that certain constituents act as boundary domains for scrambling \cite{bach1981discontinuous}*{p.~5}.

In the 1980s, evidence accumulated that context-free grammars do not adequately characterize
some aspects of the structure of natural language
\cite{joshi-vijay-shanker-weir1991}*{Section~1}.
For example, Shieber argued that Swiss German is not weakly context-free by analyzing case-marking requirements in cross-serial constructions \cite{shieber1985contextfreeness}*{Sections~1--3}.
Under the linguistic assumptions of his analysis, the language
\begin{equation}\label{eq:shieber-language}
  L = \{\tta^m \ttb^n \ttc^m \ttd^n \mid m, n \geq 1\}
\end{equation}
can be obtained from Swiss German by taking homomorphic images and intersections with regular languages.
The class of context-free languages is closed under these operations \cite{hopcroft-motwani-ullman2007}*{Theorems~7.23 and 7.27}.
Thus, if Swiss German were context-free, then \(L\) would also be context-free.
However, \(L\) is not context-free \cite{hopcroft-motwani-ullman2007}*{Example~7.20}, yielding a contradiction.

Based on formal properties of tree-adjoining grammars, Joshi proposed that the grammars needed to describe natural languages might be characterized
as \emph{mildly context-sensitive grammars}.
He proposed three properties intended to characterize this class roughly:
limited cross-serial dependencies, constant growth, and polynomial parsing
\cite{joshi1985tag}*{Section~6.3.3}.
These properties provide only a rough characterization, not a precise definition, of mildly context-sensitive grammars.
Joshi, Vijay-Shanker, and Weir later suggested that mildly context-sensitive grammars should perhaps not generate \(\MIX\)
\cite{joshi-vijay-shanker-weir1991}*{Section~1}.

Joshi also described \(\MIX\) as an extreme case of free word order that is linguistically irrelevant
\cite{joshi1985tag}*{Example~6.3.4}.
Gazdar similarly remarked that it seemed unlikely that any natural language would have a \(\MIX\)-like characteristic
\cite{gazdar1988indexed}*{p.~86}.
In the same discussion, Gazdar recorded Marsh's conjecture that \(\MIX\) is not an indexed language
\cite{gazdar1988indexed}*{p.~86}; see \cite{marsh1985conjectures}.

Kanazawa and Salvati interpret the suggestion of Joshi, Vijay-Shanker, and Weir to mean that \(\MIX\) violates the condition of limited cross-serial dependencies
\cite{kanazawa-salvati2012}*{Section~1}.
These linguistic considerations lead to the question of whether particular grammar formalisms can generate \(\MIX\).

\subsubsection{Word problem for groups}

One problem in group theory is to determine which finitely generated groups
have word problems in a given class of formal languages.
A \emph{choice of generators} for a finitely generated group \(G\) is a surjective monoid homomorphism
\(\pi\colon \Sigma^* \to G\), where \(\Sigma\) is a finite alphabet.
The \emph{word problem} of \(G\) with respect to \(\pi\) is the language \(\pi^{-1}(1_G)\),
where \(1_G\) is the neutral element of \(G\).

The language \(\MIX\) is the word problem of \(\Z^2\) with respect to the homomorphism
\(\psin{3}\colon \{\tta, \ttb, \ttc\}^* \to \Z^2\) defined by
\(\psin{3}(\tta) = (1, 0)\), \(\psin{3}(\ttb) = (0, 1)\), and \(\psin{3}(\ttc) = (-1, -1)\).
Let \(\varphi_2\colon \{a_1, \bar{a}_1, a_2, \bar{a}_2\}^* \to \Z^2\) be the choice of generators defined by
\(\varphi_2(a_1) = (1, 0)\), \(\varphi_2(\bar{a}_1) = (-1, 0)\),
\(\varphi_2(a_2) = (0, 1)\), and \(\varphi_2(\bar{a}_2) = (0, -1)\).
The word problem with respect to \(\varphi_2\) is
\[
  O_2
  =
  \{w \in \{a_1, \bar{a}_1, a_2, \bar{a}_2\}^*
  \mid
  \text{\(\len{w}_{a_1} = \len{w}_{\bar{a}_1}\) and \(\len{w}_{a_2} = \len{w}_{\bar{a}_2}\)}\}.
\]
Thus \(\MIX\) and \(O_2\) are both word problems of the same group \(\Z^2\) with respect to two different choices of generators.

The languages \(\MIX\) and \(O_2\) have a geometric interpretation.
The language \(\MIX\) consists of the closed walks in \(\Z^2\) with step vectors
\((1, 0)\), \((0, 1)\), and \((-1, -1)\)
(see Figure~\ref{fig:mix-o2-closed-walks}(\subref{fig:mix-closed-walk})).
The language \(O_2\) consists of the closed walks in \(\Z^2\) with step vectors
\((1, 0)\), \((-1, 0)\), \((0, 1)\), and \((0, -1)\)
(see Figure~\ref{fig:mix-o2-closed-walks}(\subref{fig:o2-closed-walk})).

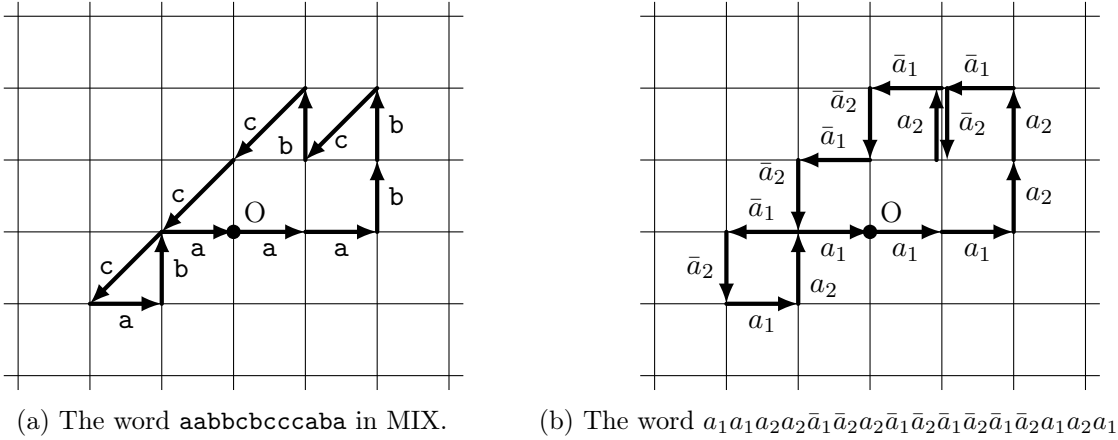
\begin{figure}[htb]
  \centering
  \begin{subfigure}[t]{0.44\textwidth}
    \centering
    \begin{tikzpicture}[scale = 0.95]
      \foreach \x in {-3, ..., 3} {
        \draw (\x, -2.2) to (\x, 3.2);
      }
      \foreach \y in {-2, ..., 3} {
        \draw (-3.2, \y) to (3.2, \y);
      }
      \coordinate (p) at (0, 0);
      \newcommand{\ArrowTo}[4]{\draw[line cap = round, -latex, ultra thick] (p) to node[midway, #1] {#2} ++(#3, #4) coordinate (p);}
      \newcommand{\MoveA}[1][]{\ArrowTo{#1}{\(\tta\)}{1}{0}}
      \newcommand{\MoveB}[1][]{\ArrowTo{#1}{\(\ttb\)}{0}{1}}
      \newcommand{\MoveC}[1][]{\ArrowTo{#1}{\(\ttc\)}{-1}{-1}}
      \MoveA[below]
      \MoveA[below]
      \MoveB[right]
      \MoveB[right]
      \MoveC[below]
      \MoveB[below left]
      \MoveC[left]
      \MoveC[left]
      \MoveC[left]
      \MoveA[below]
      \MoveB[right]
      \MoveA[below]
      \fill (0, 0) node[above right] {O} circle (1mm);
    \end{tikzpicture}
    \caption{The word \(\tta\tta\ttb\ttb\ttc\ttb\ttc\ttc\ttc\tta\ttb\tta\) in \(\MIX\).}
    \label{fig:mix-closed-walk}
  \end{subfigure}
  \hfill
  \begin{subfigure}[t]{0.54\textwidth}
    \centering
    \begin{tikzpicture}[scale = 0.95]
      \foreach \x in {-3, ..., 3} {
        \draw (\x, -2.2) to (\x, 3.2);
      }
      \foreach \y in {-2, ..., 3} {
        \draw (-3.2, \y) to (3.2, \y);
      }
      \coordinate (p) at (0, 0);
      \newcommand{\ArrowTo}[5]{\draw[line cap = round, -latex, ultra thick, #1] (p) to node[midway, #2] {#3} ++(#4, #5) coordinate (p);}
      \newcommand{\MoveX}[2][]{\ArrowTo{#1}{#2}{\(a_1\)}{1}{0}}
      \newcommand{\MoveY}[2][]{\ArrowTo{#1}{#2}{\(a_2\)}{0}{1}}
      \newcommand{\MoveXi}[2][]{\ArrowTo{#1}{#2}{\(\bar{a}_1\)}{-1}{0}}
      \newcommand{\MoveYi}[2][]{\ArrowTo{#1}{#2}{\(\bar{a}_2\)}{0}{-1}}
      \MoveX{below}
      \MoveX{below}
      \MoveY{right}
      \MoveY{right}
      \MoveXi{above}
      \MoveYi[transform canvas = {xshift = 2}]{right}
      \MoveY[transform canvas = {xshift = -2}]{left}
      \MoveXi{above}
      \MoveYi{above left}
      \MoveXi{above}
      \MoveYi{above left}
      \MoveXi{above}
      \MoveYi{left}
      \MoveX{below}
      \MoveY{below right}
      \MoveX{below}
      \fill (0, 0) node[above right] {O} circle (1mm);
    \end{tikzpicture}
    \caption{The word \(a_1 a_1 a_2 a_2 \bar{a}_1 \bar{a}_2 a_2 \bar{a}_1 \bar{a}_2 \bar{a}_1 \bar{a}_2 \bar{a}_1 \bar{a}_2 a_1 a_2 a_1\) in \(O_2\).}
    \label{fig:o2-closed-walk}
  \end{subfigure}
  \caption{Words in \(\MIX\) and \(O_2\) as closed walks in \(\Z^2\).}
  \label{fig:mix-o2-closed-walks}
\end{figure}

Whether the word problem of \(\Z^2\) is indexed has also been studied in group theory.
Gilman remarked that it did not seem to be known whether the word problem of
\(\Z \times \Z\) is indexed \cite{gilman2005formal}.
Elder later explicitly conjectured that the answer is negative \cite{elder2007gautomata}.
The relevant results and conjectures are reviewed in
Subsection~\ref{sec:group-word-problems}.

As proved in Subsection~\ref{sec:rational-equivalence}, \(\MIX\) and \(O_2\) are rationally equivalent.
Since indexed languages are closed under rational transductions,
\(\MIX\) is indexed if and only if \(O_2\) is indexed.
Thus Marsh's conjecture that \(\MIX\) is not indexed is equivalent to the group-theoretic assertion
that the word problem of \(\Z^2\) is not indexed.
To our knowledge, this question remains open
\cites{gebhardt-meunier-salvati2022,nyberg-brodda2023freeproducts}.

\subsubsection{Multiple context-free grammars}

To state the Kanazawa--Salvati conjecture, we first explain what a well-nested \emph{multiple context-free grammar} (MCFG) is.
Every well-nested MCFG generates an indexed language.
We begin by explaining how MCFGs generalize context-free grammars.

Consider the context-free grammar (CFG) with start symbol \(S\) and rules
\[
  S \to \varepsilon,
  \qquad
  S \to S A,
  \qquad
  A \to \tta S \ttb.
\]
Figure~\ref{fig:cfg-to-mcfg}(\subref{fig:cfg-derivation-tree}) shows a derivation tree for the word \(\tta\ttb\tta\tta\ttb\ttb\).
We write \(B \Rightarrow^* w\) if a nonterminal \(B \in \{S, A\}\) derives a terminal word \(w\) in zero or more steps.
A subtree rooted at a nonterminal \(B\) shows that \(B \Rightarrow^* w\), where \(w\) is obtained by reading its leaves from left to right and omitting \(\varepsilon\).
In Figure~\ref{fig:cfg-to-mcfg}(\subref{fig:cfg-inference-tree}), these assertions are displayed as an inference tree.
Each inference corresponds to an application of a CFG rule at a node of the derivation tree.
Its premises are the assertions associated with the nonterminal children of that node.
The rule \(S \to \varepsilon\) corresponds to an inference with no premises.

We can now regard each nonterminal \(B\) as a unary predicate and write \(B(w)\) for \(B \Rightarrow^* w\).
The CFG rules can then be written as
\begin{equation}\label{eq:cfg-as-mcfg}
  S(\varepsilon) \gets,
  \qquad
  S(x y) \gets S(x), A(y),
  \qquad
  A(\tta x \ttb) \gets S(x).
\end{equation}
Here the variables \(x\) and \(y\) range over terminal words.
The arrow \(\gets\) means ``if'': for example, if \(S(x)\) and \(A(y)\) hold, then \(S(x y)\) holds.
The rule \(S(\varepsilon) \gets\) has no premises.
Figure~\ref{fig:cfg-to-mcfg}(\subref{fig:cfg-predicate-tree}) displays the same inference tree in this notation.

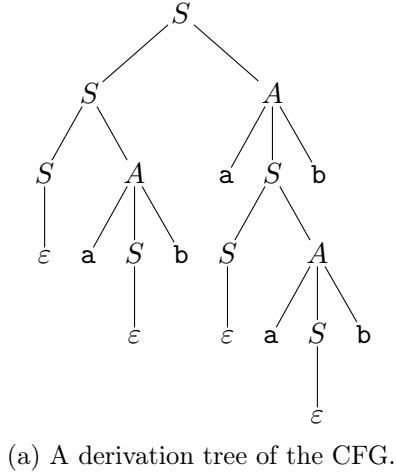
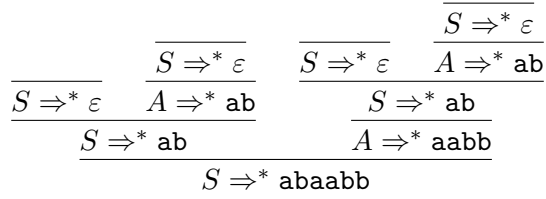
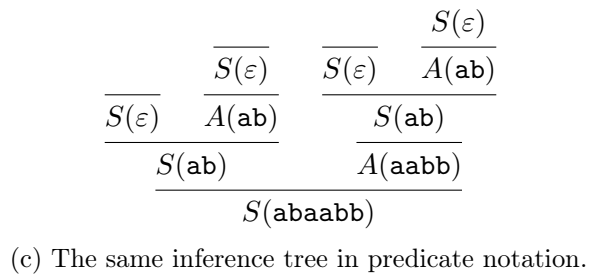
\begin{figure}[htb]
  \centering
  \begin{subfigure}[t]{0.32\textwidth}
    \centering
    \vspace{0pt}
    \begin{forest}
      for tree={s sep=3mm, l sep=3mm, inner sep=1pt}
      [\(S\)
        [\(S\)
          [\(S\)
            [\(\varepsilon\)]
          ]
          [\(A\)
            [\(\tta\)]
            [\(S\)
              [\(\varepsilon\)]
            ]
            [\(\ttb\)]
          ]
        ]
        [\(A\)
          [\(\tta\)]
          [\(S\)
            [\(S\)
              [\(\varepsilon\)]
            ]
            [\(A\)
              [\(\tta\)]
              [\(S\)
                [\(\varepsilon\)]
              ]
              [\(\ttb\)]
            ]
          ]
          [\(\ttb\)]
        ]
      ]
    \end{forest}
    \caption{A derivation tree of the CFG.}
    \label{fig:cfg-derivation-tree}
  \end{subfigure}
  \begin{minipage}[t]{0.5\textwidth}
    \vspace{0pt}
    \begin{subfigure}{\textwidth}
      \centering
      \begin{prooftree}
        \infer0{S \Rightarrow^* \varepsilon}
        \infer0{S \Rightarrow^* \varepsilon}
        \infer1{A \Rightarrow^* \tta\ttb}
        \infer2{S \Rightarrow^* \tta\ttb}
        \infer0{S \Rightarrow^* \varepsilon}
        \infer0{S \Rightarrow^* \varepsilon}
        \infer1{A \Rightarrow^* \tta\ttb}
        \infer2{S \Rightarrow^* \tta\ttb}
        \infer1{A \Rightarrow^* \tta\tta\ttb\ttb}
        \infer2{S \Rightarrow^* \tta\ttb\tta\tta\ttb\ttb}
      \end{prooftree}
      \caption{The corresponding inference tree.}
      \label{fig:cfg-inference-tree}
    \end{subfigure}

    \medskip
    \begin{subfigure}{\textwidth}
      \centering
      \begin{prooftree}
        \infer0{S(\varepsilon)}
        \infer0{S(\varepsilon)}
        \infer1{A(\tta\ttb)}
        \infer2{S(\tta\ttb)}
        \infer0{S(\varepsilon)}
        \infer0{S(\varepsilon)}
        \infer1{A(\tta\ttb)}
        \infer2{S(\tta\ttb)}
        \infer1{A(\tta\tta\ttb\ttb)}
        \infer2{S(\tta\ttb\tta\tta\ttb\ttb)}
      \end{prooftree}
      \caption{The same inference tree in predicate notation.}
      \label{fig:cfg-predicate-tree}
    \end{subfigure}
  \end{minipage}
  \caption{Three representations of a derivation of \(\tta\ttb\tta\tta\ttb\ttb\).
  In (b), each conclusion records the word derived by a nonterminal in the corresponding subtree of (a).
  The inference tree in (c) is obtained from that in (b) by replacing each assertion \(B \Rightarrow^* w\) with \(B(w)\).}
  \label{fig:cfg-to-mcfg}
\end{figure}

An MCFG generalizes this description by allowing each nonterminal \(B\) to have a fixed arity \(r \geq 1\) and to derive an \(r\)-tuple of words.
We write \(B(w_1, \dotsc, w_r)\) if the tuple \((w_1, \dotsc, w_r)\) is derivable from \(B\).
A rule of an MCFG derives a tuple by concatenating terminals and components of the tuples in its premises, using each component of those tuples at most once.
The start symbol has arity \(1\), so the grammar still generates a language of words.
The \emph{dimension} of an MCFG is the maximum arity of its nonterminals.
Thus the grammar with rules \eqref{eq:cfg-as-mcfg} has dimension \(1\).

We give an example of an MCFG that generates the non-context-free language \(L\) in \eqref{eq:shieber-language}.
Consider the MCFG with start symbol \(S\), a nonterminal \(A\) of arity \(2\), and rules
\begin{equation}\label{eq:mcfg-shieber-language}
  \begin{aligned}
    A(\tta\ttb, \ttc\ttd) &\gets,\\
    A(\tta x, \ttc y) &\gets A(x, y),\\
    A(x \ttb, y \ttd) &\gets A(x, y),\\
    S(x y) &\gets A(x, y).
  \end{aligned}
\end{equation}
The tuples derivable from \(A\) are exactly \((\tta^m \ttb^n, \ttc^m \ttd^n)\) with \(m, n \geq 1\).
Indeed, the first rule derives the tuple \((\tta\ttb, \ttc\ttd)\), corresponding to \(m = n = 1\).
Each application of the second rule increases \(m\) by \(1\) and leaves \(n\) unchanged.
Each application of the third rule increases \(n\) by \(1\) and leaves \(m\) unchanged.
The start rule concatenates the two components, so the grammar generates exactly \(L\).
For example, for \(m = 2\) and \(n = 3\), we have the following derivation:
\[
  \begin{prooftree}
    \infer0{A(\tta\ttb, \ttc\ttd)}
    \infer1{A(\tta\tta\ttb, \ttc\ttc\ttd)}
    \infer1{A(\tta\tta\ttb\ttb, \ttc\ttc\ttd\ttd)}
    \infer1{A(\tta\tta\ttb\ttb\ttb, \ttc\ttc\ttd\ttd\ttd)}
    \infer1{S(\tta\tta\ttb\ttb\ttb\ttc\ttc\ttd\ttd\ttd)\rlap{.}}
  \end{prooftree}
\]

A \emph{well-nested} MCFG is an MCFG with additional restrictions on how its rules use and arrange the components of the tuples in their premises.
The grammars with rules \eqref{eq:cfg-as-mcfg} and \eqref{eq:mcfg-shieber-language} are both well-nested.
Precise definitions of MCFGs, their dimensions, and well-nestedness are given in Subsection~\ref{sec:well-nested-mcfl}.

\subsubsection{The Kanazawa--Salvati conjecture as an intermediate problem}

Following Sorokin~\cite{sorokin2014}, we call the conjecture that \(\MIX\) is not generated by any well-nested MCFG the \emph{Kanazawa--Salvati conjecture}.
Kanazawa and Salvati had proposed that well-nested MCFGs might provide a better formal approximation to Joshi's informal notion of mildly context-sensitive grammars \cite{kanazawa-salvati2010copying}*{Section~1}.
Salvati later explicitly conjectured that neither \(\MIX\) nor \(O_2\) is a well-nested multiple context-free language \cite{salvati2015}*{Section~5}.

We first consider well-nested MCFGs of dimension at most \(2\).
The class of languages generated by well-nested MCFGs of dimension at most \(2\) coincides with the class of tree-adjoining languages
\cites{joshi-schabes1997,kanazawa2009pumping}.
Joshi conjectured that \(\MIX\) is not a tree-adjoining language
\cite{joshi1985tag}*{Example~6.3.4}.
Kanazawa and Salvati proved this conjecture.

\begin{theorem}[Kanazawa--Salvati {\cite{kanazawa-salvati2012}}]\label{thm:kanazawa-salvati-dimension-two}
  The language \(\MIX\) is not generated by any well-nested MCFG of dimension at most \(2\). \qed
\end{theorem}

In contrast, Salvati proved the following.

\begin{theorem}[Salvati {\cite{salvati2015}}]\label{thm:salvati-dimension-two}
  The language \(\MIX\) is generated by an MCFG of dimension \(2\). \qed
\end{theorem}

By Theorem~\ref{thm:kanazawa-salvati-dimension-two}, every MCFG of dimension at most \(2\) that generates \(\MIX\) must contain a non-well-nested rule.
Salvati's MCFG indeed contains such a rule.
Nederhof later gave another proof that \(\MIX\) is generated by an MCFG of dimension at most \(2\),
using the rationally equivalent language \(O_2\) and essentially the same MCFG \cite{nederhof2016}*{Sections~1 and 7}.
To our knowledge, it remains open whether \(\MIX\) is generated by a well-nested MCFG of dimension at most \(3\).

We use the following notation.
We write \(\TAL\) for the class of tree-adjoining languages and \(\IL\) for the class of indexed languages.
For \(m \geq 1\), let \(\mMCFL{m}\) denote the class of languages generated by MCFGs of dimension at most \(m\),
and let \(\mMCFLwn{m}\) denote the class of languages generated by well-nested MCFGs of dimension at most \(m\).
Let \(\MCFLwn\) denote the class of all well-nested multiple context-free languages.

These classes satisfy
\[
  \TAL
  =
  \mMCFLwn{2}
  \subsetneq
  \mMCFLwn{3}
  \subsetneq
  \dotsb,
  \qquad
  \MCFLwn
  =
  \bigcup_{m \geq 1} \mMCFLwn{m}
  \subsetneq
  \IL.
\]
For the first equality \(\TAL = \mMCFLwn{2}\), see
\cites{joshi-schabes1997,kanazawa2009pumping}.
The second equality holds by the definition of \(\MCFLwn\),
and the strict inclusions are proved in Lemma~\ref{lem:language-class-separations}.
Theorem~\ref{thm:kanazawa-salvati-dimension-two} shows that \(\MIX \notin \mMCFLwn{2}\).
The Kanazawa--Salvati conjecture asserts that \(\MIX \notin \mMCFLwn{m}\) for every \(m \geq 1\),
whereas Marsh's conjecture asserts that \(\MIX \notin \IL\)
\cite{marsh1985conjectures}.
Consequently,
\[
  \text{Marsh's conjecture}
  \implies
  \text{the Kanazawa--Salvati conjecture}
  \implies
  \MIX \notin \TAL.
\]

The Kanazawa--Salvati conjecture is therefore a natural intermediate problem
when one asks whether the word problem of \(\Z^2\) is indexed.
An affirmative resolution of the Kanazawa--Salvati conjecture would establish that \(\MIX \notin \MCFLwn\) but would not determine whether \(\MIX\) is indexed, since \(\MCFLwn \subsetneq \IL\).
A negative resolution would establish that \(\MIX \in \MCFLwn\), which would imply that \(\MIX\) is indexed and refute Marsh's conjecture.

\subsection{Main results}
\label{subsec:intro-main-results}

Let \(\Sigma\) be a finite alphabet, let \(d \geq 1\), and let \(\psi\colon \Sigma^* \to \Z^d\) be a surjective monoid homomorphism.
Put \(\Wpsi = \psi^{-1}(\bm{0})\).
This is the word problem of \(\Z^d\) with respect to \(\psi\).
For each \(r \geq 1\), we construct a well-nested MCFG \(\Gpsi{r}\)
(see Subsection~\ref{sec:simulation} for the definition).

\begin{theoremA}\label{thm:theorem-a}
  For every well-nested multiple context-free grammar \(G\) over \(\Sigma\) with \(L(G) \subseteq \Wpsi\),
  there exists \(r \geq 1\) such that \(L(G) \subseteq L(\Gpsi{r}) \subseteq \Wpsi\).
\end{theoremA}

In other words, the family \((L(\Gpsi{r}))_{r \geq 1}\) is cofinal with respect to inclusion
among the well-nested multiple context-free sublanguages of \(\Wpsi\).
Applying Theorem~\ref{thm:theorem-a} to \(\psin{3}\) and using Proposition~\ref{prop:g3-tuples} yields Corollary~\ref{cor:combinatorial-reformulation}, which establishes the explicit combinatorial reformulation stated in Subsection~\ref{subsec:intro-combinatorial-reformulation}.
This reformulation provides a concrete way to approach the conjecture.
Here \(r\) is the number of components in the tuples manipulated by \(\Gpsi{r}\), equivalently, the arity of its nonterminal \(I\).
In general, it may be necessary to choose \(r\) larger than the dimension \(m\) of \(G\), since the inclusion \(L(G) \subseteq L(\Gpsi{m})\) need not hold.

The second main result concerns the language
\[
  \MIX_4
  =
  \{w \in \{\tta, \ttb, \ttc, \ttd\}^*
  \mid
  \len{w}_{\tta} = \len{w}_{\ttb} = \len{w}_{\ttc} = \len{w}_{\ttd}\}.
\]
Theorem~\ref{thm:origin-crossing-dimension} and the rational equivalence of \(O_3\) and \(\MIX_4\) show that \(O_3, \MIX_4 \in \mMCFL{3} \setminus \mMCFL{2}\).
To our knowledge, it had remained open whether \(\MIX_4\), or equivalently \(O_3\), belonged to \(\mMCFLwn{3}\), or more generally to \(\MCFLwn\).

\begin{theoremB}\label{thm:theorem-b}
  The language \(\MIX_4\) is not a well-nested multiple context-free language.
  Equivalently, \(O_3\) is not a well-nested multiple context-free language.
\end{theoremB}

Equivalently, Theorem~\ref{thm:theorem-b} shows that
\(\MIX_4 \notin \mMCFLwn{m}\) for every dimension \(m \geq 1\).
Corollary~\ref{cor:mix-n-not-well-nested} shows that
\(\MIX_n \notin \MCFLwn\) for every \(n \geq 4\) and that
\(O_n \notin \MCFLwn\) for every \(n \geq 3\).
In other words, for every \(n \geq 3\), the word problem \(O_n\) of the free abelian group
\(\Z^n\) of rank \(n\) does not belong to \(\MCFLwn\).
For rank \(0\), \(\Z^0 = \{0\}\), and its word problem with respect to any choice of generators
\(\pi\colon \Sigma^* \to \Z^0\) is \(\Sigma^*\), hence regular.
For rank \(1\), \(O_1\) agrees with \(\MIX_2\) up to a renaming of letters and is context-free
\cite{hopcroft-motwani-ullman2007}*{Exercises~5.1.8 and 6.2.1~c)}.
Thus \(\Z^2\) is the only finitely generated free abelian group for which it remains open
whether it has a word problem in \(\MCFLwn\).
The Kanazawa--Salvati conjecture asserts that it does not.

\subsection{A combinatorial reformulation of the Kanazawa--Salvati conjecture}
\label{subsec:intro-combinatorial-reformulation}

Before turning to the technical development, we restate the Kanazawa--Salvati conjecture as an explicit combinatorial problem.
All objects needed to state this problem are defined below, and the reformulation itself does not require the definition of a multiple context-free grammar.
Recall that
\[
  \MIX
  =
  \{w \in \{\tta, \ttb, \ttc\}^*
  \mid
  \len{w}_{\tta} = \len{w}_{\ttb} = \len{w}_{\ttc}\},
\]
where \(\len{w}_x\) is the number of occurrences of \(x\) in \(w\),
and \(\varepsilon\) denotes the empty word.
Write \(\MCFLwn\) for the class of languages generated by well-nested multiple context-free grammars.
The Kanazawa--Salvati conjecture asserts that \(\MIX \notin \MCFLwn\).
The following reformulation uses one explicitly defined set of tuples for each positive integer \(r\)
and does not refer to arbitrary well-nested multiple context-free grammars.

\begin{definition}\label{def:regroup}
  Let \(\vv{w} = (w_1, \dotsc, w_m)\) be a tuple of words.
  For \(0 \leq i \leq j \leq m\), put
  \[
    w_{(i, j]} = w_{i + 1} \dotsm w_j,
  \]
  where \(w_{(i, i]} = \varepsilon\).
  For \(r \geq 1\), define
  \[
    \Regroup_r(\vv{w})
    =
    \{
    (w_{(j(0), j(1)]}, w_{(j(1), j(2)]}, \dotsc, w_{(j(r - 1), j(r)]})
    \mid
    0 = j(0) \leq j(1) \leq \dotsb \leq j(r) = m
    \}.
  \]
  Thus \(\Regroup_r(\vv{w})\) consists of all \(r\)-tuples obtained by regrouping \(w_1, \dotsc, w_m\) into \(r\) consecutive components,
  allowing empty components.
  For example,
  \begin{align*}
    \Regroup_3(\tta, \ttb, \ttc, \ttb\tta) =
    \mleft\{
      \begin{aligned}
        &(\tta\ttb\ttc\ttb\tta, \varepsilon, \varepsilon),
        (\tta\ttb\ttc, \ttb\tta, \varepsilon),
        (\tta\ttb\ttc, \varepsilon, \ttb\tta),
        (\tta\ttb, \ttc\ttb\tta, \varepsilon),
        (\tta\ttb, \ttc, \ttb\tta),\\
        &(\tta\ttb, \varepsilon, \ttc\ttb\tta),
        (\tta, \ttb\ttc\ttb\tta, \varepsilon),
        (\tta, \ttb\ttc, \ttb\tta),
        (\tta, \ttb, \ttc\ttb\tta),
        (\tta, \varepsilon, \ttb\ttc\ttb\tta),\\
        &(\varepsilon, \tta\ttb\ttc\ttb\tta, \varepsilon),
        (\varepsilon, \tta\ttb\ttc, \ttb\tta),
        (\varepsilon, \tta\ttb, \ttc\ttb\tta),
        (\varepsilon, \tta, \ttb\ttc\ttb\tta),
        (\varepsilon, \varepsilon, \tta\ttb\ttc\ttb\tta)
      \end{aligned}
    \mright\}.
  \end{align*}
  In general, if \(w_1, \dotsc, w_m\) are all nonempty words,
  then \(\Regroup_r(\vv{w})\) has \(\binom{m + r - 1}{m}\) elements
  by the usual ``balls-and-bars'' argument.
\end{definition}

For each \(r \geq 1\), define \(\mathcal{A}_r\) to be the smallest subset of
\((\{\tta, \ttb, \ttc\}^*)^r\) satisfying the following three conditions.

\begin{enumerate}[leftmargin=*]
  \item\label{item:combinatorial-empty-tuple}
        \((\underbrace{\varepsilon, \dotsc, \varepsilon}_r) \in \mathcal{A}_r\).
  \item\label{item:combinatorial-letter-insertion}
        If \((w_1, \dotsc, w_r) \in \mathcal{A}_r\) and
        \(x_1 y_1 \dotsm x_r y_r\) is a permutation of \(\tta\ttb\ttc\) with
        \(x_j, y_j \in \{\varepsilon, \tta, \ttb, \ttc\}\) for \(1 \leq j \leq r\), then
  \[
    (x_1 w_1 y_1, \dotsc, x_r w_r y_r) \in \mathcal{A}_r .
  \]
  For example, when \(r = 3\), if \((w_1, w_2, w_3) \in \mathcal{A}_3\), then the tuples
  \[
    (\tta w_1 \ttb, \ttc w_2, w_3), \quad
    (\ttc w_1 \tta, \ttb w_2, w_3), \quad
    (w_1 \ttb, w_2 \tta, \ttc w_3), \quad
    (w_1, w_2 \ttc, \tta w_3 \ttb)
  \]
  all belong to \(\mathcal{A}_3\).
  \item\label{item:combinatorial-binary-combination}
        If \((w_1, \dotsc, w_r), (w'_1, \dotsc, w'_r) \in \mathcal{A}_r\), then for every \(0 \leq i \leq r\),
  \[
    \Regroup_r
    (w_1, \dotsc, w_i, w'_1, \dotsc, w'_r, w_{i + 1}, \dotsc, w_r)
    \subseteq \mathcal{A}_r .
  \]
  For example, when \(r = 2\), if \((w_1, w_2), (w'_1, w'_2) \in \mathcal{A}_2\), then
  \[
    \Regroup_2(w_1, w'_1, w'_2, w_2)
    =
    \{
    (\varepsilon, w_1 w'_1 w'_2 w_2),
    (w_1, w'_1 w'_2 w_2),
    \dotsc,
    (w_1 w'_1 w'_2 w_2, \varepsilon)
    \}
    \subseteq \mathcal{A}_2 .
  \]
  Similarly,
  \[
    \Regroup_2(w_1, w_2, w'_1, w'_2)
    =
    \{
    (\varepsilon, w_1 w_2 w'_1 w'_2),
    (w_1, w_2 w'_1 w'_2),
    \dotsc,
    (w_1 w_2 w'_1 w'_2, \varepsilon)
    \}
    \subseteq \mathcal{A}_2 .
  \]
  This operation first inserts the components \(w'_1, \dotsc, w'_r\) into one of the \(r + 1\) gaps before, between,
  or after the components \(w_1, \dotsc, w_r\).
  It then applies \(\Regroup_r\), regrouping the resulting components into \(r\) consecutive components,
  allowing empty components, without cutting any word \(w_j\) or \(w'_j\) internally.
\end{enumerate}

Conditions~\ref{item:combinatorial-empty-tuple} and~\ref{item:combinatorial-binary-combination}
imply the following useful closure property:
\begin{equation}\label{eq:a-r-regrouping-closure}
  (w_1, \dotsc, w_r) \in \mathcal{A}_r
  \implies
  \Regroup_r(w_1, \dotsc, w_r)
  =
  \Regroup_r
  (w_1, \dotsc, w_r,
  \underbrace{\varepsilon, \dotsc, \varepsilon}_r)
  \subseteq
  \mathcal{A}_r .
\end{equation}

Finally, define \(\concat(\mathcal{A}_r)\) by
\[
  \concat(\mathcal{A}_r) = \{w_1 \dotsm w_r \mid (w_1, \dotsc, w_r) \in \mathcal{A}_r\}.
\]
By induction on the definition of \(\mathcal{A}_r\), we have \(\concat(\mathcal{A}_r) \subseteq \MIX\).
Corollary~\ref{cor:combinatorial-reformulation} characterizes the Kanazawa--Salvati conjecture as follows:
\[
  \MIX \notin \MCFLwn
  \iff
  \text{\(\concat(\mathcal{A}_r) \subsetneq \MIX\) for every \(r \geq 1\).}
\]
Equivalently, the conjecture asks whether the following statement holds:
\[
  \forall r \geq 1\;
  \exists w \in \MIX\;
  \forall (w_1, \dotsc, w_r) \in (\{\tta, \ttb, \ttc\}^*)^r\;
  [
    w = w_1 \dotsm w_r
    \implies
    (w_1, \dotsc, w_r) \notin \mathcal{A}_r
  ].
\]

It is easy to see that
\[
  \concat(\mathcal{A}_1)
  \subseteq
  \concat(\mathcal{A}_2)
  \subseteq
  \dotsb,
  \qquad
  \MIX
  =
  \bigcup_{r \geq 1} \concat(\mathcal{A}_r)
\]
(see Remark~\ref{rem:first-two-combinatorial-cases} for details).
Thus the Kanazawa--Salvati conjecture is equivalent to the assertion that this ascending chain does not stabilize.

\begin{example}\label{ex:combinatorial-witness}
  Let \(w = \tta\tta\tta\ttb\ttb\ttb\ttc\ttb\ttb\ttc\ttc\ttc\ttc\tta\tta\).
  We show that \(w \in \concat(\mathcal{A}_3)\).
  \begin{enumerate}[label=\textup{\arabic*.}, ref=\arabic*, leftmargin=*]
    \item\label{step:witness-empty}
      By condition~\ref{item:combinatorial-empty-tuple}, we have \((\varepsilon, \varepsilon, \varepsilon) \in \mathcal{A}_3\).
    \item\label{step:witness-abc}
      By step~\ref{step:witness-empty} and condition~\ref{item:combinatorial-letter-insertion}, we have \((\tta, \ttb, \ttc) \in \mathcal{A}_3\).
    \item\label{step:witness-aabbcc}
      By step~\ref{step:witness-abc} and condition~\ref{item:combinatorial-letter-insertion}, we have \((\tta\tta, \ttb\ttb, \ttc\ttc) \in \mathcal{A}_3\).
    \item\label{step:witness-binary}
      Applying condition~\ref{item:combinatorial-binary-combination} with \(i = 1\) to the tuples in steps~\ref{step:witness-aabbcc} and~\ref{step:witness-abc}, we obtain
      \[
        (\tta\tta\tta, \ttb\ttc\ttb\ttb\ttc\ttc, \varepsilon)
        \in \Regroup_3(\tta\tta, \tta, \ttb, \ttc, \ttb\ttb, \ttc\ttc)
        \subseteq \mathcal{A}_3.
      \]
    \item\label{step:witness-four}
      By step~\ref{step:witness-binary} and condition~\ref{item:combinatorial-letter-insertion}, we have \((\tta\tta\tta, \ttb\ttb\ttc\ttb\ttb\ttc\ttc\ttc, \tta) \in \mathcal{A}_3\).
    \item\label{step:witness-five}
      By step~\ref{step:witness-four} and condition~\ref{item:combinatorial-letter-insertion}, we have \((\tta\tta\tta, \ttb\ttb\ttb\ttc\ttb\ttb\ttc\ttc\ttc\ttc, \tta\tta) \in \mathcal{A}_3\).
    \item
      By step~\ref{step:witness-five} and the definition of \(\concat\), we have
      \(w = \tta\tta\tta\ttb\ttb\ttb\ttc\ttb\ttb\ttc\ttc\ttc\ttc\tta\tta \in \concat(\mathcal{A}_3)\).
  \end{enumerate}
  On the other hand, an exhaustive computer search shows that \(w \notin \concat(\mathcal{A}_2)\); see Remark~\ref{rem:first-two-combinatorial-cases}.
\end{example}

It is also noted in Remark~\ref{rem:first-two-combinatorial-cases} that
\(\concat(\mathcal{A}_r) \subsetneq \MIX\) for \(r = 1, 2\),
and hence \(r = 3\) is the first open case.

\subsection{Structure of the paper}
\label{subsec:intro-outline}

After the preliminaries in Section~\ref{sec:preliminaries},
Section~\ref{sec:proof-theorem-a} proves Theorem~\ref{thm:theorem-a}
and derives the combinatorial reformulation of the Kanazawa--Salvati conjecture,
and Section~\ref{sec:four-letter-consequence} proves Theorem~\ref{thm:theorem-b}.
Section~\ref{sec:further-background} provides supplementary background,
and Section~\ref{sec:conclusion} summarizes the consequences of the main results
and states the remaining combinatorial problem.
The appendices contain deferred proofs, auxiliary results, and related open problems.

\section{Preliminaries}\label{sec:preliminaries}

\subsection{Basic notation}
\label{sec:words-languages-homomorphisms}

We write \(\N = \{0, 1, 2, \dotsc\}\).
An \emph{alphabet} is a finite set, and its elements are called \emph{letters}.
A \emph{word} over an alphabet \(\Sigma\) is a finite sequence of letters of \(\Sigma\).
For a word \(w\), we write \(\len{w}\) for its \emph{length}.
The unique word of length zero is the \emph{empty word}, denoted by \(\varepsilon\).
We write \(\Sigma^*\) for the set of words over \(\Sigma\) and \(\Sigma^n\) for the set of words of length \(n\).
We sometimes identify a letter \(a \in \Sigma\) with the corresponding word \(a \in \Sigma^*\) of length one.
For a word \(w \in \Sigma^*\) and a letter \(a \in \Sigma\), we write \(\len{w}_a\) for the number of occurrences of \(a\) in \(w\).
A \emph{language} over \(\Sigma\) is a subset of \(\Sigma^*\).
If \(K \subseteq L \subseteq \Sigma^*\), we call \(K\) a \emph{sublanguage} of \(L\).
For words \(u, v \in \Sigma^*\), their \emph{concatenation} is denoted by \(u v\).
A word \(u\) is called a \emph{factor} of a word \(v\) if \(v = xuy\) for some \(x, y \in \Sigma^*\).
A factor \(u\) of \(v\) is called a \emph{proper factor} if \(u \neq v\).

With concatenation as the binary operation and \(\varepsilon\) as the neutral element, \(\Sigma^*\) is the \emph{free monoid} generated by \(\Sigma\).
Let \(M\) be a monoid with neutral element \(1_M\).
A map \(\psi\colon \Sigma^* \to M\) is called a \emph{monoid homomorphism} if \(\psi(\varepsilon) = 1_M\) and \(\psi(u v) = \psi(u) \psi(v)\) for every \(u, v \in \Sigma^*\).
If \(M = \Z^d\) for some \(d \geq 1\), then the group operation on \(\Z^d\) is written additively, and its neutral element is \(\bm{0}\).
Thus, for a homomorphism \(\psi\colon \Sigma^* \to \Z^d\), we have \(\psi(\varepsilon) = \bm{0}\) and \(\psi(u v) = \psi(u) + \psi(v)\).
Such a homomorphism is determined by its values on the letters of \(\Sigma\).

For \(n \in \N\), put \([n] = \{1, \dotsc, n\}\), where \([0] = \varnothing\).
For \(w = a_1 \dotsm a_n \in \Sigma^n\) and a nonempty subset \(J = \{j_1, \dotsc, j_l\} \subseteq [n]\), where \(j_1 < \dotsb < j_l\),
define \(w[J] = a_{j_1} \dotsm a_{j_l}\),
and put \(w[\varnothing] = \varepsilon\).
We call a word \(u\) a \emph{subsequence} of \(w\), and write \(u \preccurlyeq w\), if \(u = w[J]\) for some \(J \subseteq [n]\).
We write \(u \prec w\) if \(u \preccurlyeq w\) and \(u \neq w\).
For a language \(L \subseteq \Sigma^*\),
we write \(\Minsubseq(L)\) for the set of minimal elements of \(L\) with respect to \(\preccurlyeq\),
that is, the set of elements \(u \in L\) such that there is no \(v \in L\) with \(v \prec u\).
If \(w = a_1 \dotsm a_k\) (\(a_i \in \Sigma\)) is a word, we write \(\Perm(w)\) for the set of all words \(a_{\sigma(1)} \dotsm a_{\sigma(k)}\),
where \(\sigma\) ranges over all permutations of \(\{1, \dotsc, k\}\).

For an alphabet \(\Omega\) and a subset \(\Delta \subseteq \Omega\), let
\(\proj_\Delta\colon \Omega^* \to \Delta^*\) be the monoid homomorphism defined by
\(\proj_\Delta(a) = a\) for \(a \in \Delta\) and \(\proj_\Delta(a) = \varepsilon\) for \(a \in \Omega \setminus \Delta\).
We call \(\proj_\Delta\) the \emph{projection onto \(\Delta\)}.
We regard \((\Omega^*)^*\) as the set of all tuples of words over \(\Omega\).
For tuples of words, we use symbols with arrows, such as \(\vv{w} = (w_1, \dotsc, w_r)\).
For alphabets \(\Omega\) and \(\Gamma\) and a monoid homomorphism \(h\colon \Omega^* \to \Gamma^*\),
we also write \(h\) for its componentwise extension to tuples, defined by
\[
  h(\vv{w}) = (h(w_1), \dotsc, h(w_r))
  \quad\text{for}\quad
  \vv{w} = (w_1, \dotsc, w_r) \in (\Omega^*)^*.
\]
For an \(r\)-ary symbol \(A\) and \(\vv{w} = (w_1, \dotsc, w_r)\), we write \(A(\vv{w})\) for \(A(w_1, \dotsc, w_r)\).
Define the map \(\concat\colon (\Omega^*)^* \to \Omega^*\) by
\(\concat(\vv{w}) = w_1 \dotsm w_r\) for \(\vv{w} = (w_1, \dotsc, w_r) \in (\Omega^*)^*\), where \(r \geq 0\).

Let \(\vv{u}\) and \(\vv{v}\) be tuples of words over the same alphabet, each having at least one component.
We write \(\vv{u} \sqsubseteq \vv{v}\) if
\(\vv{u} \in \Regroup_r(\vv{v})\) for some \(r \geq 1\),
where \(\Regroup_r\) is defined in Definition~\ref{def:regroup}.
The relation \(\sqsubseteq\) is reflexive and transitive.
Componentwise extensions of monoid homomorphisms preserve \(\sqsubseteq\):
if \(h\colon \Omega^* \to \Gamma^*\) is a monoid homomorphism and \(\vv{u} \sqsubseteq \vv{v}\), then
\(h(\vv{u}) \sqsubseteq h(\vv{v})\).
If \(\vv{v}'\) is obtained from \(\vv{v}\) by inserting empty components, then
\(\vv{v} \sqsubseteq \vv{v}' \sqsubseteq \vv{v}\); in particular, \(\sqsubseteq\) is not antisymmetric.

\subsection{The languages
\texorpdfstring{\(\Wpsi\)}{Wpsi},
\texorpdfstring{\(\MIX_n\)}{MIXn}, and
\texorpdfstring{\(O_n\)}{On}}
\label{sec:lpsi-mixn}

For a finite alphabet \(\Sigma\), an integer \(d \geq 1\), and a surjective monoid homomorphism
\(\psi\colon \Sigma^* \to \Z^d\), put
\[
  \Wpsi = \psi^{-1}(\bm{0}).
\]
The language \(\Wpsi\) is the word problem of \(\Z^d\) with respect to the choice of generators \(\psi\).

For \(n \geq 2\), set \(\Sigman{n} = \{\tta_1, \dotsc, \tta_n\}\), and define the homomorphism \(\psin{n}\colon \Sigman{n}^* \to \Z^{n - 1}\) by
\[
  \psin{n}(\tta_i) = \bm{e}_i\quad(1 \leq i \leq n - 1),
  \qquad
  \psin{n}(\tta_n) = (-1, \dotsc, -1),
\]
where \(\bm{e}_i\) is the \(i\)-th standard basis vector of \(\Z^{n - 1}\).
Then define \(\MIX_n\) by
\[
  \MIX_n = W_{\psin{n}} = \psin{n}^{-1}(\bm{0}).
\]
For \(n \leq 4\), we write \(\tta = \tta_1\), \(\ttb = \tta_2\), \(\ttc = \tta_3\), and \(\ttd = \tta_4\).

For \(n \geq 1\), set
\(\Delta_n = \{a_1, \bar{a}_1, \dotsc, a_n, \bar{a}_n\}\) and define
\[
  O_n
  =
  \{w \in \Delta_n^*
  \mid
  \text{\(\len{w}_{a_i} = \len{w}_{\bar{a}_i}\) for every \(1 \leq i \leq n\)}\}.
\]
This is the word problem of \(\Z^n\) with respect to the choice of generators
\(\varphi_n\colon \Delta_n^* \to \Z^n\) defined by
\(\varphi_n(a_i) = \bm{e}_i\) and \(\varphi_n(\bar{a}_i) = -\bm{e}_i\) for \(1 \leq i \leq n\).

Thus \(O_n\) is the word problem of \(\Z^n\) with respect to \(\varphi_n\),
whereas \(\MIX_{n + 1}\) is the word problem of \(\Z^n\) with respect to \(\psin{n + 1}\).

\subsection{Well-nested multiple context-free languages}
\label{sec:well-nested-mcfl}

\subsubsection{Multiple context-free grammars}

Multiple context-free grammars (MCFGs) were introduced by Seki--Matsumura--Fujii--Kasami
as a generalization of context-free grammars \cite{seki1991}.
Vijay-Shanker, Weir, and Joshi introduced linear context-free rewriting systems (LCFRSs)
as a general framework for comparing the structural descriptions produced by several grammatical formalisms
\cite{vijay-shanker-weir-joshi1987}.
We recall their relation to MCFGs in Subsection~\ref{sec:equivalent-formalisms}.
Following \cite{kanazawa2019ogden}*{Section~2.1}, we write rules of MCFGs as Horn clauses;
the same rule notation is used in \cite{kanazawa-michaelis-salvati-yoshinaka2011}*{Section~2}.
This way of writing rules is equivalent to the original presentation in terms of functions.

\begin{definition}
  A \emph{multiple context-free grammar} (MCFG for short) is a tuple \(G = (N, \Sigma, P, S)\),
  where
  \begin{itemize}
    \item \(N = \bigcup_{r \geq 1} N^{(r)}\) is a finite set of \emph{nonterminals},
      and the family \((N^{(r)})_{r \geq 1}\) is pairwise disjoint,
    \item \(\Sigma\) is a finite alphabet of \emph{terminals} disjoint from \(N\),
    \item \(P\) is a finite set of \emph{rules} (defined below), and
    \item \(S \in N^{(1)}\) is the \emph{start symbol}.
  \end{itemize}
  A \emph{rule} is an expression of the form
  \begin{equation}\label{eq:rule}
    A(t_1, \dotsc, t_r) \gets B_1(x_{1, 1}, \dotsc, x_{1, r_1}), \dotsc, B_n(x_{n, 1}, \dotsc, x_{n, r_n})
  \end{equation}
  where
  \begin{itemize}
    \item \(A \in N^{(r)}\) and \(B_i \in N^{(r_i)}\) for \(1 \leq i \leq n\),
    \item the \emph{variables} \(x_{i, j}\) are pairwise distinct, and
    \item each \(t_k\) (\(1 \leq k \leq r\)) is a word over \(\Sigma \cup \{x_{i, j} \mid 1 \leq i \leq n, 1 \leq j \leq r_i\}\)
      such that each variable \(x_{i, j}\) (\(1 \leq i \leq n, 1 \leq j \leq r_i\)) occurs at most once in the concatenation \(t_1 \dotsm t_r\).
  \end{itemize}
  The left-hand side and right-hand side of \eqref{eq:rule} are called the \emph{head} and the \emph{body} of the rule, respectively.
  An \emph{atom} is an expression \(A(u_1, \dotsc, u_r)\), where \(A \in N^{(r)}\) and each \(u_i\) is a word over terminals and variables.
  An atom is \emph{ground} if it contains no variables.
\end{definition}

If \(A \in N^{(r)}\), then \(r\) is called the \emph{arity} of \(A\).
The \emph{dimension} of \(G\) is the maximum arity of its nonterminals.
The \emph{branching} of a rule is the number of occurrences of nonterminals in its body.
The \emph{branching factor} of \(G\) is the maximum branching of its rules (or \(0\) if \(P = \varnothing\)).
Rules of branching \(0\), \(1\), and \(2\) are called nullary, unary, and binary, respectively.
More generally, a rule of branching \(n\) is \(n\)-ary.
A rule whose head is the start symbol is called a \emph{start rule}.
A grammar is \emph{non-branching} if its branching factor is at most \(1\), \emph{binary branching} if its branching factor is at most \(2\),
and, more generally, \emph{\(n\)-ary branching} if its branching factor is at most \(n\).

\begin{definition}
  A rule
  \[
    A(t_1, \dotsc, t_r) \gets B_1(x_{1, 1}, \dotsc, x_{1, r_1}), \dotsc, B_n(x_{n, 1}, \dotsc, x_{n, r_n})
  \]
  is \emph{non-deleting} if every variable \(x_{i, j}\) occurs exactly once in \(t_1 \dotsm t_r\).
  It is \emph{non-permuting} if there do not exist indices \(i, j, k\) such that
  \(1 \leq i \leq n\), \(1 \leq j < k \leq r_i\), and \(x_{i, k} x_{i, j} \preccurlyeq t_1 \dotsm t_r\).
  Note that the rule is non-deleting and non-permuting if and only if \(x_{i, 1} \dotsm x_{i, r_i} \preccurlyeq t_1 \dotsm t_r\) for each \(1 \leq i \leq n\).

  A non-deleting and non-permuting rule is \emph{well-nested} if there do not exist indices \(i, j, k, i', j', k'\) such that
  \(1 \leq i, i' \leq n\), \(i \neq i'\), \(1 \leq j < k \leq r_i\), \(1 \leq j' < k' \leq r_{i'}\),
  and \(x_{i, j} x_{i', j'} x_{i, k} x_{i', k'} \preccurlyeq t_1 \dotsm t_r\).
  An MCFG is \emph{well-nested} if every rule in the grammar is well-nested.
\end{definition}

\begin{example}
  Let \(A\), \(B\), and \(C\) be nonterminals of arity \(2\).
  Then the rules
  \[
    A(x_1 y_1, y_2 x_2) \gets B(x_1, x_2), C(y_1, y_2),
    \qquad
    A(x_1, x_2 y_1 y_2) \gets B(x_1, x_2), C(y_1, y_2)
  \]
  are both well-nested, while the rule
  \[
    A(x_1 y_1, x_2 y_2) \gets B(x_1, x_2), C(y_1, y_2)
  \]
  is not well-nested.
  All three rules are non-deleting and non-permuting.
  Figure~\ref{fig:well-nested-rules} shows the order of the variables in their heads.
\end{example}

\begin{figure}[htb]
  \centering
  \begin{subfigure}[t]{0.31\textwidth}
    \centering
    \begin{tikzpicture}[x=1em, y=1ex, every node/.style={anchor=base, inner sep=0pt}]
      \path[use as bounding box] (-1.8, -6) rectangle (6.5, 2);
      \node[anchor=base east] at (-0.55, 0) {\(A(\)};
      \node at (0, 0) {\(x_1\)};
      \node at (1.6, 0) {\(y_1\)};
      \node at (2.35, 0) {\(,\)};
      \node at (3.6, 0) {\(y_2\)};
      \node at (5.2, 0) {\(x_2\)};
      \node[anchor=base west] at (5.7, 0) {\()\)};
      \draw (0, -0.7) -- (0, -5.5) -- (5.2, -5.5) -- (5.2, -0.7);
      \draw (1.6, -0.7) -- (1.6, -3.3) -- (3.6, -3.3) -- (3.6, -0.7);
    \end{tikzpicture}
    \caption{Well-nested.}
    \label{fig:well-nested-nested}
  \end{subfigure}
  \hfill
  \begin{subfigure}[t]{0.31\textwidth}
    \centering
    \begin{tikzpicture}[x=1em, y=1ex, every node/.style={anchor=base, inner sep=0pt}]
      \path[use as bounding box] (-1.8, -6) rectangle (6.5, 2);
      \node[anchor=base east] at (-0.55, 0) {\(A(\)};
      \node at (0, 0) {\(x_1\)};
      \node at (0.75, 0) {\(,\)};
      \node at (2, 0) {\(x_2\)};
      \node at (3.6, 0) {\(y_1\)};
      \node at (5.2, 0) {\(y_2\)};
      \node[anchor=base west] at (5.7, 0) {\()\)};
      \draw (0, -0.7) -- (0, -3.3) -- (2, -3.3) -- (2, -0.7);
      \draw (3.6, -0.7) -- (3.6, -3.3) -- (5.2, -3.3) -- (5.2, -0.7);
    \end{tikzpicture}
    \caption{Well-nested.}
    \label{fig:well-nested-disjoint}
  \end{subfigure}
  \hfill
  \begin{subfigure}[t]{0.31\textwidth}
    \centering
    \begin{tikzpicture}[x=1em, y=1ex, every node/.style={anchor=base, inner sep=0pt}]
      \path[use as bounding box] (-1.8, -6) rectangle (6.5, 2);
      \node[anchor=base east] at (-0.55, 0) {\(A(\)};
      \node at (0, 0) {\(x_1\)};
      \node at (1.6, 0) {\(y_1\)};
      \node at (2.35, 0) {\(,\)};
      \node at (3.6, 0) {\(x_2\)};
      \node at (5.2, 0) {\(y_2\)};
      \node[anchor=base west] at (5.7, 0) {\()\)};
      \draw (0, -0.7) -- (0, -3.3) -- (3.6, -3.3) -- (3.6, -0.7);
      \draw (1.6, -0.7) -- (1.6, -5.5) -- (5.2, -5.5) -- (5.2, -0.7);
    \end{tikzpicture}
    \caption{Not well-nested.}
    \label{fig:not-well-nested}
  \end{subfigure}
  \caption{Two well-nested rules and a rule that is not well-nested.
  Only the heads are shown; all three rules have body \(B(x_1, x_2), C(y_1, y_2)\).
  The lines under each head join \(x_1\) to \(x_2\) and \(y_1\) to \(y_2\).}
  \label{fig:well-nested-rules}
\end{figure}
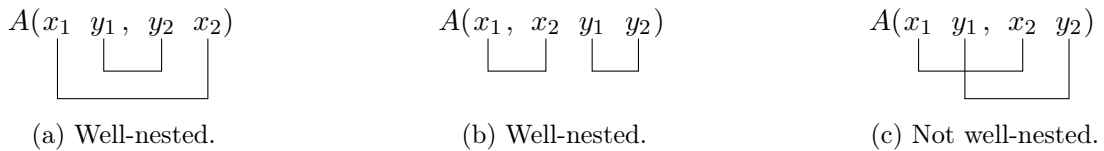

\begin{definition}
  For a rule
  \[
    \pi\colon A(t_1, \dotsc, t_r) \gets B_1(x_{1, 1}, \dotsc, x_{1, r_1}), \dotsc, B_n(x_{n, 1}, \dotsc, x_{n, r_n}),
  \]
  the \emph{terminal contribution} of \(\pi\) is defined by \(\tc(\pi) = \proj_\Sigma(t_1 \dotsm t_r)\).
\end{definition}

\begin{lemma}\label{lem:terminal-contribution-formula}
  Let \(\Sigma\) be an alphabet, let \(d \geq 1\), and let
  \(\psi\colon \Sigma^* \to \Z^d\) be a monoid homomorphism.
  Let
  \[
    \pi\colon
    A(t_1, \dotsc, t_r)
    \gets
    B_1(x_{1, 1}, \dotsc, x_{1, r_1}), \dotsc,
    B_n(x_{n, 1}, \dotsc, x_{n, r_n})
  \]
  be a non-deleting rule of an MCFG over \(\Sigma\).
  Put \(\vv{t} = (t_1, \dotsc, t_r)\) and
  \(\vv*{x}{i} = (x_{i, 1}, \dotsc, x_{i, r_i})\) for \(1 \leq i \leq n\).
  Let
  \[
    \sigma\colon
    (\Sigma \cup \{x_{i, j} \mid 1 \leq i \leq n, 1 \leq j \leq r_i\})^*
    \to \Sigma^*
  \]
  be a monoid homomorphism such that \(\sigma(a) = a\) for \(a \in \Sigma\).
  Then
  \[
    \psi(\concat(\sigma(\vv{t})))
    =
    \psi(\tc(\pi))
    +
    \sum_{i = 1}^n \psi(\concat(\sigma(\vv*{x}{i}))).
  \]
\end{lemma}

\begin{proof}
  Since \(\pi\) is non-deleting, every variable of \(\pi\) occurs exactly once in \(\concat(\vv{t})\).
  The equality follows by applying \(\psi\) and using the commutativity of \(\Z^d\).
\end{proof}

\subsubsection{Derivations and languages generated by MCFGs}

We regard each rule
\[
  \pi\colon
  A(t_1, \dotsc, t_r)
  \gets
  B_1(x_{1, 1}, \dotsc, x_{1, r_1}), \dotsc, B_n(x_{n, 1}, \dotsc, x_{n, r_n})
\]
of branching \(n\) as an \(n\)-ary \emph{function symbol}.
The \emph{terms over \(P\)} are defined inductively: if \(\pi\) has branching \(n\) and \(T_1, \dotsc, T_n\) are terms over \(P\),
then \(\pi(T_1, \dotsc, T_n)\) is a term over \(P\), and there are no other terms over \(P\).
When \(n = 0\), we write \(\pi\) instead of \(\pi()\).
We identify rules up to renaming of variables, that is, up to \(\alpha\)-equivalence.
Before applying a rule, we rename its variables if necessary to keep them distinct from all other variables in the inference.

For \(r \geq 1\), a nonterminal \(A \in N^{(r)}\), a term \(T\) over \(P\), and words \(w_1, \dotsc, w_r \in \Sigma^*\),
an expression of the form
\[
  \vdash_G T : A(w_1, \dotsc, w_r)
\]
is called a \emph{judgment}.
The judgments that hold form the smallest set closed under the following inference schema.
Let
\[
  \pi\colon
  A(t_1, \dotsc, t_r)
  \gets
  B_1(x_{1, 1}, \dotsc, x_{1, r_1}), \dotsc, B_n(x_{n, 1}, \dotsc, x_{n, r_n})
\]
be a rule of \(G\), let \(w_{i, j} \in \Sigma^*\) for each \(1 \leq i \leq n\) and \(1 \leq j \leq r_i\), and let
\[
  \sigma\colon (\Sigma \cup \{x_{i, j} \mid 1 \leq i \leq n, 1 \leq j \leq r_i\})^* \to \Sigma^*
\]
be the monoid homomorphism determined by \(\sigma(a) = a\) for \(a \in \Sigma\) and
\(\sigma(x_{i, j}) = w_{i, j}\) for \(1 \leq i \leq n\) and \(1 \leq j \leq r_i\).
If
\[
  \vdash_G T_i : B_i(w_{i, 1}, \dotsc, w_{i, r_i})
\]
holds for each \(1 \leq i \leq n\), then
\[
  \vdash_G \pi(T_1, \dotsc, T_n) : A(\sigma(t_1), \dotsc, \sigma(t_r))
\]
holds.
If \(n = 0\), then the schema has no premises and \(\vdash_G \pi : A(t_1, \dotsc, t_r)\) holds.
A term \(T\) over \(P\) is a \emph{derivation tree} for the ground atom \(A(w_1, \dotsc, w_r)\) if
\(\vdash_G T : A(w_1, \dotsc, w_r)\).
The ground atom \(A(w_1, \dotsc, w_r)\) is \emph{\(G\)-derivable} if it has a derivation tree.
Equivalently, the tuple \((w_1, \dotsc, w_r)\) is \emph{derivable from \(A\) in \(G\)}
(or \emph{\(A\)-derivable in \(G\)}).
We write
\[
  \vdash_G A(w_1, \dotsc, w_r)
\]
if the ground atom \(A(w_1, \dotsc, w_r)\) is \(G\)-derivable.
The \emph{language generated by \(G\)} is
\[
  L(G) = \{w \in \Sigma^* \mid {\vdash_G S(w)}\}.
\]
We sometimes display a derivation tree as an inference tree, with ground atoms as conclusions and rule symbols beside the inference lines.

\begin{example}\label{ex:neutral-binary-grammar}
  Let \(G_{\mathrm{ex}} = (N, \Sigma_3, P, S)\) be the MCFG over \(\Sigma_3 = \{\tta, \ttb, \ttc\}\)
  such that \(N = \{A, S\}\) with \(A \in N^{(2)}\) and \(S \in N^{(1)}\),
  and \(P\) consists of the following rules:
  \begin{align*}
    \pi_1&\colon A(\tta\ttb, \ttc) \gets,\\
    \pi_2&\colon A(\tta x_1, \ttb x_2 \ttc) \gets A(x_1, x_2),\\
    \pi_3&\colon A(x_1 y_1, y_2 x_2) \gets A(x_1, x_2), A(y_1, y_2),\\
    \pi_4&\colon S(x_1 x_2) \gets A(x_1, x_2).
  \end{align*}
  The grammar \(G_{\mathrm{ex}}\) has dimension \(2\) and branching factor \(2\).
  The grammar \(G_{\mathrm{ex}}\) is well-nested: in the unique binary rule \(\pi_3\), the variables \(y_1, y_2\) occur in the gap between \(x_1\) and \(x_2\).
  For example, the derivation tree \(\pi_4(\pi_2(\pi_3(\pi_2(\pi_1), \pi_3(\pi_1, \pi_1))))\) of \(G_{\mathrm{ex}}\) is depicted as follows:
  \[
    \begin{prooftree}
      \infer0[\(\pi_1\)]{A(\tta\ttb, \ttc)}
      \infer1[\(\pi_2\)]{A(\tta\tta\ttb, \ttb\ttc\ttc)}
      \infer0[\(\pi_1\)]{A(\tta\ttb, \ttc)}
      \infer0[\(\pi_1\)]{A(\tta\ttb, \ttc)}
      \infer2[\(\pi_3\)]{A(\tta\ttb\tta\ttb, \ttc\ttc)}
      \infer2[\(\pi_3\)]{A(\tta\tta\ttb\tta\ttb\tta\ttb, \ttc\ttc\ttb\ttc\ttc)}
      \infer1[\(\pi_2\)]{A(\tta\tta\tta\ttb\tta\ttb\tta\ttb, \ttb\ttc\ttc\ttb\ttc\ttc\ttc)}
      \infer1[\(\pi_4\)]{S(\tta\tta\tta\ttb\tta\ttb\tta\ttb\ttb\ttc\ttc\ttb\ttc\ttc\ttc)\rlap{.}}
    \end{prooftree}
  \]
  Hence the judgment
  \[
    \vdash_{G_{\mathrm{ex}}} \pi_4(\pi_2(\pi_3(\pi_2(\pi_1), \pi_3(\pi_1, \pi_1)))) : S(\tta\tta\tta\ttb\tta\ttb\tta\ttb\ttb\ttc\ttc\ttb\ttc\ttc\ttc)
  \]
  holds and \(\tta\tta\tta\ttb\tta\ttb\tta\ttb\ttb\ttc\ttc\ttb\ttc\ttc\ttc \in L(G_{\mathrm{ex}})\).
  Induction on derivation trees shows that if \((w_1, w_2)\) is derivable from \(A\),
  then \(w_1\) begins with \(\tta\) and ends with \(\ttb\), and \(w_1 w_2 \in \MIX\).
  Hence \(L(G_{\mathrm{ex}}) \subseteq \MIX\).
  For every \(n \geq 1\), the derivation tree \(\pi_4(\pi_2(\dotsm \pi_2(\pi_1) \dotsm))\),
  in which \(\pi_2\) occurs \(n - 1\) times, derives \(\tta^n\ttb^n\ttc^n\).
  Every application of the binary rule \(\pi_3\) creates the factor \(\ttb\tta\)
  and no subsequent rule application removes this factor.
  Thus, a derivation of a word in \(\tta^* \ttb^* \ttc^*\) cannot use \(\pi_3\)
  and hence consists of an application of \(\pi_1\), followed by some number of applications of \(\pi_2\), and finally an application of \(\pi_4\).
  Thus we have \(L(G_{\mathrm{ex}}) \cap \tta^* \ttb^* \ttc^* = \{\tta^n\ttb^n\ttc^n \mid n \geq 1\}\),
  which is not a context-free language
  \cite{hopcroft-motwani-ullman2007}*{Example~7.19}.
  Since the class of context-free languages is closed under intersection with regular languages
  \cite{hopcroft-motwani-ullman2007}*{Theorem~7.27},
  \(L(G_{\mathrm{ex}})\) is not context-free.
\end{example}

For integers \(m, f \geq 1\), let \(\mMCFL{m}(f)\) denote the class of languages generated by MCFGs of dimension at most \(m\) and branching factor at most \(f\).
Similarly, let \(\mMCFLwn{m}(f)\) denote the class of languages generated by well-nested MCFGs of dimension at most \(m\) and branching factor at most \(f\).
Put
\begin{align*}
  \MCFL(f)
  &=
  \bigcup_{m \geq 1} \mMCFL{m}(f),
  &
  \MCFLwn(f)
  &=
  \bigcup_{m \geq 1} \mMCFLwn{m}(f),
  \\
  \mMCFL{m}
  &=
  \bigcup_{f \geq 1} \mMCFL{m}(f),
  &
  \mMCFLwn{m}
  &=
  \bigcup_{f \geq 1} \mMCFLwn{m}(f),
  \\
  \MCFL
  &=
  \bigcup_{m, f \geq 1} \mMCFL{m}(f),
  &
  \MCFLwn
  &=
  \bigcup_{m, f \geq 1} \mMCFLwn{m}(f).
\end{align*}
The correspondence illustrated in Figure~\ref{fig:cfg-to-mcfg} converts every CFG into an equivalent well-nested MCFG of dimension \(1\), hence \(\CFL \subseteq \mMCFLwn{1} \subseteq \mMCFL{1}\), where \(\CFL\) denotes the class of context-free languages.
Conversely, given an MCFG of dimension \(1\), first remove all rules whose bodies contain a nonterminal from which no terminal word is derivable.
For each remaining rule \(A(t) \gets B_1(x_1), \dotsc, B_n(x_n)\), form a CFG rule \(A \to t'\) by replacing each occurrence of \(x_i\) in \(t\) with \(B_i\).
The resulting CFG generates the same language, so \(\mMCFL{1} \subseteq \CFL\).
Thus \(\CFL = \mMCFLwn{1} = \mMCFL{1}\).

Kanazawa--Michaelis--Salvati--Yoshinaka prove that \(\mMCFL{m}(1) = \mMCFLwn{m}(1)\) for every \(m \geq 1\)
\cite{kanazawa-michaelis-salvati-yoshinaka2011}*{Corollary~1}.
Consequently, \(\MCFL(1) = \MCFLwn(1)\).

\subsubsection{Derivation tree contexts and reduced grammars}

We use derivation tree contexts with one hole in two later arguments.
In Lemma~\ref{lem:nonterminal-value}, plugging a derivation from a nonterminal into a context from that nonterminal to the start symbol shows that all tuples derivable from the nonterminal have the same \(\psi\)-value.
In Lemma~\ref{lem:binary-rule-elimination}, a context isolates a subtree whose root is labeled by a binary rule, allowing that subtree to be replaced by another derivation of the same tuple from the same nonterminal.
Let \(H \in N^{(s)}\) be a nonterminal and fix pairwise distinct variables \(y_1, \dotsc, y_s\).
Let \(\Box_H\) be a new nullary function symbol called the \emph{hole}.
For each nonterminal \(A \in N^{(r)}\), each term \(U\) over \(P \cup \{\Box_H\}\),
and words \(\alpha_1, \dotsc, \alpha_r \in (\Sigma \cup \{y_1, \dotsc, y_s\})^*\),
an expression of the form
\[
  H(y_1, \dotsc, y_s) \vdash_G U : A(\alpha_1, \dotsc, \alpha_r)
\]
is called a \emph{hypothetical judgment}.
The hypothetical judgments that hold form the smallest set closed under the following clauses.
First, the judgment
\[
  H(y_1, \dotsc, y_s) \vdash_G \Box_H : H(y_1, \dotsc, y_s)
\]
holds.
Let
\[
  \pi\colon
  A(t_1, \dotsc, t_r)
  \gets
  B_1(x_{1, 1}, \dotsc, x_{1, r_1}), \dotsc, B_n(x_{n, 1}, \dotsc, x_{n, r_n})
\]
be a rule of \(G\), and fix \(k \in \{1, \dotsc, n\}\).
Suppose that the hypothetical judgment
\[
  H(y_1, \dotsc, y_s) \vdash_G U_k : B_k(\beta_{k, 1}, \dotsc, \beta_{k, r_k})
\]
holds for a term \(U_k\) over \(P \cup \{\Box_H\}\),
and words \(\beta_{k, 1}, \dotsc, \beta_{k, r_k} \in (\Sigma \cup \{y_1, \dotsc, y_s\})^*\).
For each \(1 \leq i \leq n\) with \(i \neq k\), choose a term \(T_i\) over \(P\) and words \(w_{i, 1}, \dotsc, w_{i, r_i} \in \Sigma^*\) such that
\[
  \vdash_G T_i : B_i(w_{i, 1}, \dotsc, w_{i, r_i})
\]
holds.
Let
\[
  \sigma\colon (\Sigma \cup \{x_{i, j} \mid 1 \leq i \leq n, 1 \leq j \leq r_i\})^* \to (\Sigma \cup \{y_1, \dotsc, y_s\})^*
\]
be the monoid homomorphism determined by \(\sigma(a) = a\) for \(a \in \Sigma\),
\(\sigma(x_{k, j}) = \beta_{k, j}\) for \(1 \leq j \leq r_k\), and
\(\sigma(x_{i, j}) = w_{i, j}\) for \(1 \leq i \leq n\), \(i \neq k\), and \(1 \leq j \leq r_i\).
Then
\[
  H(y_1, \dotsc, y_s)
  \vdash_G
  \pi(T_1, \dotsc, T_{k - 1}, U_k, T_{k + 1}, \dotsc, T_n)
  :
  A(\sigma(t_1), \dotsc, \sigma(t_r))
\]
holds.
A term \(U\) occurring in such a judgment is a \emph{derivation tree context} from \(H\) to \(A\).
The defining clauses ensure that \(U\) contains exactly one occurrence of \(\Box_H\).
We write
\[
  H(y_1, \dotsc, y_s) \vdash_G A(\alpha_1, \dotsc, \alpha_r)
\]
if \(H(y_1, \dotsc, y_s) \vdash_G U : A(\alpha_1, \dotsc, \alpha_r)\) holds for some derivation tree context \(U\).

If \(H(y_1, \dotsc, y_s) \vdash_G U : A(\alpha_1, \dotsc, \alpha_r)\) and
\(\vdash_G T : H(w_1, \dotsc, w_s)\), define the term \(U[T]\) over \(P\), the result of plugging \(T\) into \(U\), recursively by
\begin{align*}
  \Box_H[T]
  &= T,\\
  \pi(T_1, \dotsc, T_{k - 1}, U_k, T_{k + 1}, \dotsc, T_n)[T]
  &= \pi(T_1, \dotsc, T_{k - 1}, U_k[T], T_{k + 1}, \dotsc, T_n).
\end{align*}

\begin{example}
  Let \(G_{\mathrm{ex}}\) be the MCFG from Example~\ref{ex:neutral-binary-grammar}.
  Then the term \(\pi_3(\pi_1, \pi_2(\Box_A))\) is a derivation tree context, which is depicted as follows:
  \[
    \begin{prooftree}
      \infer0[\(\pi_1\)]{A(\tta\ttb, \ttc)}
      \hypo{A(y_1, y_2)}
      \infer1[\(\pi_2\)]{A(\tta y_1, \ttb y_2 \ttc)}
      \infer2[\(\pi_3\)]{A(\tta\ttb\tta y_1, \ttb y_2 \ttc \ttc)\rlap{.}}
    \end{prooftree}
  \]
  Hence the hypothetical judgment \(A(y_1, y_2) \vdash_{G_{\mathrm{ex}}} \pi_3(\pi_1, \pi_2(\Box_A)) : A(\tta\ttb\tta y_1, \ttb y_2 \ttc \ttc)\) holds.
\end{example}

\begin{lemma}\label{lem:plugging-derivation-context}
  Let \(G = (N, \Sigma, P, S)\) be an MCFG, \(H \in N^{(s)}\), \(A \in N^{(r)}\), and suppose that
  \(H(y_1, \dotsc, y_s) \vdash_G U : A(\alpha_1, \dotsc, \alpha_r)\).
  If \(\vdash_G T : H(w_1, \dotsc, w_s)\), then
  \(\vdash_G U[T] : A(\sigma(\alpha_1), \dotsc, \sigma(\alpha_r))\), where \(\sigma\) is the monoid homomorphism determined by
  \(\sigma(a) = a\) for \(a \in \Sigma\) and \(\sigma(y_i) = w_i\) for \(1 \leq i \leq s\).
\end{lemma}

\begin{proof}
  By induction on the construction of the hypothetical judgment.
\end{proof}

A \emph{subtree} of a term is a subterm at a specified occurrence.
Conversely, suppose that \(\vdash_G T : A(w_1, \dotsc, w_r)\), and choose a subtree \(T'\) of \(T\) such that
\(\vdash_G T' : H(v_1, \dotsc, v_s)\).
Replacing the chosen occurrence of \(T'\) by \(\Box_H\) yields a derivation tree context \(U\) and words
\(\alpha_1, \dotsc, \alpha_r \in (\Sigma \cup \{y_1, \dotsc, y_s\})^*\) such that
\(H(y_1, \dotsc, y_s) \vdash_G U : A(\alpha_1, \dotsc, \alpha_r)\) and \(T = U[T']\).
Moreover, \(w_i = \sigma(\alpha_i)\) for \(1 \leq i \leq r\), where \(\sigma\) is the monoid homomorphism determined by
\(\sigma(a) = a\) for \(a \in \Sigma\) and \(\sigma(y_j) = v_j\) for \(1 \leq j \leq s\).

\begin{lemma}\label{lem:context-variables-once}
  Let \(G\) be an MCFG whose rules are all non-deleting.
  If \(H(y_1, \dotsc, y_s) \vdash_G U : A(\alpha_1, \dotsc, \alpha_r)\), then each of \(y_1, \dotsc, y_s\) occurs exactly once in
  \(\alpha_1 \dotsm \alpha_r\).
\end{lemma}

\begin{proof}
  The claim follows by induction on the construction of the hypothetical judgment.
  The inductive step uses the fact that the surrounding rule is non-deleting.
\end{proof}

\begin{definition}
  Let \(G = (N, \Sigma, P, S)\) be an MCFG\@.
  A nonterminal \(A \in N^{(r)}\) is \emph{productive} if \(\vdash_G A(w_1, \dotsc, w_r)\) for some \(w_1, \dotsc, w_r \in \Sigma^*\).
  The nonterminal \(A\) is \emph{accessible from the start symbol} if
  \(A(y_1, \dotsc, y_r) \vdash_G S(t)\) for pairwise distinct variables \(y_1, \dotsc, y_r\) and some
  \(t \in (\Sigma \cup \{y_1, \dotsc, y_r\})^*\).
  The grammar \(G\) is \emph{reduced} if every nonterminal is productive and accessible from the start symbol.
\end{definition}

\begin{lemma}\label{lem:reduced}
  Let \(G = (N, \Sigma, P, S)\) be an MCFG such that \(L(G) \neq \varnothing\).
  Then there exist \(N' \subseteq N\) and \(P' \subseteq P\) such that \(G' = (N', \Sigma, P', S)\) is reduced and \(L(G') = L(G)\).
  Moreover, if \(G\) is well-nested, then \(G'\) is well-nested.
\end{lemma}

\begin{proof}
  Let \(N'\) be the set of nonterminals of \(G\) that are both productive and accessible from the start symbol.
  Let \(P'\) be the set of rules in \(P\) whose heads and bodies contain only nonterminals in \(N'\).
  Since \(L(G) \neq \varnothing\), we have \(S \in N'\), so \(G' = (N', \Sigma, P', S)\) is an MCFG\@.
  A nonterminal \(A \in N^{(r)}\) belongs to \(N'\) precisely when, for some \(w \in L(G)\), a derivation tree for \(S(w)\) has a subtree that is a derivation tree for \(A(v_1, \dotsc, v_r)\) for some \(v_1, \dotsc, v_r \in \Sigma^*\).
  Thus \(G'\) is reduced and \(L(G') = L(G)\).
  Since \(P' \subseteq P\), \(G'\) is well-nested whenever \(G\) is well-nested.
\end{proof}

\subsubsection{Reduction to binary branching for well-nested MCFGs}\label{sec:binary-branching-reduction-prelim}

The proof of Theorem~\ref{thm:theorem-a} uses the following reduction of the branching factor to at most \(2\).

\begin{lemma}[Reduction to binary branching for well-nested MCFGs]\label{lem:binary-branching-reduction}
  Let \(G = (N, \Sigma, P, S)\) be a well-nested MCFG\@.
  Then there is an equivalent well-nested MCFG \(G'\) with branching factor at most \(2\).
  In particular, \(\MCFLwn = \MCFLwn(2)\).
  If \(L(G)\) is nonempty, \(G'\) can be chosen to be reduced.
\end{lemma}

We give a proof sketch here.
The full proof appears in Appendix~\ref{app:binary-branching}.

\begin{proof}[Proof sketch]
  The following example illustrates how the construction transforms a rule of \(G\) into rules of \(G'\).
  Consider a rule of branching \(3\)
  \[
    \pi\colon
    A(x_1 \tta y_1, y_2 \tta x_2 z)
    \gets
    B_1(x_1, x_2),
    B_2(y_1, y_2),
    B_3(z)
  \]
  where \(\tta \in \Sigma\).
  Introduce fresh nonterminals \(C_2^\pi\) and \(C_3^\pi\), and replace \(\pi\) by the following rules:
  \begin{align*}
    C_2^\pi(x_1, y_1, y_2, x_2)
    &\gets
    B_1(x_1, x_2),
    B_2(y_1, y_2),\\
    C_3^\pi(x_1, y_1, y_2, x_2, z)
    &\gets
    C_2^\pi(x_1, y_1, y_2, x_2),
    B_3(z),\\
    A(x_1 \tta y_1, y_2 \tta x_2 z)
    &\gets
    C_3^\pi(x_1, y_1, y_2, x_2, z).
  \end{align*}
  The first binary rule inserts the components of \(B_2\) into the gap between \(x_1\) and \(x_2\).
  The second binary rule appends the component of \(B_3\) to the tuple produced by \(C_2^\pi\).
  The final unary rule places the variables in the original components and inserts the two occurrences of \(\tta\).
  To replace an arbitrary well-nested rule, add its body atoms one at a time in the order in which their first variables occur in the head terms.
\end{proof}

\begin{remark}
  For well-nested MCFGs, Kanazawa and Salvati show that the branching factor can be reduced to \(2\) without increasing the dimension
  \cite{kanazawa-salvati2010copying}*{Lemma~5}.
  In particular, \(\mMCFLwn{m} = \mMCFLwn{m}(2)\) for every \(m \geq 1\).
  For well-nested LCFRSs, G\'omez-Rodr\'iguez, Kuhlmann, and Satta give a binary normal form that preserves fan-out
  and restricts binary composition operations to concatenation and wrapping
  \cite{gomez-rodriguez2010efficient}*{Sections~3.3 and 4}.
  However, Lemma~\ref{lem:binary-branching-reduction} may increase nonterminal arities and hence the dimension.
\end{remark}

\subsection{Changes of generators, rational equivalence, and dimension bounds}
\label{sec:rational-equivalence}

Let \(\pi\colon \Sigma^* \to G\) and \(\rho\colon \Delta^* \to G\) be two choices of generators
for a finitely generated group \(G\).
For each \(a \in \Sigma\), choose \(u_a \in \Delta^*\) such that \(\rho(u_a) = \pi(a)\),
and for each \(b \in \Delta\), choose \(v_b \in \Sigma^*\) such that \(\pi(v_b) = \rho(b)\).
These choices induce monoid homomorphisms
\(h_1\colon \Sigma^* \to \Delta^*\) and \(h_2\colon \Delta^* \to \Sigma^*\),
defined by \(h_1(a) = u_a\) for \(a \in \Sigma\) and \(h_2(b) = v_b\) for \(b \in \Delta\).
They satisfy \(\rho \circ h_1 = \pi\) and \(\pi \circ h_2 = \rho\).
Therefore,
\[
  \pi^{-1}(1_G) = h_1^{-1}(\rho^{-1}(1_G)),
  \qquad
  \rho^{-1}(1_G) = h_2^{-1}(\pi^{-1}(1_G)).
\]
Although the word problem depends on the choice of generators,
whether it belongs to a class \(\mathcal{C}\) of languages is independent of that choice
whenever \(\mathcal{C}\) is closed under inverse homomorphisms
\cite{holt-rees-rover2017}*{3.4.3 Proposition and the subsequent paragraph}.

Taking \(\pi = \psin{n + 1}\) and \(\rho = \varphi_n\) in this construction yields monoid homomorphisms
\(h_1\colon \Sigman{n + 1}^* \to \Delta_n^*\) and
\(h_2\colon \Delta_n^* \to \Sigman{n + 1}^*\) such that
\[
  \MIX_{n + 1} = h_1^{-1}(O_n),
  \qquad
  O_n = h_2^{-1}(\MIX_{n + 1}).
\]
For \(n = 2\), one may take
\(h_1\colon \Sigman{3}^* \to \Delta_2^*\) to be the monoid homomorphism defined by
\(h_1(\tta) = a_1\), \(h_1(\ttb) = a_2\), and \(h_1(\ttc) = \bar{a}_1\bar{a}_2\).
Indeed, \(\varphi_2 \circ h_1 = \psin{3}\), and hence
\(\psin{3}^{-1}(\bm{0}) = h_1^{-1}(\varphi_2^{-1}(\bm{0}))\).
The homomorphism \(h_1\) maps the word in
Figure~\ref{fig:mix-o2-closed-walks}(\subref{fig:mix-closed-walk})
to the word in Figure~\ref{fig:mix-o2-closed-walks}(\subref{fig:o2-closed-walk}).

A rational transduction is a relation between words realized by a transducer
\cite{berstel1979transductions}*{Chapter~III}.
Two languages are \emph{rationally equivalent} if each is the image of the other under a rational transduction.
Since inverse homomorphisms are rational transductions,
the two equalities above show that the word problems of a finitely generated group
with respect to any two choices of generators are rationally equivalent
\cite{berstel1979transductions}*{Chapter~III, especially Example~5.3}.
In particular, \(\MIX_{n + 1}\) and \(O_n\) are rationally equivalent.
For every \(m \geq 1\), the class \(\mMCFL{m}\) is closed under rational transductions
\cite{seki1991}*{Theorem~3.9}.
For every \(m \geq 1\), the class \(\mMCFLwn{m}\) is also closed under rational transductions
\cite{seki-kato2008}*{Theorem~15}.
Thus, for all \(m, n \geq 1\),
\begin{align*}
  \MIX_{n + 1} \in \mMCFL{m}
  &\iff O_n \in \mMCFL{m},\\
  \MIX_{n + 1} \in \mMCFLwn{m}
  &\iff O_n \in \mMCFLwn{m}.
\end{align*}
Taking the union over \(m\) yields
\begin{equation}\label{eq:mix-origin-crossing-wn-equivalence}
  \MIX_{n + 1} \in \MCFLwn
  \iff
  O_n \in \MCFLwn
\end{equation}
for every \(n \geq 1\).
For \(n = 2\), Theorem~\ref{thm:kanazawa-salvati-dimension-two}
and the rational equivalence above imply \(O_2 \notin \TAL = \mMCFLwn{2}\).
The Kanazawa--Salvati conjecture is equivalent to \(O_2 \notin \MCFLwn\).
The class \(\IL\) of indexed languages is also closed under rational transductions
\cite{aho1968indexed}*{Lemma~3.2}.
Hence Marsh's conjecture is equivalent to \(O_2 \notin \IL\).

For every \(n \geq 2\), the following theorem determines the smallest dimension
of an MCFG generating \(O_n\).

\begin{theorem}\label{thm:origin-crossing-dimension}
  \leavevmode
  \begin{enumerate}[leftmargin=*]
    \item \(O_n \notin \mMCFL{(n - 1)}\) for every \(n \geq 2\)
      \cite{nederhof2017}*{Theorem~1}.
    \item \(O_n \in \mMCFL{n}\) for every \(n \geq 1\)
      \cite{gebhardt-meunier-salvati2022}. \qed
  \end{enumerate}
\end{theorem}

Together with the rational equivalence of \(\MIX_{n + 1}\) and \(O_n\),
Theorem~\ref{thm:origin-crossing-dimension} shows that
\(O_n, \MIX_{n + 1} \in \mMCFL{n} \setminus \mMCFL{(n - 1)}\) for every \(n \geq 2\).

\section{Proof of Theorem A}\label{sec:proof-theorem-a}

Throughout this section, fix a finite alphabet \(\Sigma\), an integer \(d \geq 1\),
and a surjective monoid homomorphism \(\psi\colon \Sigma^* \to \Z^d\).

When \(L(G) \neq \varnothing\), the proof of Theorem~\ref{thm:theorem-a} has two steps.
First, Lemma~\ref{lem:binary-branching-reduction} and the Neutralization Theorem~\ref{thm:neutralization} yield an equivalent reduced well-nested MCFG
with branching factor at most \(2\) whose nonterminals are neutral (Corollary~\ref{cor:neutral-binary-normal-form}).
Second, Theorem~\ref{thm:simulation} shows that every such grammar is simulated by \(\Gpsi{r}\) for some \(r \geq 1\).

\begin{definition}
  Let \(G\) be a well-nested MCFG over \(\Sigma\).
  A nonterminal \(A\) of \(G\) with arity \(r\) is \emph{neutral} if, for every
  \(\vv{u} \in (\Sigma^*)^r\),
  \[
    \vdash_G A(\vv{u})
    \implies
    \psi(\concat(\vv{u})) = \bm{0}.
  \]
\end{definition}

\subsection{Neutralization Theorem}\label{sec:nonterminal-values-neutralization}

The following lemma is a direct generalization of
\cite{kanazawa-salvati2012}*{Lemma~2}.

\begin{lemma}\label{lem:nonterminal-value}
  Let \(G = (N, \Sigma, P, S)\) be a reduced well-nested MCFG such that \(L(G) \subseteq \Wpsi\).
  Then for every nonterminal \(A \in N\) of arity \(r\) there exists a vector \(\bm{v}_A \in \Z^d\) such that, for every
  \(\vv{w} \in (\Sigma^*)^r\),
  \[
    \vdash_G A(\vv{w})
    \implies
    \psi(\concat(\vv{w})) = \bm{v}_A.
  \]
\end{lemma}

\begin{proof}
  Let \(A\) be a nonterminal of arity \(r\).
  Since \(G\) is reduced, \(A\) is accessible from the start symbol.
  Choose a derivation tree context \(U\) such that
  \(A(y_1, \dotsc, y_r) \vdash_G U : S(t)\) for some \(t \in (\Sigma \cup \{y_1, \dotsc, y_r\})^*\).
  By Lemma~\ref{lem:context-variables-once}, each of \(y_1, \dotsc, y_r\) occurs exactly once in \(t\).
  Put \(\bm{v}_A = -\psi(\proj_\Sigma(t))\).
  Suppose \(\vdash_G A(\vv{w})\), where \(\vv{w} = (w_1, \dotsc, w_r)\),
  and choose a derivation tree \(T\) such that \(\vdash_G T : A(\vv{w})\).
  Let \(\sigma\) be the monoid homomorphism determined by \(\sigma(a) = a\) for \(a \in \Sigma\) and
  \(\sigma(y_i) = w_i\) for \(1 \leq i \leq r\).
  By Lemma~\ref{lem:plugging-derivation-context}, \(\vdash_G U[T] : S(\sigma(t))\), and hence \(\sigma(t) \in L(G) \subseteq \Wpsi\).
  Since \(\Z^d\) is abelian and each variable \(y_i\) occurs exactly once in \(t\),
  we have \(\bm{0} = \psi(\sigma(t)) = \psi(\concat(\vv{w})) + \psi(\proj_\Sigma(t))\).
  Thus \(\psi(\concat(\vv{w})) = -\psi(\proj_\Sigma(t)) = \bm{v}_A\).
\end{proof}

\begin{remark}
  If \(G\) is a reduced well-nested MCFG such that \(L(G) \subseteq \Wpsi\), then a nonterminal \(A\) of \(G\) is neutral if and only if the vector \(\bm{v}_A\) in Lemma~\ref{lem:nonterminal-value} satisfies \(\bm{v}_A = \bm{0}\).
\end{remark}

\begin{theorem}[Neutralization Theorem]\label{thm:neutralization}
  Let \(G = (N, \Sigma, P, S)\) be a reduced well-nested MCFG such that \(L(G) \subseteq \Wpsi\).
  Then there exists a well-nested MCFG \(G' = (N', \Sigma, P', S')\) such that \(L(G') = L(G)\),
  and every nonterminal of \(G'\) is neutral.
  Moreover, the branching factor of \(G'\) is at most the branching factor of \(G\).
\end{theorem}

We give a proof sketch here.
The full proof appears in Appendix~\ref{app:neutralization}.

\begin{proof}[Proof sketch]
  The following example illustrates how the construction transforms a derivation of \(G\) into a derivation of \(G'\).
  Let \(G = (N, \Sigman{3}, P, S)\), where \(N = \{A, B, S\}\) and \(P\) consists of the following rules:
  \begin{align*}
    \pi_1&\colon A(\tta, \ttb) \gets,\\
    \pi_2&\colon A(\ttb, \tta) \gets,\\
    \pi_3&\colon A(x_1 y_1, \ttc y_2 x_2)
    \gets
    A(x_1, x_2),
    A(y_1, y_2),\\
    \pi_4&\colon B(\tta\ttb\ttb\ttc) \gets,\\
    \pi_5&\colon A(y x_1, \tta x_2 \ttc)
    \gets
    B(y),
    A(x_1, x_2),\\
    \pi_6&\colon S(x_1 \ttc x_2) \gets A(x_1, x_2).
  \end{align*}
  Then \(G\) is a reduced well-nested MCFG\@.
  It has dimension \(2\) and branching factor \(2\), and satisfies \(L(G) \subseteq \MIX\).
  The fixed values from Lemma~\ref{lem:nonterminal-value} are
  \[
    \bm{v}_A = \psin{3}(\tta\ttb) = (1, 1), \qquad
    \bm{v}_B = \psin{3}(\ttb) = (0, 1), \qquad
    \bm{v}_S = \psin{3}(\varepsilon) = \bm{0}.
  \]

  For example, \(G\) has the following derivation:
  \[
    \begin{prooftree}
      \infer0[\(\pi_4\)]{B(\tta\ttb\ttb\ttc)}
      \infer0[\(\pi_1\)]{A(\tta, \ttb)}
      \infer0[\(\pi_2\)]{A(\ttb, \tta)}
      \infer2[\(\pi_3\)]{A(\tta\ttb, \ttc\tta\ttb)}
      \infer2[\(\pi_5\)]{A(\tta\ttb\ttb\ttc\tta\ttb, \tta\ttc\tta\ttb\ttc)}
      \infer1[\(\pi_6\)]{S(\tta\ttb\ttb\ttc\tta\ttb\ttc\tta\ttc\tta\ttb\ttc)\rlap{.}}
    \end{prooftree}
  \]
  For each tuple derived from a nonterminal \(C \in N\), the neutralization construction chooses occurrences forming a subsequence \(u\) such that
  \(\psin{3}(u) = \bm{v}_C\) and removes those occurrences.
  The word \(u\) is stored in the corresponding new nonterminal of \(G'\).
  The construction transforms the derivation above into the following informal derivation:
  \[
    \begin{prooftree}
      \infer0{B^{\ttb}(\tta \placeholder \ttb\ttc)}
      \infer0{A^{\tta\ttb}(\placeholder, \placeholder)}
      \infer0{A^{\ttb\tta}(\placeholder, \placeholder)}
      \infer2{A^{\tta\ttb}(\placeholder \placeholder, \ttc\tta\ttb)}
      \infer2{A^{\ttb\tta}(\tta \placeholder \ttb\ttc \placeholder \ttb, \tta\ttc\tta\ttb\ttc)}
      \infer1{S^\varepsilon(\tta\ttb\ttb\ttc\tta\ttb\ttc\tta\ttc\tta\ttb\ttc)\rlap{.}}
    \end{prooftree}
  \]

  Here the symbol \(\placeholder\) indicates a \emph{placeholder} where the letters of the removed representative,
  shown in the superscript, are to be reinserted.
  To turn this informal derivation into a valid derivation in a well-nested MCFG,
  we replace the placeholders by commas and record whether each resulting comma is inherited from the original tuple
  or comes from a placeholder.
  The resulting derivation is as follows:
  \[
    \begin{prooftree}
      \infer0{B_{\placeholder}^{\ttb}(\tta, \ttb\ttc)}
      \infer0{A_{\placeholder \separator \placeholder}^{\tta\ttb}(\varepsilon, \varepsilon, \varepsilon, \varepsilon)}
      \infer0{A_{\placeholder \separator \placeholder}^{\ttb\tta}(\varepsilon, \varepsilon, \varepsilon, \varepsilon)}
      \infer2{A_{\placeholder \placeholder \separator}^{\tta\ttb}(\varepsilon, \varepsilon, \varepsilon, \ttc\tta\ttb)}
      \infer2{A_{\placeholder \placeholder \separator}^{\ttb\tta}(\tta, \ttb\ttc, \ttb, \tta\ttc\tta\ttb\ttc)}
      \infer1{S_\varepsilon^\varepsilon(\tta\ttb\ttb\ttc\tta\ttb\ttc\tta\ttc\tta\ttb\ttc)\rlap{.}}
    \end{prooftree}
  \]
  Note that every tuple occurring in this derivation has \(\psin{3}\)-value \(\bm{0}\).
  This illustrates how the construction makes the new nonterminals in \(G'\) neutral.
  Finally, the finiteness of the nonterminal set of \(G'\) follows from Higman's lemma (Lemma~\ref{lem:higman}).
\end{proof}

\begin{corollary}[Neutral binary normal form]\label{cor:neutral-binary-normal-form}
  Let \(G\) be a well-nested MCFG over \(\Sigma\) such that \(L(G) \neq \varnothing\) and \(L(G) \subseteq \Wpsi\).
  Then there exists an equivalent reduced well-nested MCFG \(G'\) with branching factor at most \(2\) such that every nonterminal of \(G'\) is neutral.
  Moreover, if \(G\) is non-branching, then \(G'\) can be chosen to be non-branching.
\end{corollary}

\begin{proof}
  Apply Lemma~\ref{lem:binary-branching-reduction} to \(G\), apply Theorem~\ref{thm:neutralization} to the resulting grammar,
  and finally apply Lemma~\ref{lem:reduced}.
  Denote the grammar obtained after these three steps by \(G'\).
  The construction in Theorem~\ref{thm:neutralization} and the subsequent deletion of nonterminals and rules do not increase the branching factor.

  If \(G\) is non-branching, replace the first step by an application of Lemma~\ref{lem:reduced} to \(G\).
  The remaining operations preserve the non-branching property.
\end{proof}

\begin{lemma}\label{lem:neutral-contribution}
  Let \(G\) be a reduced well-nested MCFG over \(\Sigma\) such that every nonterminal of \(G\) is neutral.
  Then for every rule \(\pi\) of \(G\), the terminal contribution \(\tc(\pi)\) belongs to \(\Wpsi\).
\end{lemma}

\begin{proof}
  Let
  \[
    \pi\colon A(t_1, \dotsc, t_m) \gets B_1(\vv*{x}{1}), \dotsc, B_n(\vv*{x}{n})
  \]
  be a rule of \(G\), where \(n \geq 0\).
  Since \(G\) is reduced, applying \(\pi\) to one \(B_i\)-derivable tuple for each \(1 \leq i \leq n\) yields an \(A\)-derivable tuple.
  It follows from Lemma~\ref{lem:terminal-contribution-formula} and the neutrality of \(A\) and the \(B_i\) that \(\psi(\tc(\pi)) = \bm{0}\).
  Hence \(\tc(\pi) \in \Wpsi\).
\end{proof}

\subsection{Simulation Theorem}\label{sec:simulation}

We define the grammars \(\Gpsi{r}\) used in Theorem~\ref{thm:simulation}.
By Lemma~\ref{lem:neutral-contribution}, the terminal contribution of every rule of a reduced grammar whose nonterminals are neutral belongs to \(\Wpsi\).
We first partition the set of positions of each nonempty terminal contribution so that each part induces a subsequence in a fixed finite set \(M_\psi\),
and then use these subsequences as the terminal contributions of letter insertion rules in \(\Gpsi{r}\).

\begin{definition}
  Let
  \[
    M_\psi = \Minsubseq(\Wpsi \setminus \{\varepsilon\}).
  \]
  An element of \(M_\psi\) is called an \emph{\(M_\psi\)-word}.
  Equivalently, a word \(u \in \Sigma^*\) is an \(M_\psi\)-word if and only if
  \(u \neq \varepsilon\), \(\psi(u) = \bm{0}\), and there is no nonempty word
  \(v \prec u\) such that \(\psi(v) = \bm{0}\).
\end{definition}

By Corollary~\ref{cor:subsequence-minimal-finite}, \(M_\psi\) is finite.

\begin{definition}
  Let \(w \in \Sigma^*\), and put \(n = \len{w}\).
  An \emph{\(M_\psi\)-decomposition} of \(w\) is a partition \([n] = J_1 \sqcup \dotsb \sqcup J_k\)
  for some \(k \in \N\) such that \(w[J_i] \in M_\psi\) for \(1 \leq i \leq k\).
  Thus the empty partition is an \(M_\psi\)-decomposition of \(\varepsilon\).
\end{definition}

\begin{lemma}\label{lem:mpsi-decomposition}
  Every word in \(\Wpsi\) has an \(M_\psi\)-decomposition.
\end{lemma}

\begin{proof}
  Let \(w \in \Wpsi\), and put \(n = \len{w}\).
  We prove the following stronger statement by induction on \(\card{J}\):
  for every \(J \subseteq [n]\) such that \(w[J] \in \Wpsi\),
  there is a partition \(J = J_1 \sqcup \dotsb \sqcup J_k\)
  for some \(k \in \N\) such that \(w[J_i] \in M_\psi\) for \(1 \leq i \leq k\).

  If \(J = \varnothing\), take \(k = 0\).
  Suppose that \(J \neq \varnothing\).
  Since \(w[J] \in \Wpsi\), choose a minimal nonempty subset \(J_0 \subseteq J\)
  such that \(w[J_0] \in \Wpsi\).
  By minimality, \(w[J_0] \in M_\psi\).
  Indeed, if a nonempty word \(v \prec w[J_0]\) belonged to \(\Wpsi\), then
  \(v = w[J'']\) for some nonempty proper subset \(J'' \subsetneq J_0\), contradicting the choice of \(J_0\).

  Put \(J' = J \setminus J_0\).
  Since \(\Z^d\) is abelian and \(J = J_0 \sqcup J'\),
  we have \(\psi(w[J]) = \psi(w[J_0]) + \psi(w[J'])\),
  and hence \(w[J'] \in \Wpsi\).
  Since \(\card{J'} < \card{J}\), it follows from the induction hypothesis that there exists a partition
  \(J' = J'_1 \sqcup \dotsb \sqcup J'_k\) such that \(w[J'_i] \in M_\psi\) for \(1 \leq i \leq k\).
  Therefore \(J = J_0 \sqcup J'_1 \sqcup \dotsb \sqcup J'_k\) is the required partition.

  Taking \(J = [n]\) proves the lemma.
\end{proof}

\begin{lemma}\label{lem:psi-n-mpsi}
  For each \(n \geq 2\), \(M_{\psin{n}} = \Perm(\tta_1 \dotsm \tta_n)\).
\end{lemma}

\begin{proof}
  By the definition of \(\psin{n}\), for every \(w \in \Sigman{n}^*\),
  \[
    w \in \psin{n}^{-1}(\bm{0})
    \iff
    \len{w}_{\tta_1} = \dotsb = \len{w}_{\tta_n}.
  \]
  Hence both \(M_{\psin{n}}\) and \(\Perm(\tta_1 \dotsm \tta_n)\) consist precisely of the words
  \(w \in \Sigman{n}^*\) such that \(\len{w}_{\tta_1} = \dotsb = \len{w}_{\tta_n} = 1\).
\end{proof}

\begin{example}\label{ex:mpsi-words}
  The set \(M_{\psin{3}}\) consists of the six permutations of \(\tta\ttb\ttc\), namely,
  \[
    M_{\psin{3}}
    =
    \{\tta\ttb\ttc, \tta\ttc\ttb, \ttb\tta\ttc, \ttb\ttc\tta, \ttc\tta\ttb, \ttc\ttb\tta\}.
  \]
  For the word \(w_1 = \tta\ttb\ttb\ttc\ttc\tta \in \MIX = W_{\psin{3}}\),
  the partition \([6] = \{1, 2, 4\} \sqcup \{3, 5, 6\}\) is an \(M_{\psin{3}}\)-decomposition of \(w_1\).
  Indeed, \(w_1[\{1, 2, 4\}] = \tta\ttb\ttc \in M_{\psin{3}}\) and \(w_1[\{3, 5, 6\}] = \ttb\ttc\tta \in M_{\psin{3}}\).

  For \(\varphi_2\), we have
  \[
    M_{\varphi_2}
    =
    \{a_1 \bar{a}_1, \bar{a}_1 a_1, a_2 \bar{a}_2, \bar{a}_2 a_2\}.
  \]
  For the word \(w_2 = a_1 a_2 \bar{a}_1 \bar{a}_2 \in W_{\varphi_2}\),
  the partition \([4] = \{1, 3\} \sqcup \{2, 4\}\) is an \(M_{\varphi_2}\)-decomposition of \(w_2\).
  Indeed, \(w_2[\{1, 3\}] = a_1 \bar{a}_1 \in M_{\varphi_2}\) and \(w_2[\{2, 4\}] = a_2 \bar{a}_2 \in M_{\varphi_2}\).
\end{example}

\subsubsection{The grammars
  \texorpdfstring{\(\Gnbpsi{r}\) and \(\Gpsi{r}\)}{Gpsi-nb[r] and Gpsi[r]}}

\begin{definition}
  Fix \(r \geq 1\).
  Let \(I\) and \(S\) be distinct nonterminals of arities \(r\) and \(1\), respectively.
  We use the following rule schemata.
  \begin{description}[style=unboxed,leftmargin=0pt]
    \item[Nullary rule] \(I(\varepsilon, \dotsc, \varepsilon) \gets\).
    \item[Letter insertion rules] for \(u_1, v_1, \dotsc, u_r, v_r \in \Sigma \cup \{\varepsilon\}\) satisfying
      \(u_1 v_1 \dotsm u_r v_r \in M_\psi\), include the rule
      \[
        I(u_1 x_1 v_1, \dotsc, u_r x_r v_r) \gets I(x_1, \dotsc, x_r).
      \]
      We call a rule of this form a \emph{letter insertion rule}.
    \item[Unary regrouping rules] for every \(r\)-tuple \(\vv{t} = (t_1, \dotsc, t_r)\) satisfying
      \[
        \vv{t} \sqsubseteq (x_1, \dotsc, x_r),
      \]
      include the rule
      \[
        I(\vv{t}) \gets I(x_1, \dotsc, x_r).
      \]
      We call a rule of this form a \emph{unary regrouping rule}.
      Since \(\sqsubseteq\) is reflexive, this schema includes the identity rule, whose head coincides with its sole body atom.
    \item[Binary rules] for every \(0 \leq i \leq r\) and every \(r\)-tuple \(\vv{t} = (t_1, \dotsc, t_r)\) satisfying
      \[
        \vv{t}
        \sqsubseteq
        (x_1, \dotsc, x_i, y_1, \dotsc, y_r, x_{i + 1}, \dotsc, x_r),
      \]
      include the rule
      \[
        I(\vv{t}) \gets I(x_1, \dotsc, x_r),\ I(y_1, \dotsc, y_r).
      \]
    \item[Start rule] \(S(x_1 \dotsm x_r) \gets I(x_1, \dotsc, x_r)\).
  \end{description}

  Let \(\Gnbpsi{r}\) be the non-branching MCFG over \(\Sigma\) with nonterminal set \(\{I, S\}\), start symbol \(S\),
  and rule set consisting of the nullary rule, the letter insertion rules,
  the unary regrouping rules, and the start rule.
  Let \(\Gpsi{r}\) be the grammar obtained from \(\Gnbpsi{r}\) by adding the binary rules.
  Thus \(\Gnbpsi{r}\) is the non-branching fragment of \(\Gpsi{r}\).

  For \(n \geq 2\), put \(\Gnbn{n}{r} = G_{\psin{n}}^{\mathrm{nb}}[r]\) and \(\Gn{n}{r} = G_{\psin{n}}[r]\).
\end{definition}

A letter insertion rule inserts the letters of an \(M_\psi\)-word immediately before or after the components.
A unary regrouping rule preserves the concatenation of the components and may change their boundaries.
A binary rule inserts the components of one premise tuple into a gap of another and then regroups the resulting sequence.
In \(\Gpsi{r}\), every application of a unary regrouping rule can be simulated using a binary rule and the nullary rule.
The unary regrouping rules allow the same regrouping in \(\Gnbpsi{r}\), which has no binary rules.

By Lemma~\ref{lem:psi-n-mpsi}, for \(n \geq 2\), the terminal contribution of each letter insertion rule in
\(\Gnbn{n}{r}\) and \(\Gn{n}{r}\) is a permutation of \(\tta_1 \dotsm \tta_n\).
Thus each such rule inserts every letter of \(\Sigman{n}\) exactly once,
and at most one letter is inserted on either side of each component.

\subsubsection{Basic properties and simulation}

\begin{lemma}\label{lem:gpsi-properties}
  For every \(r \geq 1\), the grammars \(\Gnbpsi{r}\) and \(\Gpsi{r}\) are well-nested MCFGs.
  Moreover, the nonterminals \(I\) and \(S\) are neutral in both grammars.
  Thus, \(L(\Gnbpsi{r}) \subseteq L(\Gpsi{r}) \subseteq \Wpsi\).
  For \(1 \leq r \leq s\), we also have \(L(\Gnbpsi{r}) \subseteq L(\Gnbpsi{s})\) and \(L(\Gpsi{r}) \subseteq L(\Gpsi{s})\).
\end{lemma}

\begin{proof}
  For fixed \(r\), there are finitely many ways to regroup \(r\) or \(2r\) variables into \(r\) components, and a binary rule has one of \(r + 1\) insertion positions.
  Thus the unary regrouping and binary rule schemata yield finitely many rules.
  There are also finitely many letter insertion rules because each \(u_i\) and \(v_i\) belongs to the finite set \(\Sigma \cup \{\varepsilon\}\).
  Since there is only one nullary rule and one start rule, both grammars have finite rule sets.

  The nullary rule, the letter insertion rules, the unary regrouping rules, and the start rule are non-deleting and non-permuting and have at most one body atom.
  Every binary rule is also non-deleting and non-permuting, and the variables of the second body atom occur in a single gap among those of the first.
  Hence both grammars are well-nested.

  Every letter insertion rule has terminal contribution in \(M_\psi \subseteq \Wpsi\), and every other rule has empty terminal contribution.
  Induction on derivations therefore shows that \(I\) is neutral.
  The start rule then shows that \(S\) is neutral.
  Since \(\Gnbpsi{r}\) is a subgrammar of \(\Gpsi{r}\), it follows that \(L(\Gnbpsi{r}) \subseteq L(\Gpsi{r}) \subseteq \Wpsi\).

  Let \(1 \leq r \leq s\).
  Pad every \(I\)-tuple in a derivation with \(s - r\) empty components.
  The nullary rule, each letter insertion rule, and the start rule lift under this padding by assigning \(\varepsilon\) to every added position.
  A unary regrouping rule \(I(t_1, \dotsc, t_r) \gets I(x_1, \dotsc, x_r)\) lifts to
  \[
    I(t_1, \dotsc, t_r, x_{r + 1}, \dotsc, x_s)
    \gets
    I(x_1, \dotsc, x_s),
  \]
  which is a unary regrouping rule since \((t_1, \dotsc, t_r, x_{r + 1}, \dotsc, x_s) \sqsubseteq (x_1, \dotsc, x_s)\).
  Hence \(L(\Gnbpsi{r}) \subseteq L(\Gnbpsi{s})\).

  The same padding construction applies to \(\Gpsi{r}\).
  In the head tuple of each binary rule, replace \(y_r\) and \(x_r\) by \(y_r y_{r + 1} \dotsm y_s\) and \(x_r x_{r + 1} \dotsm x_s\), respectively, and append \(s - r\) empty components.
  The resulting rule is a binary rule of \(\Gpsi{s}\), using the same insertion position \(i\) if \(i < r\) and insertion position \(s\) if \(i = r\).
  Hence \(L(\Gpsi{r}) \subseteq L(\Gpsi{s})\).
\end{proof}

\begin{theorem}[Simulation Theorem]\label{thm:simulation}
  Let \(G = (N, \Sigma, P, S)\) be a reduced well-nested MCFG with branching factor at most \(2\) such that every nonterminal of \(G\) is neutral.
  Let \(m\) be the dimension of \(G\), and put \(c = \max_{\pi \in P}\len{\tc(\pi)}\) and \(r = m + c\).
  Then \(L(G) \subseteq L(\Gpsi{r}) \subseteq \Wpsi\).
  Moreover, if \(G\) is non-branching, then \(L(G) \subseteq L(\Gnbpsi{r}) \subseteq \Wpsi\).
\end{theorem}

We give a proof sketch here.
The full proof appears in Appendix~\ref{app:simulation}.

\begin{proof}[Proof sketch]
  The following example illustrates the simulation of \(G\)-derivations by \(\Gn{2}{r}\) when \(\psi = \psin{2}\).
  Let \(G = (N, \Sigma_2, P, S)\), where \(N = \{A, S\}\) and \(P\) consists of the following rules:
  \begin{align*}
    \pi_1&\colon A(\tta\tta, \ttb\ttb) \gets,\\
    \pi_2&\colon A(x_1 y_1 \ttb \tta y_2, \tta x_2 \ttb) \gets A(x_1, x_2), A(y_1, y_2),\\
    \pi_3&\colon S(\ttb x_1 \tta x_2) \gets A(x_1, x_2).
  \end{align*}
  Then \(G\) is a reduced well-nested MCFG\@.
  It has dimension \(2\) and branching factor \(2\), and its terminal contributions are
  \(\tc(\pi_1) = \tta\tta\ttb\ttb\), \(\tc(\pi_2) = \ttb\tta\tta\ttb\), and \(\tc(\pi_3) = \ttb\tta\).
  Since these words belong to \(\MIX_2\), induction on derivations shows that both nonterminals are neutral, and thus \(L(G) \subseteq \MIX_2\).
  Here \(c = 4\), so the value of \(r = m + c\) in Theorem~\ref{thm:simulation} is \(6\).

  For example, \(G\) has the following derivation:
  \begin{equation}\label{eq:simulation-example-derivation}
    \begin{prooftree}
      \infer0[\(\pi_1\)]{A(\tta\tta, \ttb\ttb)}
      \infer0[\(\pi_1\)]{A(\tta\tta, \ttb\ttb)}
      \infer2[\(\pi_2\)]{A(\tta\tta\tta\tta\ttb\tta\ttb\ttb, \tta\ttb\ttb\ttb)}
      \infer1[\(\pi_3\)]{S(\ttb\tta\tta\tta\tta\ttb\tta\ttb\ttb\tta\tta\ttb\ttb\ttb)\rlap{.}}
    \end{prooftree}
  \end{equation}
  This derivation is simulated by \(\Gn{2}{6}\) as shown in Figure~\ref{fig:simulation-example}.
  General \(G\)-derivations are simulated in the same way.
  When \(G\) is non-branching, the simulation uses only rules of \(\Gnbpsi{r}\).
  \begin{figure}[htbp]
    \centering
    \[
      \begin{prooftree}
        \infer0[nullary rule]{I(\varepsilon, \varepsilon, \varepsilon, \varepsilon, \varepsilon, \varepsilon)}
        \infer1[letter insertion rule]{I(\underline{\tta}, \varepsilon, \varepsilon, \underline{\ttb}, \varepsilon, \varepsilon)}
        \infer1[letter insertion rule]{I(\tta, \underline{\tta}, \varepsilon, \ttb, \underline{\ttb}, \varepsilon)}
        \infer1[unary regrouping rule]{I(\tta\tta, \ttb\ttb, \varepsilon, \varepsilon, \varepsilon, \varepsilon)}
        \infer0[nullary rule]{I(\varepsilon, \varepsilon, \varepsilon, \varepsilon, \varepsilon, \varepsilon)}
        \infer1[letter insertion rule]{I(\underline{\tta}, \varepsilon, \varepsilon, \underline{\ttb}, \varepsilon, \varepsilon)}
        \infer1[letter insertion rule]{I(\tta, \underline{\tta}, \varepsilon, \ttb, \underline{\ttb}, \varepsilon)}
        \infer1[unary regrouping rule]{I(\tta\tta, \ttb\ttb, \varepsilon, \varepsilon, \varepsilon, \varepsilon)}
        \infer2[binary rule]{I(\tta\tta\tta\tta, \varepsilon, \ttb\ttb, \varepsilon, \ttb\ttb, \varepsilon)}
        \infer1[letter insertion rule]{I(\tta\tta\tta\tta\underline{\ttb}, \underline{\tta}, \ttb\ttb, \varepsilon, \ttb\ttb, \varepsilon)}
        \infer1[letter insertion rule]{I(\tta\tta\tta\tta\ttb, \tta, \ttb\ttb, \underline{\tta}, \ttb\ttb\underline{\ttb}, \varepsilon)}
        \infer1[unary regrouping rule]{I(\tta\tta\tta\tta\ttb\tta\ttb\ttb, \tta\ttb\ttb\ttb, \varepsilon, \varepsilon, \varepsilon, \varepsilon)}
        \infer1[unary regrouping rule]{I(\varepsilon, \tta\tta\tta\tta\ttb\tta\ttb\ttb, \tta\ttb\ttb\ttb, \varepsilon, \varepsilon, \varepsilon)}
        \infer1[letter insertion rule]{I(\underline{\ttb}, \tta\tta\tta\tta\ttb\tta\ttb\ttb\underline{\tta}, \tta\ttb\ttb\ttb, \varepsilon, \varepsilon, \varepsilon)}
        \infer1[unary regrouping rule]{I(\ttb\tta\tta\tta\tta\ttb\tta\ttb\ttb\tta\tta\ttb\ttb\ttb, \varepsilon, \varepsilon, \varepsilon, \varepsilon, \varepsilon)}
        \infer1[start rule]{S(\ttb\tta\tta\tta\tta\ttb\tta\ttb\ttb\tta\tta\ttb\ttb\ttb)\rlap{.}}
      \end{prooftree}
    \]
    \caption{A derivation in \(\Gn{2}{6}\) simulating the \(G\)-derivation in \eqref{eq:simulation-example-derivation}. Newly inserted letters are underlined.}
    \label{fig:simulation-example}
  \end{figure}
\end{proof}

\subsubsection{Theorem A and the combinatorial reformulation}

\begin{theoremA}
  For every well-nested multiple context-free grammar \(G\) over \(\Sigma\) with \(L(G) \subseteq \Wpsi\),
  there exists \(r \geq 1\) such that \(L(G) \subseteq L(\Gpsi{r}) \subseteq \Wpsi\).
\end{theoremA}

\begin{proof}
  If \(L(G) = \varnothing\), then \(r = 1\) satisfies the required inclusions by Lemma~\ref{lem:gpsi-properties}.
  Otherwise let \(G'\) be the equivalent grammar supplied by Corollary~\ref{cor:neutral-binary-normal-form}.
  By Theorem~\ref{thm:simulation}, there exists an integer \(r \geq 1\) such that
  \(L(G) = L(G') \subseteq L(\Gpsi{r}) \subseteq \Wpsi\).
\end{proof}

\begin{corollary}\label{cor:nonbranching-refinement}
  Let \(G\) be a non-branching well-nested MCFG over \(\Sigma\) such that \(L(G) \subseteq \Wpsi\).
  Then there exists \(r \geq 1\) such that \(L(G) \subseteq L(\Gnbpsi{r}) \subseteq \Wpsi\).
\end{corollary}

\begin{proof}
  If \(L(G) = \varnothing\), then \(r = 1\) satisfies the required inclusions by Lemma~\ref{lem:gpsi-properties}.
  Otherwise, the non-branching clause of Corollary~\ref{cor:neutral-binary-normal-form} yields an equivalent reduced non-branching well-nested MCFG whose nonterminals are neutral.
  The non-branching clause of Theorem~\ref{thm:simulation} yields the required value of \(r\).
\end{proof}

\begin{corollary}\label{cor:n-specialization}
  Let \(n \geq 2\) and let \(G\) be a well-nested MCFG over \(\Sigman{n}\).
  If \(L(G) \subseteq \MIX_n\), then there exists \(r \geq 1\) such that \(L(G) \subseteq L(\Gn{n}{r}) \subseteq \MIX_n\).
\end{corollary}

\begin{proof}
  Apply Theorem~\ref{thm:theorem-a} with \(\psi = \psin{n}\).
\end{proof}

\begin{proposition}\label{prop:g3-tuples}
  For every \(r \geq 1\), the set of tuples derivable from \(I\) in \(\Gn{3}{r}\) is \(\mathcal{A}_r\).
  Thus,
  \[
    L(\Gn{3}{r}) = \concat(\mathcal{A}_r).
  \]
\end{proposition}

\begin{proof}
  By Example~\ref{ex:mpsi-words}, \(M_{\psin{3}} = \Perm(\tta\ttb\ttc)\).
  Hence the nullary, letter insertion, and binary rules of \(\Gn{3}{r}\) realize,
  respectively, the first, second, and third conditions defining \(\mathcal{A}_r\).
  Hence induction on the definition of \(\mathcal{A}_r\) shows that every tuple in \(\mathcal{A}_r\) is derivable from \(I\).

  Conversely, induction on \(I\)-derivations shows that every tuple derivable from \(I\) belongs to \(\mathcal{A}_r\).
  The nullary, letter insertion, and binary cases follow directly from the three defining conditions.
  The unary regrouping case follows from the regrouping closure property in \eqref{eq:a-r-regrouping-closure}.

  Thus the tuples derivable from \(I\) are exactly the tuples in \(\mathcal{A}_r\).
  Since the start rule is the only rule with head \(S\),
  \(L(\Gn{3}{r}) = \concat(\mathcal{A}_r)\).
\end{proof}

\begin{corollary}[Combinatorial reformulation]\label{cor:combinatorial-reformulation}
  The following are equivalent:
  \begin{enumerate}[leftmargin=*]
    \item\label{item:combinatorial-reformulation-conjecture}
      \(\MIX\) is not a well-nested multiple context-free language;
    \item\label{item:combinatorial-reformulation-grammar}
      \(L(\Gn{3}{r}) \subsetneq \MIX\) for every \(r \geq 1\);
    \item\label{item:combinatorial-reformulation-tuples}
      \(\concat(\mathcal{A}_r) \subsetneq \MIX\) for every \(r \geq 1\).
  \end{enumerate}
\end{corollary}

\begin{proof}
  We first show the equivalence of \ref{item:combinatorial-reformulation-conjecture}
  and \ref{item:combinatorial-reformulation-grammar}.
  By Lemma~\ref{lem:gpsi-properties}, \(L(\Gn{3}{r}) \subseteq \MIX\) for every \(r \geq 1\).
  If \(L(\Gn{3}{r}) = \MIX\) for some \(r\), then \(\MIX\) is a well-nested MCFL\@.
  Conversely, if \(\MIX\) is a well-nested MCFL, Corollary~\ref{cor:n-specialization} with \(n = 3\) implies
  \(\MIX \subseteq L(\Gn{3}{r}) \subseteq \MIX\) for some \(r\), hence \(L(\Gn{3}{r}) = \MIX\).
  The equivalence of \ref{item:combinatorial-reformulation-grammar}
  and \ref{item:combinatorial-reformulation-tuples} is Proposition~\ref{prop:g3-tuples}.
\end{proof}

\begin{remark}\label{rem:first-two-combinatorial-cases}
  For each \(w \in \Wpsi\), the singleton language \(\{w\}\) is a well-nested multiple context-free language.
  Applying Theorem~\ref{thm:theorem-a} to a grammar generating \(\{w\} \subseteq \Wpsi\) and using Lemma~\ref{lem:gpsi-properties}, we obtain
  \[
    \Wpsi = \bigcup_{r \geq 1} L(\Gpsi{r}).
  \]
  In particular, by Proposition~\ref{prop:g3-tuples},
  \[
    \MIX = \bigcup_{r \geq 1} L(\Gn{3}{r})
    = \bigcup_{r \geq 1} \concat(\mathcal{A}_r).
  \]
  Thus Corollary~\ref{cor:combinatorial-reformulation} reformulates
  the Kanazawa--Salvati conjecture as the assertion that the ascending chain
  \[
    \concat(\mathcal{A}_1) \subseteq \concat(\mathcal{A}_2) \subseteq \dotsb
  \]
  does not stabilize.
  By Lemma~\ref{lem:gpsi-properties},
  \(L(\Gn{3}{r}) \subseteq \MIX\) for every \(r \geq 1\).
  Since \(\Gn{3}{r}\) is a well-nested MCFG of dimension \(r\),
  Theorem~\ref{thm:kanazawa-salvati-dimension-two} shows that this inclusion is proper for \(r = 1, 2\).
  Thus \(\concat(\mathcal{A}_r) = L(\Gn{3}{r}) \subsetneq \MIX\) for \(r = 1, 2\).
  More explicitly, \(\mathcal{A}_1 = \{(\varepsilon)\}\) by definition, and hence
  \(\concat(\mathcal{A}_1) = \{\varepsilon\} \subsetneq \MIX\).
  For \(r = 2\), an exhaustive computer search shows that the word in Example~\ref{ex:combinatorial-witness}
  is lexicographically first among the shortest words
  in \(\MIX \setminus \concat(\mathcal{A}_2)\).
\end{remark}

\section{Proof of Theorem B}\label{sec:four-letter-consequence}

The proof of Theorem~\ref{thm:theorem-b} proceeds as follows.
By the result of Bishop--Elder--Evetts--Gallot--Levine, for every \(r \geq 1\),
there exists a word \(\mathcal{W}_r \in \MIX_2 \setminus L(\Gnbn{2}{r})\).
Put \(k = \len{\mathcal{W}_r}_{\tta} = \len{\mathcal{W}_r}_{\ttb}\).
Lemma~\ref{lem:binary-rule-elimination} shows that a \(\Gn{4}{r}\)-derivation of
\(\mathcal{W}_r\ttc^k\ttd^k\) can be replaced by a derivation without binary rules.
Projecting the latter derivation onto \(\Sigman{2}\) yields a
\(\Gnbn{2}{r}\)-derivation of \(\mathcal{W}_r\),
contrary to \(\mathcal{W}_r \notin L(\Gnbn{2}{r})\).

We begin with the result of Bishop--Elder--Evetts--Gallot--Levine.
They call a non-permuting, non-erasing, and non-branching MCFG an R-MCFG\@.
Their non-erasing condition is our non-deleting condition, and R-MCFGs are automatically well-nested;
hence the class of languages generated by R-MCFGs is \(\MCFLwn(1)\).
They prove that the language denoted by \(L\) in their paper, which is \(\MIX_2\) in our notation,
is not generated by any R-MCFG\@.
Thus their result can be stated in our terminology as follows.

\begin{theorem}[Bishop--Elder--Evetts--Gallot--Levine {\cite{bishop2026}}]\label{thm:beegl-nonbranching}
  \(\MIX_2 \notin \MCFLwn(1)\). \qed
\end{theorem}

\begin{proposition}\label{prop:two-letter-witnesses}
  For every \(r \geq 1\), there exists a word \(\mathcal{W}_r \in \MIX_2 \setminus L(\Gnbn{2}{r})\).
\end{proposition}

\begin{proof}
  By Lemma~\ref{lem:gpsi-properties}, \(\Gnbn{2}{r}\) is a non-branching well-nested MCFG such that \(L(\Gnbn{2}{r}) \subseteq \MIX_2\).
  Theorem~\ref{thm:beegl-nonbranching} implies that this inclusion is proper.
\end{proof}

See Appendix~\ref{app:beegl-witnesses} for an explicit choice of \(\mathcal{W}_r\) from the construction of
Bishop--Elder--Evetts--Gallot--Levine.

The key observation behind the next lemma is that a word \(w\ttc^k\ttd^k\) satisfying its hypotheses has no nonempty proper factor in \(\MIX_4\).
Consequently, in a derivation of such a word, if a binary rule is used as close to the root as possible, one of its two premises must derive a tuple consisting entirely of empty words.

\begin{lemma}[Binary rule elimination]\label{lem:binary-rule-elimination}
  Let \(r \geq 1\), \(w \in \Sigman{2}^*\), and suppose that
  \(\len{w}_{\tta} = \len{w}_{\ttb} = k\).
  Then
  \[
    w \ttc^k \ttd^k \in L(\Gn{4}{r})
    \implies
    w \ttc^k \ttd^k \in L(\Gnbn{4}{r}).
  \]
\end{lemma}

\begin{proof}
  Put \(z = w \ttc^k \ttd^k\).
  Among all derivation trees \(T\) satisfying \(\vdash_{\Gn{4}{r}} T : S(z)\), choose one with the minimum number of occurrences of binary rules.
  If \(T\) contains no binary rule, then \(T\) is already a \(\Gnbn{4}{r}\)-derivation tree of \(z\).

  Otherwise, choose an occurrence of a binary rule closest to the root of \(T\).
  Then there exist a binary rule \(\pi\), a derivation tree context \(U\), and derivation trees \(T_0, T_1\) such that
  \(T = U[\pi(T_0, T_1)]\) and \(U\) contains no binary rule.
  For some tuples \(\vv*{u}{0}\), \(\vv*{u}{1}\), and \(\vv{v}\), we have
  \[
    \vdash_{\Gn{4}{r}} T_0 : I(\vv*{u}{0}),
    \qquad
    \vdash_{\Gn{4}{r}} T_1 : I(\vv*{u}{1}),
    \qquad
    \vdash_{\Gn{4}{r}} \pi(T_0, T_1) : I(\vv{v}).
  \]

  The start and unary regrouping rules occurring in \(U\) preserve concatenation,
  while each letter insertion rule inserts one occurrence of each of \(\tta, \ttb, \ttc, \ttd\).
  Hence there exist \(w_0 \in \Sigman{2}^*\) and \(m \geq 0\) such that
  \(\concat(\vv{v}) = w_0 \ttc^m \ttd^m\) and \(\len{w_0}_{\tta} = \len{w_0}_{\ttb} = m\).

  Put \(s = \concat(\vv*{u}{0})\) and \(t = \concat(\vv*{u}{1})\).
  By the definition of the binary rules, there exist words \(s'\) and \(s''\) such that
  \(\concat(\vv{v}) = s' t s''\) and \(s = s' s''\).
  By Lemma~\ref{lem:gpsi-properties}, \(s, t \in \MIX_4\).

  If \(t \neq \varepsilon\), then \(t\) contains every letter of \(\Sigman{4}\),
  so it contains a letter from \(w_0\) and a letter from the final \(\ttd^m\) block.
  Since \(t\) is a factor of \(w_0 \ttc^m \ttd^m\), it therefore contains the whole \(\ttc^m\) block.
  Hence \(s = \varepsilon\), since \(s \in \MIX_4\) and \(s\) contains no \(\ttc\).
  Thus either \(s = \varepsilon\) or \(t = \varepsilon\).

  Put \(i = 0\) if \(t = \varepsilon\), and \(i = 1\) otherwise.
  Then every component of \(\vv*{u}{1 - i}\) is \(\varepsilon\).
  By the definition of \(\pi\), \(\vv{v}\) is a regrouping of a tuple obtained from \(\vv*{u}{i}\) by inserting the components of \(\vv*{u}{1 - i}\).
  All the inserted components are empty, so the resulting tuple is a regrouping of \(\vv*{u}{i}\).
  By transitivity of \(\sqsubseteq\), \(\vv{v} \sqsubseteq \vv*{u}{i}\).
  Hence there is a unary regrouping rule \(\rho\) such that \(\vdash_{\Gn{4}{r}} \rho(T_i) : I(\vv{v})\).
  It follows that \(\vdash_{\Gn{4}{r}} U[\rho(T_i)] : S(z)\).
  Since \(U\) contains no binary rule and \(\rho\) is unary, the derivation tree \(U[\rho(T_i)]\) contains fewer occurrences of binary rules than
  \(T = U[\pi(T_0, T_1)]\), contradicting the choice of \(T\).
  Therefore \(T\) contains no binary rule and is a \(\Gnbn{4}{r}\)-derivation tree of \(z\).
\end{proof}

\begin{theoremB}
  The language \(\MIX_4\) is not a well-nested multiple context-free language.
  Equivalently, \(O_3\) is not a well-nested multiple context-free language.
\end{theoremB}

\begin{proof}
  Assume, for contradiction, that \(\MIX_4\) is a well-nested MCFL\@.
  By Corollary~\ref{cor:n-specialization} with \(n = 4\), there exists \(r \geq 1\) such that \(L(\Gn{4}{r}) = \MIX_4\).
  Choose \(\mathcal{W}_r \in \MIX_2 \setminus L(\Gnbn{2}{r})\) from Proposition~\ref{prop:two-letter-witnesses}.
  Put \(k = \len{\mathcal{W}_r}_{\tta} = \len{\mathcal{W}_r}_{\ttb}\).
  Then \(\mathcal{W}_r \ttc^k \ttd^k \in \MIX_4 = L(\Gn{4}{r})\).
  By Lemma~\ref{lem:binary-rule-elimination},
  \(\mathcal{W}_r \ttc^k \ttd^k \in L(\Gnbn{4}{r})\).

  Let \(h = \proj_{\Sigman{2}}\colon \Sigman{4}^* \to \Sigman{2}^*\).
  Apply \(h\) to every tuple occurring in a \(\Gnbn{4}{r}\)-derivation of \(\mathcal{W}_r \ttc^k \ttd^k\).
  The nullary and start rules commute with \(h\).
  Since componentwise extensions of monoid homomorphisms preserve \(\sqsubseteq\), an application of a unary regrouping rule projects to an application of a unary regrouping rule.
  Each letter insertion rule projects to a letter insertion rule of \(\Gnbn{2}{r}\), because \(h\) erases \(\ttc, \ttd\)
  and leaves one \(\tta\) and one \(\ttb\) at distinct positions.
  Hence \(\mathcal{W}_r = h(\mathcal{W}_r \ttc^k \ttd^k) \in L(\Gnbn{2}{r})\), a contradiction.
  Therefore \(\MIX_4\) is not a well-nested multiple context-free language.
  The equivalent statement for \(O_3\) follows from \eqref{eq:mix-origin-crossing-wn-equivalence} with \(n = 3\).
\end{proof}

\begin{corollary}\label{cor:mix-n-not-well-nested}
  The language \(\MIX_n\) is not a well-nested multiple context-free language for every \(n \geq 4\),
  and \(O_n\) is not a well-nested multiple context-free language for every \(n \geq 3\).
\end{corollary}

\begin{proof}
  The case \(n = 4\) is Theorem~\ref{thm:theorem-b}, so suppose that \(n \geq 5\).
  Let \(h_n\colon \Sigman{4}^* \to \Sigman{n}^*\) be the monoid homomorphism defined by
  \(h_n(\tta_i) = \tta_i\) for \(1 \leq i \leq 3\) and \(h_n(\tta_4) = \tta_4\tta_5 \dotsm \tta_n\).
  Then \(\MIX_4 = h_n^{-1}(\MIX_n)\).
  The class \(\MCFLwn\) is closed under rational transductions and hence under inverse homomorphisms
  \cite{seki-kato2008}*{Theorem~15}.
  Thus, if \(\MIX_n \in \MCFLwn\), then
  \(\MIX_4 = h_n^{-1}(\MIX_n) \in \MCFLwn\), contrary to Theorem~\ref{thm:theorem-b}.
  The assertion about \(O_n\) follows from the assertion about \(\MIX_{n + 1}\)
  and \eqref{eq:mix-origin-crossing-wn-equivalence}.
\end{proof}

\begin{remark}
  The words \(w\ttc^k\ttd^k\) in Lemma~\ref{lem:binary-rule-elimination} have the following geometric interpretation.
  Every word \(w\ttc^k\ttd^k\) as in the lemma represents
  a closed walk in \(\Z^3\) without self-intersections.
  Let \(H\) be the plane in \(\mathbb{R}^3\) spanned by \((1, 1, 0)\) and \((0, 0, 1)\).
  Orthogonal projection onto \(H\) maps the step vectors
  \(\psin{4}(\tta)\), \(\psin{4}(\ttb)\), \(\psin{4}(\ttc)\), and \(\psin{4}(\ttd)\)
  to \((1 / 2, 1 / 2, 0)\), \((1 / 2, 1 / 2, 0)\), \((0, 0, 1)\), and \((-1, -1, -1)\), respectively.
  The projected walk therefore traverses the three sides of the triangle with vertices
  \((0, 0, 0)\), \((k, k, 0)\), and \((k, k, k)\), in cyclic order and without backtracking.
  Hence it has no self-intersections, and neither does the original walk.
  Equivalently, \(w\ttc^k\ttd^k\) has no nonempty proper factor belonging to \(\MIX_4\).

  Consider a tuple derivable from \(I\) whose concatenation is
  \(w_0\ttc^m\ttd^m\), where \(w_0 \in \Sigman{2}^*\) and
  \(\len{w_0}_{\tta} = \len{w_0}_{\ttb} = m\).
  If the last rule used in its derivation is a binary rule,
  the concatenation of one premise tuple is a factor of \(w_0\ttc^m\ttd^m\),
  and deleting this factor yields the concatenation of the other premise tuple.
  By Lemma~\ref{lem:gpsi-properties}, this factor belongs to \(\MIX_4\).
  Hence it is either empty or the entire word, and every component of at least
  one premise tuple is therefore \(\varepsilon\).
  If the last rule used in the derivation is a letter insertion rule,
  then \(m \geq 1\), and the concatenation of its premise tuple is
  \(w_1\ttc^{m - 1}\ttd^{m - 1}\) for some \(w_1 \in \Sigman{2}^*\) such that
  \(\len{w_1}_{\tta} = \len{w_1}_{\ttb} = m - 1\).

  The three-letter analogue of Lemma~\ref{lem:binary-rule-elimination} is false,
  so the argument used to prove Theorem~\ref{thm:theorem-b} does not extend directly
  to the Kanazawa--Salvati conjecture.
  Proposition~\ref{prop:binary-rule-elimination-three-letters}
  in Appendix~\ref{app:beegl-witnesses} gives an explicit counterexample.
\end{remark}

\section{Further background}
\label{sec:further-background}

This section places the two main results in the broader landscape of word problems and multiple context-free languages.
It is independent of the proofs above.

\subsection{Word problems and classes of formal languages}
\label{sec:group-word-problems}

For a choice of generators \(\pi\colon \Sigma^* \to G\),
the syntactic monoid of the word problem \(\pi^{-1}(1_G)\) is isomorphic to \(G\)
\cite{holt-rees-rover2017}*{3.4.9 Proposition}.
Thus the word problem contains enough information to recover \(G\) up to isomorphism.
Together with the invariance under changes of generators discussed in
Subsection~\ref{sec:rational-equivalence}, this makes it natural,
for a language class \(\mathcal{C}\) closed under inverse homomorphisms,
to ask which finitely generated groups have word problems in \(\mathcal{C}\).
Anisimov proved that the word problem of a finitely generated group \(G\) is regular
if and only if \(G\) is finite \cite{anisimov1971}.
Combining Muller and Schupp's theorem with Dunwoody's accessibility theorem shows
that the word problem of a finitely generated group \(G\) is context-free
if and only if \(G\) is virtually free, that is, \(G\) has a free subgroup of finite index
\cites{muller-schupp1983,dunwoody1985accessibility}.

Beyond the regular and context-free cases, word problems have been studied
for several other language and automaton classes.
Brough proved that every group virtually isomorphic to a finitely generated subgroup
of a direct product of free groups has a poly-context-free word problem,
and conjectured the converse \cite{brough2014}.
Elder, Kambites, and Ostheimer proved that a finitely generated group has
a word problem accepted by a blind \(n\)-counter automaton if and only if
it is virtually free abelian of rank at most \(n\)
\cite{elder-kambites-ostheimer2008}*{Theorem~1}.
Other counter models, including one-counter automata, have also been studied
\cite{holt-owens-thomas2008}.
Rino Nesin and Thomas proved that a finitely generated group has a word problem
that is a terminal Petri net language if and only if it is virtually abelian
\cite{rino-nesin-thomas2015}*{Theorem~1}.

The classification of groups with indexed word problems remains open.
Gilman's shrinking lemma was motivated by the study of groups for which the word problem is indexed
\cite{gilman1996shrinking}.
Gilman and Shapiro proved that an accessible group whose word problem is recognized
by a deterministic nested stack automaton with limited erasing,
accepting by final state and empty stack, is virtually free
\cite{gilman-shapiro1998}*{Theorem~1.1}.
This theorem concerns a restricted nested-stack model and does not settle the question
for arbitrary indexed word problems.
Gilman later remarked that it did not seem to be known whether the word problem
of \(\Z \times \Z\) is indexed \cite{gilman2005formal}.
Kanazawa and Salvati subsequently pointed out the connection between this question
and the indexedness of \(\MIX\)
\cite{kanazawa-salvati2012}*{Section~1, footnote~3}.
Elder explicitly conjectured a negative answer and proposed a strategy based on intersecting
the word problem with a particular regular language
\cite{elder2007gautomata}.
To our knowledge, no finitely generated group that is not virtually free is known
to have an indexed word problem
\cites{gilman-kropholler-schleimer2018,nyberg-brodda2023freeproducts}.
Since \(\MCFLwn \subseteq \IL\), no finitely generated group that is not virtually free
is currently known to have a word problem in \(\MCFLwn\) either.

For multiple context-free languages, Theorem~\ref{thm:salvati-dimension-two}
shows that a word problem of \(\Z^2\) is a \(2\)-MCFL\@.
Ho proved more generally that \(O_n\) is a multiple context-free language
for every \(n \geq 1\) \cite{ho2018}*{Theorem~2}.
More precisely, Theorem~\ref{thm:origin-crossing-dimension} shows that
\(O_n \in \mMCFL{n} \setminus \mMCFL{(n - 1)}\) for every \(n \geq 2\).
Kropholler and Spriano proved that if all vertex groups in a finite graph of groups
have multiple context-free word problems and all edge groups are finite,
then the fundamental group also has a multiple context-free word problem
\cite{kropholler-spriano2019}*{Theorem~A}.
Duncan, Elder, Frenkel, and Lyu proved that the rational subset membership problem
is decidable for every finitely generated group with a multiple context-free word problem
\cite{duncan-elder-frenkel-lyu2026}*{Lemma~6.1}.
Theorem~\ref{thm:theorem-b} and Corollary~\ref{cor:mix-n-not-well-nested}
show that no free abelian group of rank at least \(3\) has a word problem in \(\MCFLwn\).

\subsection{Equivalent formalisms and terminology}
\label{sec:equivalent-formalisms}

Multiple context-free languages are exactly the languages generated by string-based linear context-free rewriting systems
and by finite-copying parallel rewriting systems
\cite{rambow-satta1999}*{Sections~1 and 5}.
Denkinger characterized multiple context-free languages in terms of restricted tree stack automata:
for every \(m \geq 1\), the \(m\)-restricted tree stack automata recognize exactly the languages in \(\mMCFL{m}\)
\cite{denkinger2016automata}*{Theorem~4.12}.
Well-nested multiple context-free languages form a subclass of the multiple context-free languages
and are exactly the languages generated by coupled-context-free grammars
and by non-duplicating macro grammars, also called variable-linear macro grammars
\cite{kanazawa2009pumping}*{Section~3}.
Hotz and Pitsch gave a uniform parsing method for the coupled-context-free grammar hierarchy
\cite{hotz-pitsch1996}.
Sorokin introduced displacement context-free grammars and proved that they generate exactly the well-nested multiple context-free languages
\cite{sorokin2013normalforms}.
In the same paper, he established Chomsky- and Greibach-style normal forms for displacement context-free grammars.
Vijay-Shanker and Weir proved that tree-adjoining grammars, head grammars, linear indexed grammars, and the combinatory categorial grammars considered in their paper generate exactly the same class of languages
\cite{vijay-shanker-weir1994}.
For each \(m \geq 1\), the class \(\mMCFLwn{m}\) coincides with the class of languages generated
by variable-linear macro grammars of arity at most \(m - 1\)
\cite{kanazawa2009pumping}.

Table~\ref{tab:mcfg-terminology} summarizes the terminology used for MCFG parameters in related literature.
In papers on LCFRS, arity and branching are often called fan-out and rank, respectively \cites{rambow-satta1999,gomez-rodriguez2010efficient}.

\begin{table}[htb]
  \scriptsize
  \caption{Terminology and notation for the parameters called arity, dimension, branching, and branching factor in this paper.
  A dash means that no corresponding term or notation is recorded here from the cited source.}
  \label{tab:mcfg-terminology}
  \centering
  \begin{tabular}{l|l|l|l|l}
    \hline
    Paper & arity of \(A\) & dimension of \(G\) & branching of \(\pi\) & branching factor of \(G\)\\
    \hline
    \hline
    Seki et al.~\cite{seki1991}
    & \(d(A)\)
    & \(m\) in \(m\)-mcfg
    & \(a(f)\)
    & -- \\
    \hline
    Seki--Kato~\cite{seki-kato2008}
    & \(\dim(A)\)
    & \(\dim(G)\)
    & \(\operatorname{rank}(f)\)
    & \(\operatorname{rank}(G)\) \\
    \hline
    Kanazawa~\cite{kanazawa2009pumping}
    & rank of \(A\)
    & \(m\) in \(m\)-MCFG
    & rank of \(\pi\)
    & -- \\
    \hline
    Kanazawa--Salvati~\cite{kanazawa-salvati2010copying}
    & arity of \(A\)
    & dimension of \(G\)
    & --
    & branching factor (or rank)\\
    \hline
    Yoshinaka et al.~\cite{yoshinaka2010}
    & \(\dim(A)\)
    & \(q\) in \(q\)-MCFG\((r)\)
    & \(\operatorname{rank}(\pi)\)
    & \(r\) in \(q\)-MCFG\((r)\) \\
    \hline
    Kanazawa et al.~\cite{kanazawa-michaelis-salvati-yoshinaka2011}
    & \(\operatorname{arity}(A)\)
    & dimension of \(G\)
    & \(\operatorname{rank}(p)\)
    & \(\operatorname{rank}(G)\) \\
    \hline
    Kanazawa et al.~\cite{kanazawa-kobele-michaelis-salvati-yoshinaka2014}
    & dimension of \(A\)
    & \(m\) in \(m\)-MCFG
    & rank of \(\pi\)
    & -- \\
    \hline
    Salvati~\cite{salvati2015}
    & arity of \(A\)
    & \(k\) in \(k\)-MCFG
    & --
    & -- \\
    \hline
    Denkinger~\cite{denkinger2016automata}
    & fan-out of \(f\)
    & \(k\) in \(k\)-MCFG
    & --
    & -- \\
    \hline
    Nederhof~\cite{nederhof2017}
    & fanout of \(A\)
    & fanout of \(G\)
    & rank of \(\pi\)
    & rank of \(G\) \\
    \hline
    Kanazawa~\cite{kanazawa2019ogden}
    & dimension of \(A\)
    & \(m\) in \(m\)-MCFG
    & --
    & \(r\)-ary branching \\
    \hline
    Gebhardt et al.~\cite{gebhardt-meunier-salvati2022}
    & rank \(\rho(A)\)
    & dimension of \(G\)
    & --
    & -- \\
    \hline
    Bishop et al.~\cite{bishop2026}
    & rank of \(H\)
    & --
    & --
    & -- \\
    \hline
    Duncan et al.~\cite{duncan-elder-frenkel-lyu2026}
    & rank of \(A\)
    & \(k\) in \(k\)-MCF
    & --
    & -- \\
    \hline
    Rambow--Satta~\cite{rambow-satta1999}
    & fan-out \(\varphi(A)\)
    & fan-out \(\varphi(G)\)
    & rank \(\rho(p)\)
    & rank \(\rho(G)\) \\
    \hline
    G\'omez-Rodr\'iguez et al.~\cite{gomez-rodriguez2010efficient}
    & fan-out \(\varphi(A)\)
    & fan-out \(\varphi(G)\)
    & rank \(\rho(p)\)
    & rank \(\rho(G)\) \\
    \hline
  \end{tabular}
  \par\smallskip
  \begin{minipage}{\textwidth}
    \raggedright
    Bishop et al.\ do not introduce a name for the branching factor,
    but their definition of non-branching agrees with ours
    \cite{bishop2026}*{Definition~6.1}.

    For a rule with head \(A\) and composition function \(f\),
    the fan-out of \(f\) in Denkinger's terminology is the arity of \(A\) in ours.
  \end{minipage}
\end{table}

\subsection{Inclusions and separations}
\label{sec:language-class-landscape}

Figure~\ref{fig:mcfl-inclusions} summarizes the known inclusion relations among the classes considered in this paper.
The separation results underlying the figure are stated and proved in Appendix~\ref{app:language-class-landscape}.

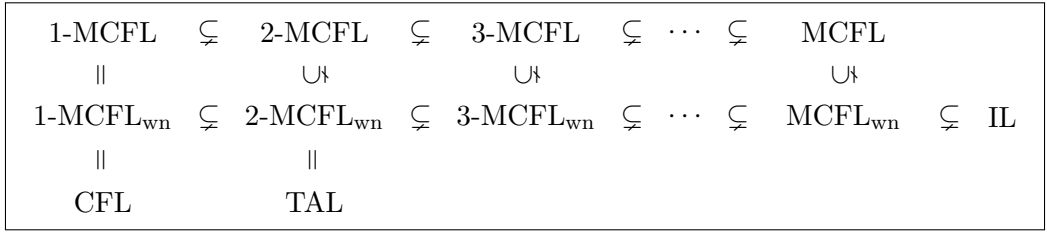
\begin{figure}[H]
  \newcommand{\SEP}{14mm}
  \centering
  \fbox{
    \begin{tikzpicture}
      \node (1MCFL) at (0, 0) {\(\mMCFL{1}\)};
      \node (rel12) at (\SEP * 1, 0) {\(\subsetneq\)};
      \node (2MCFL) at (\SEP * 2, 0) {\(\mMCFL{2}\)};
      \node (rel23) at (\SEP * 3, 0) {\(\subsetneq\)};
      \node (3MCFL) at (\SEP * 4, 0) {\(\mMCFL{3}\)};
      \node (rel3d) at (\SEP * 5, 0) {\(\subsetneq\)};
      \node (dMCFL) at (\SEP * 5.5, 0) {\(\cdots\)};
      \node (reldM) at (\SEP * 6, 0) {\(\subsetneq\)};
      \node (MCFL) at (\SEP * 7, 0) {\(\MCFL\)};
      \node (1MCFLwn) at (0, -\SEP * 0.8) {\(\mMCFLwn{1}\)};
      \node (relwn12) at (\SEP * 1, -\SEP * 0.8) {\(\subsetneq\)};
      \node (2MCFLwn) at (\SEP * 2, -\SEP * 0.8) {\(\mMCFLwn{2}\)};
      \node (relwn23) at (\SEP * 3, -\SEP * 0.8) {\(\subsetneq\)};
      \node (3MCFLwn) at (\SEP * 4, -\SEP * 0.8) {\(\mMCFLwn{3}\)};
      \node (relwn3d) at (\SEP * 5, -\SEP * 0.8) {\(\subsetneq\)};
      \node (dMCFLwn) at (\SEP * 5.5, -\SEP * 0.8) {\(\cdots\)};
      \node (relwndM) at (\SEP * 6, -\SEP * 0.8) {\(\subsetneq\)};
      \node (MCFLwn) at (\SEP * 7, -\SEP * 0.8) {\(\MCFLwn\)};
      \node (relwnI) at (\SEP * 8, -\SEP * 0.8) {\(\subsetneq\)};
      \node (IL) at (\SEP * 8.5, -\SEP * 0.8) {\(\IL\)};
      \node (CFL) at (0, -\SEP * 0.8 * 2) {\(\CFL\)};
      \node (TAL) at (\SEP * 2, -\SEP * 0.8 * 2) {\(\TAL\)};
      \node (relV1) at (\SEP * 0, -\SEP * 0.8 * 0.5) {\(\rotatebox{90}{\(=\)}\)};
      \node (relV2) at (\SEP * 2, -\SEP * 0.8 * 0.5) {\(\rotatebox{90}{\(\subsetneq\)}\)};
      \node (relV3) at (\SEP * 4, -\SEP * 0.8 * 0.5) {\(\rotatebox{90}{\(\subsetneq\)}\)};
      \node (relVM) at (\SEP * 7, -\SEP * 0.8 * 0.5) {\(\rotatebox{90}{\(\subsetneq\)}\)};
      \node (rel1C) at (\SEP * 0, -\SEP * 0.8 * 1.5) {\(\rotatebox{90}{\(=\)}\)};
      \node (rel1T) at (\SEP * 2, -\SEP * 0.8 * 1.5) {\(\rotatebox{90}{\(=\)}\)};
    \end{tikzpicture}
  }
  \caption{Known proper inclusions and equalities among the language classes considered in this paper.}
  \label{fig:mcfl-inclusions}
\end{figure}

In particular, the hierarchy
\[
  \mMCFLwn{1}
  \subsetneq
  \mMCFLwn{2}
  \subsetneq
  \mMCFLwn{3}
  \subsetneq
  \dotsb
\]
does not terminate at any fixed dimension, and its union is \(\MCFLwn\).
Thus, for a language \(L\), proving \(L \notin \mMCFLwn{m}\) for one fixed \(m\)
is weaker than proving \(L \notin \MCFLwn\).
The Kanazawa--Salvati conjecture asserts that \(\MIX \notin \MCFLwn\),
and Theorem~\ref{thm:theorem-b} proves that \(\MIX_4 \notin \MCFLwn\).

\section{Conclusion}\label{sec:conclusion}

Theorem~\ref{thm:theorem-a} shows that every well-nested multiple context-free sublanguage of \(\Wpsi\)
is contained in \(L(\Gpsi{r})\) for some \(r \geq 1\).
Theorem~\ref{thm:theorem-b} shows that \(\MIX_4 \notin \MCFLwn\).
By Corollary~\ref{cor:mix-n-not-well-nested}, \(\MIX_n \notin \MCFLwn\) for every \(n \geq 4\)
and \(O_n \notin \MCFLwn\) for every \(n \geq 3\).

The Kanazawa--Salvati conjecture that \(\MIX \notin \MCFLwn\) remains open.
Proposition~\ref{prop:g3-tuples} shows that
\(L(\Gn{3}{r}) = \concat(\mathcal{A}_r)\) for every \(r \geq 1\).
By Corollary~\ref{cor:combinatorial-reformulation}, the Kanazawa--Salvati conjecture is equivalent to
\(\concat(\mathcal{A}_r) \subsetneq \MIX\) for every \(r \geq 1\).

By Remark~\ref{rem:first-two-combinatorial-cases},
\(\concat(\mathcal{A}_r) \subsetneq \MIX\) for \(r = 1, 2\).
Hence the first open case is \(r = 3\):
it is not known whether \(\concat(\mathcal{A}_3) = \MIX\).
The full Kanazawa--Salvati conjecture requires
\(\concat(\mathcal{A}_r) \subsetneq \MIX\) for every \(r \geq 3\).

\appendix

\section{Proof of the binary branching reduction}\label{app:binary-branching}

\begin{proof}[Proof of Lemma~\ref{lem:binary-branching-reduction}]
  It suffices to replace one rule of branching \(n > 2\) by finitely many well-nested rules of branching at most \(2\).
  Fix a well-nested rule
  \[
    \pi\colon
    A(t_1, \dotsc, t_r)
    \gets
    B_1(x_{1, 1}, \dotsc, x_{1, r_1}), \dotsc, B_n(x_{n, 1}, \dotsc, x_{n, r_n})
  \]
  of branching \(n > 2\), and put \(\vv{t} = (t_1, \dotsc, t_r)\).
  Since reordering the atoms in the body of \(\pi\) does not change the generated language, we may assume that
  \(x_{1, 1} \dotsm x_{n, 1} \preccurlyeq \concat(\vv{t})\).
  For \(1 \leq k \leq n\), put
  \[
    R_k
    =
    r_1 + \dotsb + r_k,
    \qquad
    X_k
    =
    \{x_{i, j} \mid 1 \leq i \leq k, 1 \leq j \leq r_i\},
    \qquad
    s_k
    =
    \proj_{X_k}(\concat(\vv{t})).
  \]
  Then \(s_k\) lists the variables in \(X_k\) in the order in which they occur in \(\concat(\vv{t})\).
  Write \(s_k = z_{k, 1} \dotsm z_{k, R_k}\), and put \(\vv*{z}{k} = (z_{k, 1}, \dotsc, z_{k, R_k})\).

  For the rule \(\pi\), introduce fresh nonterminals \(C_2^\pi, \dotsc, C_n^\pi\), with arities \(R_2, \dotsc, R_n\), and set
  \(C_1^\pi = B_1\).
  Replace \(\pi\) by the following rules:
  \begin{align*}
    C_k^\pi(\vv*{z}{k})
    &\gets
    C_{k - 1}^\pi(\vv*{z}{k - 1}),
    B_k(x_{k, 1}, \dotsc, x_{k, r_k})
    \qquad
    (2 \leq k \leq n),\\
    A(t_1, \dotsc, t_r)
    &\gets C_n^\pi(\vv*{z}{n}).
  \end{align*}

  We check that the new rules are well-nested.
  For every \(2 \leq k \leq n\), there exists \(0 \leq l \leq R_{k - 1}\) such that
  \begin{equation}\label{eq:binary-branching-insertion}
    \vv*{z}{k}
    =
    (
      z_{k - 1, 1}, \dotsc, z_{k - 1, l},
      x_{k, 1}, \dotsc, x_{k, r_k},
      z_{k - 1, l + 1}, \dotsc, z_{k - 1, R_{k - 1}}
    ),
  \end{equation}
  where the first block is omitted if \(l = 0\), and the last block is omitted if \(l = R_{k - 1}\).
  Indeed, if no such \(l\) existed, there would be indices \(i < k\), \(\lambda < \mu\), and \(\nu\)
  such that \(x_{k, \lambda} x_{i, \nu} x_{k, \mu} \preccurlyeq \concat(\vv{t})\).
  Since \(i < k\), \(x_{1, 1} \dotsm x_{n, 1} \preccurlyeq \concat(\vv{t})\), and \(\pi\) is non-permuting,
  \(x_{i, 1}\) occurs before \(x_{k, \lambda}\) in \(\concat(\vv{t})\).
  Hence \(\nu > 1\), and \(x_{i, 1} x_{k, \lambda} x_{i, \nu} x_{k, \mu} \preccurlyeq \concat(\vv{t})\),
  contradicting the well-nestedness of \(\pi\).

  It follows from \eqref{eq:binary-branching-insertion} that each binary rule is well-nested.
  The final unary rule is non-deleting and non-permuting by the definition of \(s_n\), and it is well-nested because it has only one body atom.

  By construction, each application of \(\pi\) can be replaced by the chain of new rules.
  Since \(C_2^\pi, \dotsc, C_n^\pi\) are fresh and occur only in this chain,
  each occurrence of the chain in a derivation can be replaced by one application of \(\pi\);
  hence the replacement preserves the generated language.

  Replacing every rule of branching greater than \(2\) in this way yields an equivalent well-nested MCFG of branching factor at most \(2\).
  If \(L(G)\) is nonempty, applying Lemma~\ref{lem:reduced} to this grammar yields an equivalent reduced well-nested subgrammar with branching factor at most \(2\).
\end{proof}

\section{Higman's lemma}\label{app:higman}

\begin{lemma}[Higman's lemma {\cite{higman1952}}]\label{lem:higman}
  Let \(\Sigma\) be a finite alphabet.
  Then the subsequence order \(\preccurlyeq\) on \(\Sigma^*\) is a well-quasi-order.
  Equivalently, for every infinite sequence \(w_0, w_1, w_2, \dotsc\) of words over \(\Sigma\), there exist indices \(i < j\) such that \(w_i \preccurlyeq w_j\).
\end{lemma}

\begin{proof}
  An infinite sequence \((s_i)_{i \in \N}\) is called \emph{bad} if there are no indices \(i < j\) such that \(s_i \preccurlyeq s_j\).
  Assume, for contradiction, that a bad sequence exists.

  Construct a sequence \((u_i)_{i \in \N}\) recursively as follows.
  Suppose that \(u_0, \dotsc, u_{i - 1}\) have been chosen.
  Put
  \[
    E_i = \mleft\{ u \in \Sigma^* \setmid
      \begin{aligned}
        &\text{there exist words \(z_{i + 1}, z_{i + 2}, \dotsc \in \Sigma^*\) such that}\\
        &\text{\(u_0, \dotsc, u_{i - 1}, u, z_{i + 1}, z_{i + 2}, \dotsc\) is a bad sequence}
      \end{aligned}
    \mright\}.
  \]
  The set \(E_0\) is nonempty by assumption.
  If \(E_i\) is nonempty, choose \(u_i \in E_i\) of minimum length.
  If \(z_{i + 1}, z_{i + 2}, \dotsc\) witness that \(u_i \in E_i\), then \(z_{i + 1} \in E_{i + 1}\).
  In particular, \(E_{i + 1}\) is nonempty, so the recursive construction is well-defined.
  By construction, every finite prefix of \((u_i)_{i \in \N}\) can be extended to a bad sequence, so \((u_i)_{i \in \N}\) is itself bad.

  Each \(u_i\) is nonempty, since otherwise \(u_i = \varepsilon \preccurlyeq u_{i + 1}\).
  Write \(u_i = a_i v_i\), where \(a_i \in \Sigma\) and \(v_i \in \Sigma^*\).
  Since \(\Sigma\) is finite, there exist a letter \(a \in \Sigma\) and indices \(0 < i(0) < i(1) < \dotsb\) such that \(a_{i(j)} = a\) for every \(j \in \N\).
  Let \((w_i)_{i \in \N}\) denote the sequence
  \[
    u_0, u_1, \dotsc, u_{i(0) - 1}, v_{i(0)}, v_{i(1)}, \dotsc.
  \]
  We show that \((w_i)_{i \in \N}\) is a bad sequence.
  Assume, for contradiction, that \(w_i \preccurlyeq w_j\) for some indices \(i < j\).
  If \(i < j < i(0)\), then \(u_i = w_i \preccurlyeq w_j = u_j\), which contradicts the badness of \((u_i)_{i \in \N}\).
  If \(i < i(0) \leq j\), then there exists \(k \in \N\) such that
  \(u_i = w_i \preccurlyeq w_j = v_{i(k)} \preccurlyeq a v_{i(k)} = u_{i(k)}\), a contradiction.
  If \(i(0) \leq i < j\), then there exist \(k < l\) such that \(v_{i(k)} = w_i \preccurlyeq w_j = v_{i(l)}\),
  hence \(u_{i(k)} = a v_{i(k)} \preccurlyeq a v_{i(l)} = u_{i(l)}\), a contradiction.
  Thus \((w_i)_{i \in \N}\) is bad.

  Hence \(v_{i(0)} = w_{i(0)} \in E_{i(0)}\).
  However, \(\len{v_{i(0)}} < \len{u_{i(0)}}\), which contradicts the choice of \(u_{i(0)}\).
\end{proof}

\begin{corollary}\label{cor:subsequence-minimal-finite}
  Let \(\Sigma\) be a finite alphabet.
  For every subset \(X \subseteq \Sigma^*\), the set \(\Minsubseq(X)\) is finite.
\end{corollary}

\begin{proof}
  Since the elements of \(\Minsubseq(X)\) are pairwise incomparable under \(\preccurlyeq\), \(\Minsubseq(X)\) is an antichain in \((\Sigma^*, \preccurlyeq)\).
  By Higman's lemma (Lemma~\ref{lem:higman}), this antichain is finite.
\end{proof}

\section{Minimal representatives and the Neutralization Theorem}\label{app:neutralization}

\subsection{Minimal representatives}\label{app:representatives}

As in Section~\ref{sec:proof-theorem-a}, fix a finite alphabet \(\Sigma\), an integer \(d \geq 1\),
and a surjective monoid homomorphism \(\psi\colon \Sigma^* \to \Z^d\).

\begin{definition}
  For \(\bm{v} \in \Z^d\), let
  \[
    \Reppsi(\bm{v})
    =
    \Minsubseq(\psi^{-1}(\bm{v})).
  \]
  An element of \(\Reppsi(\bm{v})\) is called a
  \emph{minimal representative} of \(\bm{v}\) with respect to \(\psi\).
\end{definition}

\begin{remark}[\(\Rep_{\psin{n}}(\bm{v})\) as the set of geodesics for \(\bm{v}\)]
  Let \(n \geq 2\) and \(\bm{v} \in \Z^{n - 1}\).
  We say that a word \(w \in \Sigman{n}^*\) is \emph{geodesic for \(\bm{v}\)} if \(\psin{n}(w) = \bm{v}\) and \(w\) has minimum length among all words with \(\psin{n}\)-value \(\bm{v}\).

  By the definition of \(\psin{n}\), a word \(w \in \psin{n}^{-1}(\bm{v})\) is geodesic for \(\bm{v}\) if and only if at least one letter of \(\Sigman{n}\) is absent from \(w\).
  Since the same condition characterizes membership in \(\Rep_{\psin{n}}(\bm{v})\), the set \(\Rep_{\psin{n}}(\bm{v})\) is exactly the set of words that are geodesic for \(\bm{v}\).
\end{remark}

\begin{lemma}\label{lem:rep-finite}
  For every \(\bm{v} \in \Z^d\), the set \(\Reppsi(\bm{v})\) is finite.
\end{lemma}

\begin{proof}
  Apply Corollary~\ref{cor:subsequence-minimal-finite} to \(X = \psi^{-1}(\bm{v})\).
\end{proof}

\begin{lemma}\label{lem:representative-subsequence}
  For every \(\bm{v} \in \Z^d\) and every \(z \in \psi^{-1}(\bm{v})\),
  there exists \(u \in \Reppsi(\bm{v})\) such that \(u \preccurlyeq z\).
\end{lemma}

\begin{proof}
  Take a shortest subsequence \(u \preccurlyeq z\) satisfying \(\psi(u) = \bm{v}\).
  Then no \(s \prec u\) satisfies \(\psi(s) = \bm{v}\), and thus \(u \in \Reppsi(\bm{v})\).
\end{proof}

Note that \(\Reppsi(\bm{0}) = \{\varepsilon\}\) since \(\varepsilon\) is a subsequence of every word in \(\Wpsi\), while \(\varepsilon \notin M_\psi\).

\subsection{Proof of the Neutralization Theorem}

We use the following join and split operations in this proof and in Appendix~\ref{app:simulation}.

\begin{definition}\label{def:join-split}
  For an alphabet \(\Omega\) and \(u = d_1 \dotsm d_k \in \Omega^k\) (\(k \geq 0\)), define the map
  \(\join_u\colon (\Omega^*)^{k + 1} \to \Omega^*\) by
  \(\join_u(w_0, \dotsc, w_k) = w_0 d_1 w_1 \dotsm d_k w_k\).

  For \(\Delta \subseteq \Omega\) and a word \(w \in \Omega^*\),
  \(w\) has a unique factorization \(w = w_0 d_1 w_1 \dotsm d_k w_k\),
  where \(k \geq 0\), \(d_1, \dotsc, d_k \in \Delta\), and \(w_0, \dotsc, w_k \in (\Omega \setminus \Delta)^*\).
  Define the map \(\Split_\Delta\colon \Omega^* \to (\Omega^*)^*\) by
  \(\Split_\Delta(w) = (w_0, \dotsc, w_k)\).
  For a letter \(a \in \Omega\), we write \(\Split_a\) for \(\Split_{\{a\}}\).
\end{definition}

The tuple \(\Split_\Delta(w)\) has \(\len{\proj_\Delta(w)} + 1\) components,
and satisfies \(\join_{\proj_\Delta(w)}(\Split_\Delta(w)) = w\).
For example, let \(\Omega = \{\tta, \ttb, \ttc\}\), \(\Delta = \{\tta, \ttb\}\), and \(w = \tta\ttb\ttc\tta\ttc\ttb\).
Then \(\Split_\Delta(w) = (\varepsilon, \varepsilon, \ttc, \ttc, \varepsilon)\).
The four commas correspond to the four letters of \(\proj_\Delta(w) = \tta\ttb\tta\ttb\).

\begin{proof}[Proof of Theorem~\ref{thm:neutralization}]
  For each nonterminal \(A\), let \(\bm{v}_A\) be the fixed \(\psi\)-value from Lemma~\ref{lem:nonterminal-value}.

  For a nonterminal \(A\) of arity \(r\) and \(u = a_1 \dotsm a_k \in \Reppsi(\bm{v}_A)\) (\(a_i \in \Sigma\)),
  an \emph{instantiation layout} for \((A, u)\) is a word \(\eta \in \{\separator, \placeholder\}^{r + k - 1}\)
  such that \(\len{\eta}_{\separator} = r - 1\) and \(\len{\eta}_{\placeholder} = k\).
  For each such pair \((u, \eta)\), we introduce a new nonterminal \(A_\eta^u\) of arity \(r + k\).
  We call the symbol \(\separator\) a \emph{separator} and the symbol \(\placeholder\) a \emph{placeholder}.
  The \(i\)-th symbol of the instantiation layout \(\eta\) specifies how the \(i\)-th comma in \(A_\eta^u(s_0, \dotsc, s_{r + k - 1})\) is interpreted.
  If this symbol is the \(j\)-th separator in \(\eta\), then this comma represents the \(j\)-th comma in \(A(x_1, \dotsc, x_r)\).
  If this symbol is the \(j\)-th placeholder in \(\eta\), then this comma represents the position reserved for the insertion of the \(j\)-th letter \(a_j\) of \(u\).

  Let \(\Omega\) be any alphabet containing \(\Sigma\) and disjoint from \(\{\separator, \placeholder\}\).
  For \(u = a_1 \dotsm a_k \in \Reppsi(\bm{v}_A)\) and an instantiation layout \(\eta\) for \((A, u)\), define an \emph{instantiation map}
  \(\inst_\eta^u\colon (\Omega^*)^{r + k} \to (\Omega^*)^r\) by the following formula,
  where the join and split operations are taken over \(\Omega \cup \{\separator, \placeholder\}\):
  \[
    \inst_\eta^u(\vv{s})
    =
    \Split_{\separator}
    (
      \join_u(
        \Split_{\placeholder}(
          \join_\eta(\vv{s})
        )
      )
    )
    \quad\text{for}\quad
    \vv{s} \in (\Omega^*)^{r + k}.
  \]
  If \(h\colon \Omega^* \to \Sigma^*\) is a monoid homomorphism such that
  \(h(a) = a\) for \(a \in \Sigma\), then
  \begin{equation}\label{eq:inst-homomorphism}
    h(\inst_\eta^u(\vv{s}))
    =
    \inst_\eta^u(h(\vv{s}))
  \end{equation}
  for every \(\vv{s} \in (\Omega^*)^{r + k}\).
  For example, suppose that a nonterminal \(C\) has arity \(4\) and \(u = \tta\ttb\tta \in \Reppsi(\bm{v}_C)\), and put
  \(\eta = \placeholder \placeholder \separator \separator \placeholder \separator\).
  Then, for \(s_0, s_1, s_2, s_3, s_4, s_5, s_6 \in \Omega^*\),
  \begin{align*}
    \inst_{\placeholder \placeholder \separator \separator \placeholder \separator}^{\tta\ttb\tta}(s_0, s_1, s_2, s_3, s_4, s_5, s_6)
    &=
    \Split_{\separator}
    (
      \join_{\tta\ttb\tta}(
        \Split_{\placeholder}(
          \join_{\placeholder \placeholder \separator \separator \placeholder \separator}(s_0, s_1, s_2, s_3, s_4, s_5, s_6)
        )
      )
    )\\
    &=
    \Split_{\separator}
    (
      \join_{\tta\ttb\tta}(
        \Split_{\placeholder}(
          s_0 \placeholder s_1 \placeholder s_2 \separator s_3 \separator s_4 \placeholder s_5 \separator s_6
        )
      )
    )\\
    &=
    \Split_{\separator}
    (
      \join_{\tta\ttb\tta}(
        s_0, s_1, s_2 \separator s_3 \separator s_4, s_5 \separator s_6
      )
    )\\
    &=
    \Split_{\separator}
    (
      s_0 \tta s_1 \ttb s_2 \separator s_3 \separator s_4 \tta s_5 \separator s_6
    )\\
    &=
    (s_0 \tta s_1 \ttb s_2, s_3, s_4 \tta s_5, s_6).
  \end{align*}
  Since \(\psi\) is a monoid homomorphism and \(\Z^d\) is abelian,
  it follows from the definition of \(\inst_\eta^u\) that every \(\vv{z} \in (\Sigma^*)^{r + k}\) satisfies
  \begin{equation}\label{eq:inst-general}
    \psi(\concat(\inst_\eta^u(\vv{z}))) = \psi(\concat(\vv{z})) + \psi(u).
  \end{equation}

  \proofheading{Construction of \(G'\)}
  Put
  \[
    N' = \{A_\eta^u \mid A \in N, \ u \in \Reppsi(\bm{v}_A), \ \text{\(\eta\) is an instantiation layout for \((A, u)\)}\}.
  \]
  Since \(\bm{v}_S = \bm{0}\) and \(\Reppsi(\bm{0}) = \{\varepsilon\}\),
  set \(S' = S_\varepsilon^\varepsilon\).

  For each rule
  \[
    \pi\colon
    A(t_1, \dotsc, t_r)
    \gets
    B_1(x_{1, 1}, \dotsc, x_{1, r_1}), \dotsc, B_n(x_{n, 1}, \dotsc, x_{n, r_n})
  \]
  in \(P\), put \(\vv{t} = (t_1, \dotsc, t_r)\) and
  \(\vv*{x}{i} = (x_{i, 1}, \dotsc, x_{i, r_i})\) for \(1 \leq i \leq n\).
  Let \(P'_\pi\) be the set of all rules
  \[
    A_\eta^u(\vv{s})
    \gets
    (B_1)_{\eta_1}^{u_1}(\vv*{x}{1}'), \dotsc, (B_n)_{\eta_n}^{u_n}(\vv*{x}{n}')
  \]
  satisfying the following conditions.
  \begin{enumerate}
    \item\label{item:neutralization-body-data} For each \(1 \leq i \leq n\), we have \(u_i \in \Reppsi(\bm{v}_{B_i})\),
    \(\eta_i\) is an instantiation layout for \((B_i, u_i)\), and \(\vv*{x}{i}'\) is a tuple of \(r_i + \len{u_i}\) variables.
    The variables occurring in \(\vv*{x}{1}', \dotsc, \vv*{x}{n}'\) are pairwise distinct and fresh.
    \item\label{item:neutralization-head-data} We have \(u \in \Reppsi(\bm{v}_A)\), \(\eta\) is an instantiation layout for \((A, u)\), and
    \(\vv{s}\) has \(r + \len{u}\) components.
    \item\label{item:neutralization-rule-compatibility} Let \(\tau\) be the unique monoid homomorphism determined by
    \(\tau(a) = a\) for \(a \in \Sigma\) and, for each \(1 \leq i \leq n\), by
    \begin{equation}\label{eq:neutralization-body}
      \tau(\vv*{x}{i})
      =
      \inst_{\eta_i}^{u_i}(\vv*{x}{i}').
    \end{equation}
    We require
    \begin{equation}\label{eq:neutralization-head}
      \inst_\eta^u(\vv{s})
      =
      \tau(\vv{t}).
    \end{equation}
  \end{enumerate}
  Put \(P' = \bigcup_{\pi \in P} P'_\pi\).

  To illustrate the construction of \(P'_\pi\), consider the rule
  \[
    \pi_5\colon A(y x_1, \tta x_2 \ttc)
    \gets
    B(y),
    A(x_1, x_2)
  \]
  in the proof sketch.
  For the nonterminals \(A\) and \(B\) of this grammar, we have \(\Rep_{\psin{3}}(\bm{v}_A) = \{\tta\ttb, \ttb\tta\}\) and \(\Rep_{\psin{3}}(\bm{v}_B) = \{\ttb\}\).
  For example, the following rules belong to \(P'_{\pi_5}\):
  \begin{align*}
    A_{\placeholder \placeholder \separator}^{\ttb\tta}(y'_0, y'_1 x'_0, x'_1 \ttb x'_2, \tta x'_3 \ttc)
    &\gets
    B_{\placeholder}^{\ttb}(y'_0, y'_1),
    A_{\placeholder \placeholder \separator}^{\tta\ttb}(x'_0, x'_1, x'_2, x'_3)\\
    (\text{where}\quad& \tau(y) = y'_0 \ttb y'_1, \tau(x_1) = x'_0 \tta x'_1 \ttb x'_2, \tau(x_2) = x'_3),\\
    A_{\separator \placeholder \placeholder}^{\tta\ttb}(y'_0 \ttb y'_1 x'_0 \tta x'_1, \varepsilon, x'_2, x'_3 \ttc)
    &\gets
    B_{\placeholder}^{\ttb}(y'_0, y'_1),
    A_{\placeholder \separator \placeholder}^{\tta\ttb}(x'_0, x'_1, x'_2, x'_3)\\
    (\text{where}\quad& \tau(y) = y'_0 \ttb y'_1, \tau(x_1) = x'_0 \tta x'_1, \tau(x_2) = x'_2 \ttb x'_3).
  \end{align*}
  In the informal notation of the proof sketch, the corresponding inference schemata are as follows:
  \begin{align*}
    \begin{prooftree}
      \hypo{B^{\ttb}(y'_0 \placeholder y'_1)}
      \hypo{A^{\tta\ttb}(x'_0 \placeholder x'_1 \placeholder x'_2, x'_3)}
      \infer2{A^{\ttb\tta}(y'_0 \placeholder y'_1 x'_0 \placeholder x'_1 \ttb x'_2, \tta x'_3 \ttc)\rlap{,}}
    \end{prooftree}
    &&
    \begin{prooftree}
      \hypo{B^{\ttb}(y'_0 \placeholder y'_1)}
      \hypo{A^{\tta\ttb}(x'_0 \placeholder x'_1, x'_2 \placeholder x'_3)}
      \infer2{A^{\tta\ttb}(y'_0 \ttb y'_1 x'_0 \tta x'_1, \placeholder x'_2 \placeholder x'_3 \ttc)\rlap{.}}
    \end{prooftree}
  \end{align*}
  The first rule is the one used in the derivation displayed in the proof sketch.

  \proofheading{Finiteness}
  By Lemma~\ref{lem:rep-finite}, each set \(\Reppsi(\bm{v}_A)\) is finite.
  If \(A\) has arity \(r\) and \(\len{u} = k\), then there are exactly \(\binom{r + k - 1}{k}\) instantiation layouts for \((A, u)\).
  Thus \(N'\) is finite.
  Fix a rule \(\pi \in P\), representatives \(u_1 \in \Reppsi(\bm{v}_{B_1}), \dotsc, u_n \in \Reppsi(\bm{v}_{B_n}), u \in \Reppsi(\bm{v}_A)\),
  instantiation layouts \(\eta_i\) for \((B_i, u_i)\) for \(1 \leq i \leq n\),
  and an instantiation layout \(\eta\) for \((A, u)\).
  For these fixed choices, condition~\ref{item:neutralization-body-data} determines the tuples \(\vv*{x}{i}'\) up to \(\alpha\)-equivalence,
  and \eqref{eq:neutralization-body} then determines the homomorphism \(\tau\) and the tuple \(\tau(\vv{t})\).
  Since the tuple \(\vv{s}\) has \(r + \len{u}\) components and total length
  \(\len{\concat(\tau(\vv{t}))} - \len{u}\) by \eqref{eq:neutralization-head},
  and its components use only letters of \(\Sigma\) and variables occurring in \(\vv*{x}{1}', \dotsc, \vv*{x}{n}'\),
  only finitely many choices of \(\vv{s}\) are possible.
  There are also only finitely many choices of
  \((u_1, \eta_1), \dotsc, (u_n, \eta_n)\) and \((u, \eta)\),
  so \(P'_\pi\) is finite.
  Since \(P\) is finite, so is \(P'\).

  \proofheading{Well-nestedness and branching factor}
  Consider a rule in \(P'_\pi\).
  Since \(\pi\) is non-deleting and non-permuting and instantiation maps neither delete nor reorder variables, condition~\ref{item:neutralization-rule-compatibility} implies that this rule is non-deleting and non-permuting.
  By condition~\ref{item:neutralization-rule-compatibility}, for each \(i, j\), the variables occurring in \(\tau(x_{i, j})\) form a consecutive block of the variables in \(\vv*{x}{i}'\), and these blocks occur in \(\concat(\vv{s})\) in the order in which the variables \(x_{i, j}\) occur in \(\concat(\vv{t})\).
  If the new rule were not well-nested, the four variables witnessing this would belong to four distinct blocks.
  The corresponding variables of \(\pi\) would then witness that \(\pi\) is not well-nested, a contradiction.
  Each rule in \(P'_\pi\) has the same branching as \(\pi\), so the branching factor of \(G'\) is at most that of \(G\).

  \proofheading{The inclusion \(L(G') \subseteq L(G)\)}
  We prove by induction on \(G'\)-derivations that, for every nonterminal \(A\) of arity \(r\),
  every \(u \in \Reppsi(\bm{v}_A)\), every instantiation layout \(\eta\) for \((A, u)\), and every \(\vv{z} \in (\Sigma^*)^{r + \len{u}}\),
  \begin{equation}\label{eq:neutralization-to-original}
    \vdash_{G'} A_\eta^u(\vv{z})
    \implies
    \vdash_G A(\inst_\eta^u(\vv{z})).
  \end{equation}
  Suppose that \(\vdash_{G'} A_\eta^u(\vv{z})\), and choose a \(G'\)-derivation of this judgment.
  Let \(\pi' \in P'\) be the last rule used in this derivation, and choose a rule \(\pi \in P\) such that \(\pi' \in P'_\pi\).
  By the definition of derivability, there exists a monoid homomorphism \(\sigma\) such that \(\sigma(a) = a\) for \(a \in \Sigma\),
  \(\vdash_{G'} (B_i)_{\eta_i}^{u_i}(\sigma(\vv*{x}{i}'))\) for \(1 \leq i \leq n\),
  and \(\vv{z} = \sigma(\vv{s})\).
  For each \(1 \leq i \leq n\), the induction hypothesis for the \(i\)-th premise implies that
  \(\vdash_G B_i(\inst_{\eta_i}^{u_i}(\sigma(\vv*{x}{i}')))\).
  By Equations~\eqref{eq:inst-homomorphism} and~\eqref{eq:neutralization-body},
  \(\inst_{\eta_i}^{u_i}(\sigma(\vv*{x}{i}')) = \sigma(\inst_{\eta_i}^{u_i}(\vv*{x}{i}')) = \sigma(\tau(\vv*{x}{i}))\).
  Hence the inference schema for \(G\), with rule \(\pi\) and homomorphism
  \(\sigma \circ \tau\), yields the judgment \(\vdash_G A(\sigma(\tau(\vv{t})))\).
  Applying \(\sigma\) to \eqref{eq:neutralization-head} and using \eqref{eq:inst-homomorphism}, we obtain
  \(\sigma(\tau(\vv{t})) = \sigma(\inst_\eta^u(\vv{s})) = \inst_\eta^u(\sigma(\vv{s})) = \inst_\eta^u(\vv{z})\).
  This proves \eqref{eq:neutralization-to-original}.
  Since \(\inst_\varepsilon^\varepsilon\) is the identity map, the case \(A_\eta^u = S'\) of \eqref{eq:neutralization-to-original} proves \(L(G') \subseteq L(G)\).

  \proofheading{The inclusion \(L(G) \subseteq L(G')\)}
  We prove the following statement by induction on \(G\)-derivations:
  if a nonterminal \(A\) of arity \(r\) and a tuple
  \(\vv{w} \in (\Sigma^*)^r\) satisfy \(\vdash_G A(\vv{w})\), then there exist \(u \in \Reppsi(\bm{v}_A)\), an instantiation layout \(\eta\) for \((A, u)\),
  and \(\vv{z} \in (\Sigma^*)^{r + \len{u}}\) such that \(\vdash_{G'} A_\eta^u(\vv{z})\) and \(\inst_\eta^u(\vv{z}) = \vv{w}\).

  Suppose that \(\vdash_G A(\vv{w})\) for a nonterminal \(A\) of arity \(r\) and a tuple \(\vv{w} \in (\Sigma^*)^r\), and choose a \(G\)-derivation of this judgment.
  Let \(\pi\) be the last rule used in this derivation, and write it as in the definition of \(P'_\pi\).
  By the definition of derivability, there exists a monoid homomorphism \(\sigma\) such that \(\sigma(a) = a\) for \(a \in \Sigma\),
  \(\vdash_G B_i(\sigma(\vv*{x}{i}))\) for \(1 \leq i \leq n\),
  and \(\vv{w} = \sigma(\vv{t})\).
  By the induction hypothesis, for each \(1 \leq i \leq n\), there exist
  \(u_i \in \Reppsi(\bm{v}_{B_i})\), an instantiation layout \(\eta_i\) for \((B_i, u_i)\),
  and a tuple \(\vv*{z}{i} \in (\Sigma^*)^{r_i + \len{u_i}}\) such that
  \(\vdash_{G'} (B_i)_{\eta_i}^{u_i}(\vv*{z}{i})\) and \(\inst_{\eta_i}^{u_i}(\vv*{z}{i}) = \sigma(\vv*{x}{i})\).
  Choose pairwise distinct fresh variables forming a tuple \(\vv*{x}{i}'\) of length \(r_i + \len{u_i}\) for each \(1 \leq i \leq n\), and define \(\tau\) by \eqref{eq:neutralization-body}.
  Let \(\sigma'\) be the unique monoid homomorphism determined by \(\sigma'(a) = a\) for \(a \in \Sigma\) and \(\sigma'(\vv*{x}{i}') = \vv*{z}{i}\) for \(1 \leq i \leq n\).
  By Equations~\eqref{eq:inst-homomorphism} and~\eqref{eq:neutralization-body},
  \((\sigma' \circ \tau)(\vv*{x}{i})
  = \sigma'(\inst_{\eta_i}^{u_i}(\vv*{x}{i}'))
  = \inst_{\eta_i}^{u_i}(\sigma'(\vv*{x}{i}'))
  = \inst_{\eta_i}^{u_i}(\vv*{z}{i})
  = \sigma(\vv*{x}{i})\)
  for each \(1 \leq i \leq n\).
  Since the tuples \(\vv*{x}{1}, \dotsc, \vv*{x}{n}\) contain all variables of \(\pi\), and both \(\sigma' \circ \tau\) and \(\sigma\) fix every letter of \(\Sigma\), we have \(\sigma' \circ \tau = \sigma\).

  Put \(w = \proj_\Sigma(\concat(\tau(\vv{t})))\).
  By \eqref{eq:neutralization-body} and the definition of \(\inst_{\eta_i}^{u_i}\),
  \(\proj_\Sigma(\concat(\tau(\vv*{x}{i}))) = \proj_\Sigma(\concat(\inst_{\eta_i}^{u_i}(\vv*{x}{i}'))) = u_i\)
  for each \(1 \leq i \leq n\).
  Applying Lemma~\ref{lem:terminal-contribution-formula} to \(\proj_\Sigma \circ \tau\), we obtain
  \[
    \psi(w)
    =
    \psi(\tc(\pi)) + \sum_{i = 1}^n \psi(u_i)
    =
    \psi(\tc(\pi)) + \sum_{i = 1}^n \bm{v}_{B_i},
  \]
  where the second equality follows from \(u_i \in \Reppsi(\bm{v}_{B_i})\).
  Since \(G\) is reduced, applying \(\pi\) to one \(B_i\)-derivable tuple for each \(1 \leq i \leq n\) yields an \(A\)-derivable tuple.
  By Lemmas~\ref{lem:terminal-contribution-formula} and~\ref{lem:nonterminal-value},
  \[
    \bm{v}_A
    =
    \psi(\tc(\pi)) + \sum_{i = 1}^n \bm{v}_{B_i}.
  \]
  Therefore \(\psi(w) = \bm{v}_A\).

  By Lemma~\ref{lem:representative-subsequence}, choose \(u \in \Reppsi(\bm{v}_A)\) such that \(u \preccurlyeq w\).
  By the definition of \(w\), we also have \(u \preccurlyeq \concat(\tau(\vv{t}))\).
  Choose positions in \(\concat(\tau(\vv{t}))\) whose corresponding subsequence is \(u\), and delete the letters at those positions.
  Reading the component boundaries of \(\tau(\vv{t})\) as separators and the deleted letters as placeholders determines an instantiation layout \(\eta\) for \((A, u)\).
  The remaining factors yield a tuple \(\vv{s}\) with \(r + \len{u}\) components, which satisfies \eqref{eq:neutralization-head} by construction.
  Thus the rule
  \[
    A_\eta^u(\vv{s})
    \gets
    (B_1)_{\eta_1}^{u_1}(\vv*{x}{1}'), \dotsc,
    (B_n)_{\eta_n}^{u_n}(\vv*{x}{n}')
  \]
  belongs to \(P'_\pi\).
  Put \(\vv{z} = \sigma'(\vv{s})\).
  Since \(\sigma'(\vv*{x}{i}') = \vv*{z}{i}\) for \(1 \leq i \leq n\),
  applying this rule with \(\sigma'\) to the derivations of the judgments
  \(\vdash_{G'} (B_i)_{\eta_i}^{u_i}(\vv*{z}{i})\) for \(1 \leq i \leq n\)
  yields the judgment \(\vdash_{G'} A_\eta^u(\vv{z})\).
  Using \eqref{eq:inst-homomorphism} and~\eqref{eq:neutralization-head}, we also obtain \(\inst_\eta^u(\vv{z})
  = \inst_\eta^u(\sigma'(\vv{s}))
  = \sigma'(\inst_\eta^u(\vv{s}))
  = (\sigma' \circ \tau)(\vv{t})
  = \sigma(\vv{t})
  = \vv{w}\).
  Since \(\Reppsi(\bm{v}_S) = \{\varepsilon\}\) and \(\varepsilon\) is the only instantiation layout for \((S, \varepsilon)\), applying the preceding statement with \(A = S\) proves \(L(G) \subseteq L(G')\).

  \proofheading{Neutrality}
  Let \(A_\eta^u \in N'\) and suppose that \(\vdash_{G'} A_\eta^u(\vv{z})\).
  By \eqref{eq:neutralization-to-original}, Lemma~\ref{lem:nonterminal-value}, and \eqref{eq:inst-general},
  \[
    \psi(\concat(\vv{z}))
    =
    \psi(\concat(\inst_\eta^u(\vv{z}))) - \psi(u)
    =
    \bm{v}_A - \bm{v}_A
    =
    \bm{0}.
  \]
  Thus every nonterminal of \(G'\) is neutral.
\end{proof}

\section{Proof of the Simulation Theorem}\label{app:simulation}

We use the join and split operations defined in Definition~\ref{def:join-split}.

For \(1 \leq s \leq r\) and a tuple \(\vv{u} = (u_1, \dotsc, u_s) \in (\Sigma^*)^s\), put
\(\pad_r(\vv{u}) = (u_1, \dotsc, u_s, \underbrace{\varepsilon, \dotsc, \varepsilon}_{r - s}) \in (\Sigma^*)^r\).
Since \(\pad_r(\vv{u})\) is obtained from \(\vv{u}\) by inserting empty components,
\(\vv{u} \sqsubseteq \pad_r(\vv{u}) \sqsubseteq \vv{u}\).

\begin{proof}[Proof of Theorem~\ref{thm:simulation}]
  The inclusion \(L(\Gpsi{r}) \subseteq \Wpsi\) follows from Lemma~\ref{lem:gpsi-properties}.
  We prove by induction on \(G\)-derivations that, for every nonterminal \(A\) of arity \(s\) and every
  \(\vv{w} \in (\Sigma^*)^s\),
  \[
    \vdash_G A(\vv{w})
    \implies
    \vdash_{\Gpsi{r}} I(\pad_r(\vv{w})).
  \]
  For non-branching \(G\), we simultaneously prove
  \(\vdash_G A(\vv{w}) \implies \vdash_{\Gnbpsi{r}} I(\pad_r(\vv{w}))\).

  Suppose that \(\vdash_G A(\vv{w})\), and choose a \(G\)-derivation of this judgment.
  Let \(\pi\) be the last rule used in this derivation, and write it as
  \[
    \pi\colon
    A(t_1, \dotsc, t_s)
    \gets
    B_1(x_{1, 1}, \dotsc, x_{1, r_1}), \dotsc,
    B_n(x_{n, 1}, \dotsc, x_{n, r_n}),
  \]
  where \(0 \leq n \leq 2\).
  Put \(\vv{t} = (t_1, \dotsc, t_s)\).
  By the definition of derivability, there exists a monoid homomorphism \(\sigma\) such that \(\sigma(a) = a\) for \(a \in \Sigma\),
  \(\vdash_G B_i(\sigma(\vv*{x}{i}))\) for \(1 \leq i \leq n\), and \(\vv{w} = \sigma(\vv{t})\).
  Write \(z = \tc(\pi) = a_1 \dotsm a_l\).
  By Lemma~\ref{lem:neutral-contribution}, \(z \in \Wpsi\), and \(l \leq c\).
  Using the separator \(\separator\) as a fresh letter,
  put
  \begin{align*}
    p
    &=
    \proj_{\Sigma \cup \{\separator\}}(\join_{\separator^{s - 1}}(\vv{t}))
    =
    p_1 \dotsm p_{l + s - 1} \in (\Sigma \cup \{\separator\})^{l + s - 1},\\
    \vv{t}'
    &=
    \Split_{\Sigma \cup \{\separator\}}(\join_{\separator^{s - 1}}(\vv{t}))
    =
    (t_1', \dotsc, t_{l + s}')
    \in
    (\{x_{i, j} \mid 1 \leq i \leq n, 1 \leq j \leq r_i\}^*)^{l + s}.
  \end{align*}
  Then \(\proj_\Sigma(p) = z\).
  The tuple \(\vv{t}'\) has \(l + s \leq c + m = r\) components.

  We first derive \(\pad_r(\sigma(\vv{t}'))\).
  If \(n = 0\), \(\pad_r(\sigma(\vv{t}'))\) is derived by the nullary rule.
  If \(n = 1\), since \(\pi\) is non-deleting and non-permuting, \(\vv{t}' \sqsubseteq \vv*{x}{1}\), and hence
  \[
    \pad_r(\sigma(\vv{t}'))
    \sqsubseteq
    \sigma(\vv{t}')
    \sqsubseteq
    \sigma(\vv*{x}{1})
    \sqsubseteq
    \pad_r(\sigma(\vv*{x}{1})).
  \]
  The induction hypothesis shows that \(\pad_r(\sigma(\vv*{x}{1}))\) is derivable from \(I\).
  The displayed relation and the definition of the unary regrouping rules show that
  \(\pad_r(\sigma(\vv{t}'))\) is also derivable from \(I\).

  Suppose that \(n = 2\).
  By the well-nestedness of \(\pi\), after possibly swapping the two body atoms, there exists \(0 \leq j \leq r_1\) such that
  \begin{equation}\label{eq:simulation-binary-order}
    \vv{t}'
    \sqsubseteq
    (x_{1, 1}, \dotsc, x_{1, j},
    x_{2, 1}, \dotsc, x_{2, r_2},
    x_{1, j + 1}, \dotsc, x_{1, r_1}).
  \end{equation}
  Here the sequence \(x_{1, 1}, \dotsc, x_{1, j}\) is omitted if \(j = 0\), and
  the sequence \(x_{1, j + 1}, \dotsc, x_{1, r_1}\) is omitted if \(j = r_1\).
  The induction hypothesis yields derivations of
  \(\pad_r(\sigma(\vv*{x}{1}))\) and \(\pad_r(\sigma(\vv*{x}{2}))\).
  Let \(\vv{b}\) be the \(2r\)-tuple obtained by inserting \(\pad_r(\sigma(\vv*{x}{2}))\) after the first \(j\) components of
  \(\pad_r(\sigma(\vv*{x}{1}))\).
  The tuple \(\vv{b}\) is obtained from the componentwise \(\sigma\)-image of the tuple on the right-hand side of
  \eqref{eq:simulation-binary-order} by inserting empty components.
  Hence \(\pad_r(\sigma(\vv{t}')) \sqsubseteq \vv{b}\), so a binary rule of \(\Gpsi{r}\) derives
  \(\pad_r(\sigma(\vv{t}'))\) from \(\pad_r(\sigma(\vv*{x}{1}))\) and \(\pad_r(\sigma(\vv*{x}{2}))\).

  For \(1 \leq j \leq l\), let \(\xi(j)\) be the unique index such that \(p_{\xi(j)} = a_j\) and
  \(\len{\proj_\Sigma(p_1 \dotsm p_{\xi(j)})} = j\).
  For \(1 \leq i < l + s\), put \(t_i'' = \sigma(t_i')\proj_\Sigma(p_i)\), put
  \(t_{l + s}'' = \sigma(t_{l + s}')\), and let \(\vv{t}'' = (t_1'', \dotsc, t_{l + s}'')\).
  By Lemma~\ref{lem:mpsi-decomposition}, choose an \(M_\psi\)-decomposition
  \([l] = J_1 \sqcup \dotsb \sqcup J_k\) of \(z\).
  For each \(1 \leq i \leq k\), one letter insertion rule appends \(a_j\) to the \(\xi(j)\)-th component for every \(j \in J_i\).
  Since \(j < j'\) implies \(\xi(j) < \xi(j')\), the terminal contribution of this rule is \(z[J_i] \in M_\psi\).
  Thus \(\pad_r(\vv{t}'')\) is derivable from \(I\).

  The occurrences of \(\separator\) in \(p\) correspond exactly to the component boundaries of \(\vv{t}\), so
  \(\vv{w} = \sigma(\vv{t}) \sqsubseteq \vv{t}''\), and hence
  \[
    \pad_r(\vv{w})
    \sqsubseteq
    \vv{w}
    \sqsubseteq
    \vv{t}''
    \sqsubseteq
    \pad_r(\vv{t}'').
  \]
  The displayed relation and the definition of the unary regrouping rules show that
  \(\pad_r(\vv{w})\) is derivable from \(I\), completing the induction.

  In particular, if \(\vdash_G S(w)\), then \(I(\pad_r((w)))\) is derivable.
  Since \(\concat(\pad_r((w))) = w\), the start rule derives \(w\).
  Hence \(L(G) \subseteq L(\Gpsi{r})\).
  If \(G\) is non-branching, then \(n \leq 1\) at every step of the induction.
  Thus the binary case does not occur, and the induction proves \(L(G) \subseteq L(\Gnbpsi{r})\).
\end{proof}

\section{Explicit counterexample words from Bishop--Elder--Evetts--Gallot--Levine}\label{app:beegl-witnesses}

We give an explicit choice of the counterexample word in Proposition~\ref{prop:two-letter-witnesses},
following the construction of
Bishop--Elder--Evetts--Gallot--Levine \cite{bishop2026}.
The proof below reads a derivation backwards.
We first show that unary regrouping rules can be replaced by left and right merge rules, whose reverses split one component into two adjacent components.

For \(r \geq 1\) and \(1 \leq i < r\), call the unary regrouping rules
\[
  I(x_1, \dotsc, x_{i - 1}, x_i x_{i + 1}, \varepsilon, x_{i + 2}, \dotsc, x_r)
  \gets
  I(x_1, \dotsc, x_r)
\]
and
\[
  I(x_1, \dotsc, x_{i - 1}, \varepsilon, x_i x_{i + 1}, x_{i + 2}, \dotsc, x_r)
  \gets
  I(x_1, \dotsc, x_r)
\]
the \emph{left merge rule} and the \emph{right merge rule} at position \(i\), respectively.
These are unary regrouping rules, not additional rules of the grammars.

\begin{lemma}\label{lem:regroup-by-merges}
  Let \(r \geq 1\).
  Let \(\vv{v} = (v_1, \dotsc, v_r)\) and \(\vv{w} = (w_1, \dotsc, w_r)\) be tuples of words.
  If \(\vv{v} \sqsubseteq \vv{w}\), then \(\vv{v}\) can be obtained from \(\vv{w}\) by a (possibly empty) sequence of applications of left and right merge rules.
\end{lemma}

\begin{proof}
  Choose indices \(0 = j(0) \leq j(1) \leq \dotsb \leq j(r) = r\) such that
  \[
    (v_1, \dotsc, v_r)
    =
    (w_{(j(0), j(1)]}, \dotsc, w_{(j(r - 1), j(r)]}).
  \]
  We proceed by induction on \(D = \sum_{i = 0}^r \abs{i - j(i)}\).
  If \(D = 0\), then \(j(i) = i\) for every \(i\), so no rule application is needed.

  Suppose that \(D > 0\) and that \(j(i) > i\) for some \(i\).
  Let \(k\) be the largest such index.
  Then \(j(k) = j(k + 1) = k + 1\).
  Let \(l\) be the smallest index such that \(j(l) = j(k)\).
  We have \(j(l - 1) < j(l) = j(l + 1)\) and \(l < j(l)\).
  Replacing \(j(l)\) by \(j(l) - 1\) preserves the weak inequalities among the \(j(i)\) and decreases \(D\) by one.
  Hence, by the induction hypothesis, the tuple
  \[
    (
      w_{(j(0), j(1)]}, \dotsc,
      w_{(j(l - 1), j(l) - 1]},
      w_{(j(l) - 1, j(l + 1)]},
      \dotsc,
      w_{(j(r - 1), j(r)]}
    )
  \]
  can be obtained from \(\vv{w}\).
  Applying the left merge rule at position \(l\) yields \(\vv{v}\).

  It remains to consider the case in which \(j(i) \leq i\) for every \(i\).
  Since \(D > 0\), let \(k\) be the smallest index such that \(j(k) < k\).
  Then \(j(k - 1) = j(k) = k - 1\).
  Let \(l\) be the largest index such that \(j(l) = j(k)\).
  We have \(j(l - 1) = j(l) < j(l + 1)\) and \(j(l) < l\).
  Replacing \(j(l)\) by \(j(l) + 1\) preserves the weak inequalities among the \(j(i)\) and decreases \(D\) by one.
  Hence, by the induction hypothesis, the tuple
  \[
    (
      w_{(j(0), j(1)]}, \dotsc,
      w_{(j(l - 1), j(l) + 1]},
      w_{(j(l) + 1, j(l + 1)]},
      \dotsc,
      w_{(j(r - 1), j(r)]}
    )
  \]
  can be obtained from \(\vv{w}\).
  Applying the right merge rule at position \(l\) yields \(\vv{v}\).
\end{proof}

\begin{corollary}\label{cor:unary-regrouping-normalization}
  For every \(r \geq 1\), each derivation tree in \(\Gnbpsi{r}\) or \(\Gpsi{r}\) can be replaced by a derivation tree for the same ground atom with the following properties:
  \begin{itemize}
    \item Every unary regrouping rule used is a left or right merge rule.
    \item The ground premise and conclusion of each application of a unary regrouping rule are distinct.
    \item Every other rule application remains unchanged.
  \end{itemize}
\end{corollary}

\begin{proof}
  Consider a ground application of a unary regrouping rule \(I(\vv{t}) \gets I(x_1, \dotsc, x_r)\) determined by a monoid homomorphism \(\sigma\) as in the definition of derivability.
  Since componentwise extensions of monoid homomorphisms preserve \(\sqsubseteq\),
  \(\sigma(\vv{t}) \sqsubseteq (\sigma(x_1), \dotsc, \sigma(x_r))\).
  Lemma~\ref{lem:regroup-by-merges} replaces this application by a possibly empty sequence of applications of left and right merge rules with the same ground premise and conclusion.
  Repeating this replacement for every application of a unary regrouping rule yields the required derivation tree.
  If an application introduced by the replacement has the same ground premise and conclusion, omit it.
  When \(r = 1\), the identity rule is the only unary regrouping rule, so every application of a unary regrouping rule can be omitted.
\end{proof}

Recall that \(\MIX_2 = \psin{2}^{-1}(0)\), where
\(\psin{2}\colon \Sigman{2}^* \to \Z\) is given by
\(\psin{2}(\tta) = 1\) and \(\psin{2}(\ttb) = -1\).
Fix \(r \geq 1\), and put \(k = \max\{2, r\}\) and \(m = 24 k + 7\).
Following \cite{bishop2026}*{Definition~7.5}, define
\[
  \mathcal{T}_0 = \ttb^2,
  \qquad
  \mathcal{T}_n = \tta^{m^n} \mathcal{T}_{n - 1}^{2 m} \tta^{m^n}
  \quad
  (1 \leq n \leq 2 k),
  \qquad
  \mathcal{W}_r = \tta^{m^{2 k}} \mathcal{T}_{2 k} \tta^{m^{2 k}}.
\]
The calculation in \cite{bishop2026}*{Lemma~7.7} shows that
\(\psin{2}(\mathcal{T}_n) = -2 m^n\) for \(0 \leq n \leq 2 k\);
hence \(\mathcal{W}_r \in \MIX_2\).

Assume, for contradiction, that
\(\mathcal{W}_r \in L(\Gnbn{2}{r})\).
Since \(r \leq k\), it follows from Lemma~\ref{lem:gpsi-properties} that
\(\mathcal{W}_r \in L(\Gnbn{2}{k})\).
Let \(\vv{w}\) be the \(k\)-tuple derived from \(I\) immediately before the start rule is applied.
Then \(\concat(\vv{w}) = \mathcal{W}_r\), and
\[
  (\mathcal{W}_r,
  \underbrace{\varepsilon, \dotsc, \varepsilon}_{k - 1})
  \sqsubseteq
  \vv{w}.
\]
Choose an \(I\)-derivation of \(\vv{w}\).
The displayed relation and the definition of the unary regrouping rules extend it to an \(I\)-derivation of
\((\mathcal{W}_r, \varepsilon, \dotsc, \varepsilon)\) in \(\Gnbn{2}{k}\).
Apply Corollary~\ref{cor:unary-regrouping-normalization} to the resulting derivation.

Read the resulting sequence of \(I\)-tuples backwards, from \((\mathcal{W}_r, \varepsilon, \dotsc, \varepsilon)\) to \((\varepsilon, \dotsc, \varepsilon)\).
Reversing an application of a left or right merge rule splits one component into two adjacent components.
Reversing an application of a letter insertion rule removes one \(\tta\) and one \(\ttb\) from the corresponding ends of components.
Replace each such step by two successive steps, first removing \(\ttb\) and then removing \(\tta\).
Every tuple \(\vv{v}\) belonging to the normalized \(\Gnbn{2}{k}\)-derivation satisfies \(\psin{2}(\concat(\vv{v})) = 0\) by Lemma~\ref{lem:gpsi-properties}.
Every intermediate tuple \(\vv{v}\) introduced by splitting a reversed letter insertion step satisfies \(\psin{2}(\concat(\vv{v})) = 1\).
Beginning with \((\mathcal{W}_r, \varepsilon, \dotsc, \varepsilon)\), each step either splits a component or deletes a letter from an end of a component.
It follows by induction that, for every tuple \(\vv{v} = (v_1, \dotsc, v_k)\) in the resulting sequence, there exist words \(u_0, \dotsc, u_k\) such that
\(\mathcal{W}_r = u_0 v_1 u_1 \dotsm v_k u_k\).
Thus every tuple in the sequence is a
\(\mathcal{W}_r\)-sentential tuple with \(C = 1\)
in the notation of \cite{bishop2026}*{Definition~7.23}.

In the terminology of Bishop--Elder--Evetts--Gallot--Levine
\cite{bishop2026}*{Definitions~7.13 and~7.20},
the tuple \((\mathcal{W}_r, \varepsilon, \dotsc, \varepsilon)\) has a \(2\)-decomposition
\cite{bishop2026}*{Remark~7.21}, and
repeated applications of \cite{bishop2026}*{Proposition~7.31} show that
every subsequent tuple in the resulting sequence has a decomposition.
However, the final tuple \((\varepsilon, \dotsc, \varepsilon)\) has no such decomposition,
because none of its components is \(n\)-heavy for any \(2 \leq n \leq 2 k\).
Therefore \(\mathcal{W}_r \notin L(\Gnbn{2}{r})\).

\begin{proposition}\label{prop:binary-rule-elimination-three-letters}
  Put
  \(N = \len{\mathcal{W}_2}_{\tta} = \len{\mathcal{W}_2}_{\ttb}\).
  Then \(\mathcal{W}_2 \ttc^N \in L(\Gn{3}{2}) \setminus L(\Gnbn{3}{2})\).
  In particular, the three-letter analogue of
  Lemma~\ref{lem:binary-rule-elimination} is false.
\end{proposition}

\begin{proof}
  We first prove by induction on \(n\) that \((w, \ttc^n) \in \mathcal{A}_2\)
  for every \(w \in \Sigman{2}^*\) such that
  \(\len{w}_{\tta} = \len{w}_{\ttb} = n\).
  For \(n = 0\), this is the first defining condition for
  \(\mathcal{A}_2\).
  Suppose that \(n > 0\).

  If the first and last letters of \(w\) are different, then,
  after exchanging \(\tta\) and \(\ttb\) if necessary, write
  \(w = \tta u \ttb\).
  By the induction hypothesis,
  \((u, \ttc^{n - 1}) \in \mathcal{A}_2\).
  The second defining condition for \(\mathcal{A}_2\) then implies
  \((w, \ttc^n) \in \mathcal{A}_2\).

  Suppose that the first and last letters of \(w\) are equal.
  After exchanging \(\tta\) and \(\ttb\) if necessary, write
  \(w = \tta u \tta\).
  Since \(\psin{2}(\tta) = 1\) and \(\psin{2}(\tta u) = -1\), and each successive letter changes the \(\psin{2}\)-value by \(1\) or \(-1\),
  there is a factorization \(u = v_1 v_2\) such that
  \(\psin{2}(\tta v_1) = 0\).
  Since \(\psin{2}(w) = 0\), we also have
  \(\psin{2}(v_2 \tta) = 0\).
  Put \(j = \len{\tta v_1}_{\tta}\).
  By the induction hypothesis,
  \((\tta v_1, \ttc^j), (v_2 \tta, \ttc^{n - j}) \in \mathcal{A}_2\).
  The third defining condition for \(\mathcal{A}_2\), with the
  second tuple inserted after the first component, implies that
  \((\tta v_1 v_2 \tta, \ttc^{n - j}\ttc^j) = (w, \ttc^n) \in \mathcal{A}_2\).

  Taking \(n = N\) and \(w = \mathcal{W}_2\) in the preceding conclusion and applying Proposition~\ref{prop:g3-tuples},
  we obtain \(\mathcal{W}_2\ttc^N \in L(\Gn{3}{2})\).

  Assume, for contradiction, that
  \(\mathcal{W}_2\ttc^N \in L(\Gnbn{3}{2})\), and let
  \(h = \proj_{\Sigman{2}}\colon \Sigman{3}^* \to \Sigman{2}^*\).
  Apply \(h\) componentwise to every tuple in a
  \(\Gnbn{3}{2}\)-derivation of \(\mathcal{W}_2\ttc^N\).
  The nullary and start rules commute with \(h\).
  Since componentwise extensions of monoid homomorphisms preserve \(\sqsubseteq\), an application of a unary regrouping rule projects to an application of a unary regrouping rule.
  Each letter insertion rule projects to a letter insertion rule
  of \(\Gnbn{2}{2}\), since \(h\) erases \(\ttc\).
  The resulting derivation shows that
  \(\mathcal{W}_2 \in L(\Gnbn{2}{2})\), contradicting the
  construction of \(\mathcal{W}_2\).
\end{proof}

\section{Separations of language classes and related open problems}
\label{app:language-class-landscape}

This appendix records the separation results underlying
Figure~\ref{fig:mcfl-inclusions} and two related open problems.

\begin{lemma}\label{lem:language-class-separations}
  The following statements hold.
  \begin{enumerate}[leftmargin=*]
    \item\label{item:language-class-hierarchies}
      For every \(m \geq 1\), \(\mMCFL{m} \subsetneq \mMCFL{(m + 1)}\) and \(\mMCFLwn{m} \subsetneq \mMCFLwn{(m + 1)}\).

    \item\label{item:fixed-dimension-separations}
      For every \(m \geq 2\),
      \begin{itemize}[leftmargin=*]
        \item \(\mMCFLwn{m} \subsetneq \mMCFL{m} \cap \mMCFLwn{(m + 1)}\).
        \item \(\mMCFL{m}\) and \(\mMCFLwn{(m + 1)}\) are incomparable.
        \item \(\mMCFL{m} \cup \mMCFLwn{(m + 1)} \subsetneq \mMCFL{(m + 1)}\).
      \end{itemize}

    \item\label{item:unbounded-dimension-separations}
      We have:
      \begin{itemize}[leftmargin=*]
        \item \(\MCFLwn \subsetneq \MCFL\) and \(\MCFLwn \subsetneq \IL\).
          In fact, \(\mMCFL{2} \nsubseteq \MCFLwn\).
        \item \(\MCFL\) and \(\IL\) are incomparable.
          More precisely, \(\mMCFL{3} \nsubseteq \IL\) and \(\IL \nsubseteq \MCFL\).
      \end{itemize}
  \end{enumerate}
\end{lemma}

\begin{proof}
  For \ref{item:language-class-hierarchies}, let
  \[
    L_m
    =
    \{a_1^n a_2^n \dotsm a_{2 m + 1}^n \mid n \geq 1\}.
  \]
  Seki--Matsumura--Fujii--Kasami proved that
  \(L_m \notin \mMCFL{m}\)
  \cite{seki1991}*{Lemma~3.3}.
  The language \(L_m\) is generated by the following variant of the grammar in their Example~2.1(3):
  \begin{align*}
    A(a_1 a_2, \dotsc, a_{2 m - 1} a_{2 m}, a_{2 m + 1})
    &\gets,\\
    A(a_1 x_1 a_2, \dotsc, a_{2 m - 1} x_m a_{2 m},
      a_{2 m + 1} x_{m + 1})
    &\gets A(x_1, \dotsc, x_{m + 1}),\\
    S(x_1 \dotsm x_{m + 1})
    &\gets A(x_1, \dotsc, x_{m + 1}).
  \end{align*}
  Here, \(S\) and \(A\) are nonterminals of arity \(1\) and \(m + 1\), respectively.
  This grammar is a well-nested MCFG of dimension \(m + 1\), so
  \(L_m \in \mMCFLwn{(m + 1)}\).
  Thus both inclusions in \ref{item:language-class-hierarchies} are strict.

  For \ref{item:fixed-dimension-separations}, let
  \[
    \mathrm{RESP}_m
    =
    \{a_1^i a_2^i b_1^j b_2^j \dotsm
    a_{2 m - 1}^i a_{2 m}^i b_{2 m - 1}^j b_{2 m}^j
    \mid i, j \geq 0\},
  \]
  following the notation of Kanazawa and Salvati
  \cite{kanazawa-salvati2010copying}*{Section~4}.
  Seki and Kato proved that
  \[
    \mathrm{RESP}_m
    \in
    \mMCFL{m} \setminus \mMCFLwn{m}
  \]
  for every \(m \geq 2\) \cite{seki-kato2008}*{Lemma~13}.
  We next show that \(\mathrm{RESP}_m \in \mMCFLwn{(m + 1)}\).
  Let \(H_m\) be the MCFG with nonterminals \(S\) of arity \(1\) and \(A, B\) of arity \(m + 1\), and with the following rules:
  \[
  \begin{aligned}
    B(\varepsilon, \dotsc, \varepsilon)
    &\gets,\\
    B(x_1, b_1 x_2 b_2, \dotsc, b_{2m - 1} x_{m + 1} b_{2m})
    &\gets B(x_1, \dotsc, x_{m + 1}),\\
    A(x_1, \dotsc, x_{m + 1})
    &\gets B(x_1, \dotsc, x_{m + 1}),\\
    A(a_1 x_1, a_2 x_2 a_3, \dotsc, a_{2m - 2} x_m a_{2m - 1}, a_{2m} x_{m + 1})
    &\gets A(x_1, \dotsc, x_{m + 1}),\\
    S(x_1 \dotsm x_{m + 1})
    &\gets A(x_1, \dotsc, x_{m + 1}).
  \end{aligned}
  \]
  The grammar \(H_m\) is a well-nested MCFG of dimension \(m + 1\) satisfying
  \(L(H_m) = \mathrm{RESP}_m\).
  Together with the result of Seki and Kato, this proves the first assertion in
  \ref{item:fixed-dimension-separations}.

  Fix \(L_0 \in \CFL \setminus \mathrm{EDT0L}\), let \(\#\) be a fresh letter, and set
  \[
    D = \{w\#w \mid w \in L_0\}.
  \]
  Here, \(\mathrm{EDT0L}\) stands for ``Extended, Deterministic, Table, \(0\)-interaction, Lindenmayer''
  \cites{freden-martinsen2023,bishop2026}.
  For example, \(L_0\) may be taken to be a one-sided Dyck language over one pair of parentheses
  \cite{rozoy1987dyck}*{Section~5 and Theorem~6.3.2}.
  Kanazawa and Salvati
  \cite{kanazawa-salvati2010copying}*{Theorem~6(i) and Corollary~9}
  showed that
  \[
    D \in \mMCFL{2} \setminus \MCFLwn.
  \]
  For \(m \geq 2\), the languages \(L_m\) and \(D\) show that \(\mMCFL{m}\) and \(\mMCFLwn{(m + 1)}\) are incomparable.

  Fix \(m \geq 2\).
  Let \(\Sigma\) be an alphabet such that
  \(L_m \cup \mathrm{RESP}_{m + 1} \subseteq \Sigma^*\), and let
  \(\mathtt{0}\) and \(\mathtt{1}\) be fresh letters such that
  \(\mathtt{0}, \mathtt{1} \notin \Sigma\).
  Put
  \(K_m = \mathtt{0} L_m \cup \mathtt{1} \mathrm{RESP}_{m + 1}\).
  Both \(L_m\) and \(\mathrm{RESP}_{m + 1}\) belong to
  \(\mMCFL{(m + 1)}\).
  Since \(\mMCFL{(m + 1)}\) is closed under concatenation and union
  \cite{seki1991}*{Theorem~3.9(2)}, we have \(K_m \in \mMCFL{(m + 1)}\).
  The languages \(\mathtt{0}\Sigma^*\) and \(\mathtt{1}\Sigma^*\) are regular,
  and we have
  \(L_m = \proj_\Sigma(K_m \cap \mathtt{0}\Sigma^*)\) and
  \(\mathrm{RESP}_{m + 1}
    = \proj_\Sigma(K_m \cap \mathtt{1}\Sigma^*)\).
  The class \(\mMCFL{m}\) is closed under intersection with regular languages and under substitutions, and hence under homomorphisms
  \cite{seki1991}*{Theorem~3.9(1) and (3)}.
  If \(K_m \in \mMCFL{m}\), these closure properties would imply \(L_m \in \mMCFL{m}\), a contradiction.
  Thus \(K_m \notin \mMCFL{m}\).
  Intersection with a regular language and homomorphisms are rational transductions.
  The class \(\mMCFLwn{(m + 1)}\) is closed under rational transductions
  \cite{seki-kato2008}*{Theorem~15}.
  If \(K_m \in \mMCFLwn{(m + 1)}\), this closure property would imply
  \(\mathrm{RESP}_{m + 1} \in \mMCFLwn{(m + 1)}\), a contradiction.
  Thus \(K_m \notin \mMCFLwn{(m + 1)}\).
  This proves the last assertion in \ref{item:fixed-dimension-separations}.

  For \ref{item:unbounded-dimension-separations}, well-nested multiple context-free languages are indexed
  \cite{kanazawa-salvati2010copying}*{Section~1}.
  The language \(D\) above shows that \(\mMCFL{2} \nsubseteq \MCFLwn\) and hence that \(\MCFLwn \subsetneq \MCFL\).
  The language
  \[
    \{w\#w\#w \mid w \in L_0\}
  \]
  belongs to \(\mMCFL{3}\) \cite{kanazawa-salvati2010copying}*{Section~1}.
  If it belonged to \(\IL\), the triple-copying theorem of Engelfriet and Skyum would imply that \(L_0 \in \mathrm{EDT0L}\)
  \cite{engelfriet-skyum1976copying}*{Theorem~2}.
  Thus \(\mMCFL{3} \nsubseteq \IL\).
  Finally, the language
  \[
    E = \{\tta^{2^n} \mid n \geq 1\}
  \]
  is indexed but not semilinear
  \cite{joshi-vijay-shanker-weir1991}*{Section~4}.
  Since every MCFL is semilinear \cite{seki1991}*{Theorem~3.7}, we have \(E \in \IL \setminus \MCFL\).
  It follows that \(\MCFLwn \subsetneq \IL\) and that \(\MCFL\) and \(\IL\) are incomparable.
\end{proof}

\begin{openproblem}\label{prob:remaining-class-relations}
  To our knowledge, the following questions remain open.
  \begin{enumerate}[leftmargin=*]
    \item\label{item:open-intersection} Is \(\MCFLwn = \MCFL \cap \IL\)?

    \item\label{item:open-dimension-two-separation} Is \(\mMCFL{2} \nsubseteq \IL\)?
  \end{enumerate}
\end{openproblem}

The classes \(\MCFL\) and \(\IL\) are closed under rational transductions, and every language in \(\MCFL\) is semilinear.
Hence \(\MCFL \cap \IL\) is a semilinear \emph{rational cone}.
Salvati suggested that every semilinear rational cone contained in \(\IL\) is contained in \(\MCFLwn\)
\cite{salvati2015}*{Section~5}.
Thus Salvati's suggestion implies an affirmative answer to Open Problem~\ref{prob:remaining-class-relations}\ref{item:open-intersection}.

By Theorem~\ref{thm:salvati-dimension-two}, \(\MIX \in \mMCFL{2}\).
Thus Marsh's conjecture would answer
Open Problem~\ref{prob:remaining-class-relations}\ref{item:open-dimension-two-separation}
affirmatively, with \(\MIX\) as a witness
\cite{marsh1985conjectures}.
Kanazawa and Salvati conjectured the following double-copying theorem
for OI macro languages, or equivalently for indexed languages
\cite{kanazawa-salvati2010copying}*{Introduction and Section~6}:
\[
  \{w\#w \mid w \in L_0\} \in \IL
  \implies
  L_0 \in \mathrm{EDT0L}.
\]
Since \(L_0 \notin \mathrm{EDT0L}\), this conjecture would imply
\(D \notin \IL\).
Together with \(D \in \mMCFL{2}\), this would answer
Open Problem~\ref{prob:remaining-class-relations}%
\ref{item:open-dimension-two-separation}
affirmatively, with \(D\) as a witness.

An affirmative answer to
Open Problem~\ref{prob:remaining-class-relations}\ref{item:open-intersection}
would also imply an affirmative answer to
Open Problem~\ref{prob:remaining-class-relations}\ref{item:open-dimension-two-separation}.
Indeed, the language \(D\) above belongs to \(\mMCFL{2} \setminus \MCFLwn\);
if \(\MCFLwn = \MCFL \cap \IL\), then \(D \notin \IL\).

\section*{Declaration on the use of generative AI}

The core mathematical ideas and proof strategies in this work are the author's own.
This includes the proof strategy for the Neutralization Theorem~\ref{thm:neutralization} and the statements of Theorems~\ref{thm:theorem-a} and~\ref{thm:theorem-b}.
Under the author's direction, multiple OpenAI GPT models, including GPT-5.5, GPT-5.5 Pro, GPT-5.6 Sol, and GPT-6 Astra, were used to work out proof details, prepare drafts of the manuscript based on those ideas and strategies, search the literature, and revise the manuscript.
The manuscript was developed through many rounds in which the author reviewed the resulting text, specified corrections and desired changes, and instructed the tools to revise it.
Because the text was repeatedly generated and revised in this way, there is no clear sentence-by-sentence division between AI-generated and author-written text.
The author verified every mathematical statement and proof step, checked every citation against the cited source, and reviewed every sentence of the final manuscript.
The author takes full responsibility for the content and final form of the manuscript.

\renewcommand{\sectionname}{}

\end{document}